\documentclass{article} 

\usepackage{balance}
\usepackage[utf8]{inputenc}
\usepackage{graphicx}
\usepackage{amsmath}
\usepackage{amsthm}
\usepackage{booktabs}
\usepackage{algorithm}
\usepackage{algorithmic}
\usepackage{authblk}
\usepackage[numbers]{natbib}
\usepackage{multicol}

\usepackage[pdfencoding=auto]{hyperref}
\usepackage{bookmark}%

\usepackage{geometry}

\usepackage{amsmath}
\usepackage{amsthm}
\usepackage{amssymb}
\usepackage{pifont}
\usepackage{dsfont}
\usepackage{tikz}
\usetikzlibrary{
  arrows.meta,
  positioning,
}
\usepackage[T1]{fontenc}

\tikzset{
  node/.style={
    circle,draw,thick,
    inner sep=0pt,
    minimum size=1.5em,
    node distance=1ex and 2em,
  },
  othernode/.style={
    circle,draw,thick,
    inner sep=0pt,
    minimum size=1.0em,
    node distance=1ex and 2em,
  },
  conn/.style={
    -{Straight Barb[angle=60:2pt 3]},
    thick,
  },
}

\newtheorem{example}{Example}
\newtheorem{theorem}{Theorem}
\newtheorem{lemma}{Lemma}
\newtheorem{definition}{Definition}
\newtheorem{axiom}{Axiom}

\newcommand{\cmark}{\ding{51}}%
\newcommand{\xmark}{\ding{55}}%

\title{Axiomatic Analysis of Medial Centrality Measures}

\author{Wiktoria Ko{\'s}ny\thanks{wiktoria.kosny@gmail.com} }
\author{Oskar Skibski\thanks{oskar.skibski@mimuw.edu.pl} }

\affil{University of Warsaw, Poland}

\begin{document}

\maketitle 

\begin{abstract}
We perform the first axiomatic analysis of medial centrality measures. These measures, also called betweenness-like centralities, assess the role of a node in connecting others in the network. We focus on a setting with one target node and several source nodes. We consider three classic medial centrality measures adapted to this setting: Betweenness Centrality, Stress Centrality and Random Walk Betweenness Centrality. While Betweenness and Stress Centralities assume that the information in the network follows shortest paths, Random Walk Betweenness Centrality assumes it moves randomly along the edges. We develop the first axiomatic characterizations of all three measures. Our analysis shows that Random Walk Betweenness, while conceptually different, shares several common properties with classic Betweenness and Stress Centralities.
\end{abstract}

\section{Introduction}
Identifying key elements in complex interconnected systems is one of the fundamental challenges of network science.
More than a hundred such methods, called \emph{centrality measures}, have been proposed \cite{Brandes:Erlebach:2005,Jackson:2008}.
In this paper, we focus on \emph{medial centralities}~\cite{Borgatti:Everett:2006} which, along with the (radial) distance-based centralities and feedback centralities, constitute one of the major classes of centrality measures.
They have been used in many applications, including the analysis of protein interaction networks~\cite{Joy:etal:2005}, identifying gatekeepers in the social or covert networks (i.e., nodes with the ability to control information flow) \cite{Coffman:etal:2004} or assessing monitoring capabilities in computer networks~\cite{Dolev:etal:2010}.

Medial centralities, also called \emph{betweenness-like centralities}, assess a node by the role it plays as an intermediary in a network.
Their canonical example is Betweenness Centrality---arguably, one of the three most important centrality measures.
Betweenness Centrality and other classic medial centralities first assess the role in connecting two nodes, a source and a target, and then sum these values over all pairs of nodes.
In the case of Betweenness and Stress Centralities the number of shortest paths going through a node is counted.
In turn, Random Walk Betweenness Centrality focuses on the visits of the random walk.

The choice of a suitable centrality measure for a specific goal out of multiple similar concepts is often hard.
On one hand, real-world situations are often hard to translate to graph notions.
In particular, in the citation network it is not clear whether we should focus on the shortest paths or the random walk.
On the other hand, intuitive interpretations of centralities are often misleading.
To give an example, Betweenness Centrality is often highly correlated with the degree, as high degree imposes a high number of shortest paths a node is on.
Resulting from that, centrality measures are often chosen based on their popularity rather than on their fit for the application.
This misleads the analysis and leads to poorly funded conclusions.

That is why, in recent years, the axiomatic approach has gained popularity in centrality analysis \cite{Boldi:Vigna:2014,Bloch:etal:2019,Skibski:etal:2019:attachment}.
In this approach, simple properties called \emph{axioms} are defined that highlight specific behaviors of a measure.
A carefully designed set of axioms allows one to uniquely characterize a measure.
Such results deepen the understanding of centrality measures and highlight the differences and similarities between them.
Also, if axioms are natural and based on simple graph operations, they allow a practitioner to test whether a specific property is desirable in the application at hand.

There are several papers that apply the axiomatic approach to distance-based centralities \cite{Skibski:2023:closeness,Garg:2009} and even more that concern feedback centralities \cite{Dequiedt:Zenou:2017,Was:Skibski:2021:feedback} (see Related Work for details).
Medial centralities, however, have not been studied using the axiomatic approach.
In particular, no axiomatization of Betweenness Centrality and its variants has been proposed to date.

The main reason for the lack of such results is the complex nature of betweenness-like measures. 
In fact, it is especially hard to find any graph operations that do not change these centralities and such operations are usually the basis of axiomatic characterizations.
Adding even a single edge in a graph may completely change the structure of the set of shortest paths and, as a consequence, values of medial centralities based on shortest paths.
As a result, out of dozens of axioms proposed in the literature, Betweenness Centrality satisfies only few simplest.

To cope with this challenge, in this paper we focus on a setting with one target node and arbitrary many source nodes.
In this way, we concentrate on the key aspect of medial centralities which allows us to better identify similarities and differences between them.
At the same time, the setting is simpler to analyze which allows us to obtain strong axiomatic results which was not possible for the general model so far.

Moreover, centrality analysis in  our setting with one target node has several natural applications.
In the World Wide Web, it can indicate the role of websites in directing users to a specific page.
In the financial network, centralities can identify top intermediaries responsible for transferring money to a specific bank account.
In the communication networks, they can assess the role in controlling the flow of information going to one specific entity.

We consider three classic medial measures adapted to this setting: Betweenness, Stress and Random Walk Betweenness and create the first axiom system that enables us to characterize all of them.

To this end, first, we ask the question: what are the properties satisfied by all three centralities? 
We identify four such properties.
\emph{Locality} states that the importance of a node does not depend on separate parts of the network. 
\emph{Additivity} imposes that the centrality is additive in respect to node weights.
\emph{Node Redirect} states that merging out-twins does not affect centralities of other nodes.
Finally, \emph{Target Proxy} says that if all paths to the target goes through one specific node, then it can be considered a target. 
While Locality and Additivity can be considered general axioms that should be satisfied by all reasonable centrality measures, Node Redirect and Target Proxy capture similarities between considered measures (in particular, they are not satisfied by measures based on flow).

Furthermore, we propose two axioms that are specific for centralities based on shortest-paths.
\emph{Symmetry} concerns the scenario where there is only one source node. The axiom states that in such a case, the importance would not change if the graph is reversed and the source replaced with the target. 
Now, \emph{Direct Link Domination} states that if a node has a link to the target, its other edges may be deleted.
Interestingly, these two axioms are not satisfied by Random Walk Betweenness.

Finally, we propose two borderline axioms that specify the centrality for a trivial graph with only two nodes, a source and a target, and $k$ edges from the source to the target.
Specifically, Atom $1$-$1$ states the centrality of both nodes equal one and Atom $k$-$k$ states they are equal $k$. 
There are two goals of these axioms.
First of all, they highlight the difference in the way medial centralities treat multiple edges between nodes. 
Betweenness Centrality and Random Walk Betweenness ignore them and assign value $1$ to both nodes. 
In turn, Stress Centrality assigns value $k$.
Second of all, they serve as a boundary case to pinpoint specific values of a measure and preclude other measures that only differ by scalar multiplication.

In our main result, we show that Betweenness Centrality is uniquely characterized by four common axioms (Locality, Additivity, Node Redirect, Target Proxy), two shortest-paths specific axioms (Symmetry, Direct Link Domination) and Atom $1$-$1$.
Specifically, if a centrality measures satisfies all seven axioms, then it must be equal to Betweenness Centrality.
On top of that, we show that if we replace Atom $1$-$1$ with Atom $k$-$k$, then we obtain an axiomatization of Stress Centrality. 
Most of these axioms are new, only two out of eight (Locality and Node Redirect) are known axioms adapted to our setting.

Furthermore, we show how this axiomatization can be extended for Random Walk Betweenness Centrality.
To this end, we use five previous axioms (Locality, Additivity, Node Redirect, Target Proxy and Atom $1$-$1$) and add two axioms from a recent axiomatization of PageRank~\cite{Was:Skibski:2020:pagerank}: \emph{Edge Swap} and \emph{Edge Multiplication}.
This is possible, as Random Walk Betweenness Centrality can be considered a border case of PageRank of a modified graph without the damping factor (see Preliminaries for details).

As a result, we obtain a joint axiomatic characterization of three classic medial centralities that highlights their differences as well as similarities.
See Table~\ref{tab:sat} for a summary.

\begin{table}[t]
\centering{
 \begin{tabular}{ p{10.5em} c | c  c} 
 \textbf{Axiom} & \ \ \textbf{$B^t$}\ \  & \ \ \textbf{$S^t$}\ \  & \textbf{$RWB^t$} \\
 \midrule
    Locality & 
    \cmark & 
    \cmark & 
    \cmark \\ 
    Additivity &
    \cmark & 
    \cmark & 
    \cmark \\ 
    Node Redirect &
    \cmark & 
    \cmark & 
    \cmark \\ 
    Target Proxy &
    \cmark &  
    \cmark & 
    \cmark \\
 \midrule
    Symmetry &
    \cmark & 
    \cmark & 
    \xmark \\
    Direct Link Domination &
    \cmark & 
    \cmark & 
    \xmark \\
 \midrule
    Atom $1$-$1$ &
    \cmark & 
    \xmark & 
    \cmark \\ 
    Atom $k$-$k$ &
    \xmark & 
    \cmark & 
    \xmark \\ 
 \midrule
    Edge Swap &
    \xmark & 
    \xmark & 
    \cmark \\
    Edge Multiplication &
    \xmark & 
    \xmark & 
    \cmark \\ 
 \end{tabular}
}
\caption{Axiomatic characterizations of medial centralities.}
\label{tab:sat}
\end{table}

\subsection{Related Work}

The axiomatic approach to centrality measure was initiated by the work of \citet{Sabidussi:1966} and \citet{Nieminen:1973} who characterized which functions are centrality measures. 
The proposed axioms were inspired mostly by distance-based centralities that is why medial centralities violates most of them.

Later on, \citet{Boldi:Vigna:2014} proposed three axioms and checked that out of the standard centrality measures only Harmonic Centrality satisfies all three of them. Betweenness Centrality, the only medial centrality considered, violates all three axioms (it satisfies one axiom only under additional assumptions).

A more popular way to apply the axiomatic approach to centrality analysis is creating an axiomatic characterization of a measure or the whole class of measures. 
Our work contributes to this line of research.
A few papers have considered distance-based centralities~\cite{Skibski:2023:closeness,Garg:2009,Skibski:Sosnowska:2018:distance} and game-theoretic centralities~\cite{Skibski:etal:2018:gtc,Skibski:etal:2019:attachment}.
Most papers, however, focused on feedback centralities: measures in which the importance of a node is defined based on the number and the importance of its neighbors (or direct predecessors in directed graphs).
Specifically, \citet{Brink:Gilles:2000} axiomatized $\beta$-measure, \citet{Altman:Tennenholtz:2005} the Seeley index (a simplified version of PageRank) and \citet{Kitti:2016} eigenvector centrality. 
More recently, \citet{Dequiedt:Zenou:2017} and \citet{Was:Skibski:2018:eigenvector} created a joint axiomatization of eigenvector and Katz centrality for undirected and directed graphs, respectively. 
Also, \citet{Was:Skibski:2020:pagerank} created an axiomatization of PageRank.

Since centrality measures from different classes vary significantly, Betweenness and Stress Centralities satisfy only a few of the simplest axioms from these papers (e.g., they satisfy Anonymity that states isomorphic nodes have equal centralities).
However, when restricted to one target node, they both satisfy Node Redirect~\cite{Was:Skibski:2020:pagerank} which is a meaningful axiom; hence, we use it in our paper.

In turn, Random Walk Betweenness satisfies several axioms proposed in the axiomatization of PageRank~\cite{Was:Skibski:2020:pagerank}.
That is why our last result is related to this axiomatization.
We use 3 axioms proposed in \cite{Was:Skibski:2020:pagerank} and show they can be combined with our other axioms to obtain the axiomatization of Random Walk Betweenness.
On a high level, our proof has a similar structure to the proof of the PageRank axiomatization.
However, since most of the axioms are new and we consider a different setting the main part of the proof is also new.

\section{Preliminaries}
In our work, we consider weighted directed multigraphs with possible self-loops. 
This model generalizes to unweighted and undirected graphs and could be used to represent the World Wide Web, social or financial networks, among others.

A graph is a pair $G = (V, E)$, where $V$ is the set of nodes and $E$ is the multiset of edges, that is, ordered pairs of nodes $(v, u) \in V \times V$. We will denote the number of occurrences of an element $e$ in the multiset $E$ by $m_e(E)$. An edge $(v, u)$ is \emph{outgoing} from the node $v$, which is the \emph{start} of the edge and is \emph{incoming} to the node $u$, which is the \emph{end} of the edge. If $v = u$, the edge is a \emph{self-loop}.

Let $\Gamma^+_v((V, E)) = \{(v, u) \in E\}$ denote the multiset of edges outgoing from the node $v$ (including self-loops) and let its cardinality be the \emph{out-degree} of the node $v$. Let $N^+_v((V, E)) = \{u \in V: (v, u) \in E\}$ denote the set of \emph{direct successors} of $v$. Similarly, let $\Gamma^-_v((V, E)) = \{(u, v) \in E\}$ denote the multiset of edges incoming to the node $v$ (including self-loops) and let its cardinality be the \emph{in-degree} of the node $v$. Let $N^-_v((V, E)) = \{u \in V: (u, v) \in E\}$ denote the set of \emph{direct predecessors} of $v$. Also, let $\Gamma_v(G) = \Gamma^+_v(G) \cup \Gamma^-_v(G)$. 

Two nodes $v$ and $u$ are \emph{out-twins} if they have the same outgoing edges, that is for every $w \in V$ there holds $m_{(v, w)}(E) = m_{(u, w)}(E)$. 

A \emph{path} $p = (e_1, e_2, ..., e_k)$ of \emph{length} $k$, is a sequence of edges of the graph (i.e.,  $e_1,\dots,e_k \in E$) in which every edge starts with a node with which the previous edge ended, that is for every $i \in \{1, ..., k - 1\}$ there are some nodes $v, u, w \in V$ such that $e_i = (v, u)$, $e_{i+1} = (u, w)$. The start of the first edge is the \emph{start of the path} and the end of the last edge is \emph{the end of the path}. A cycle is a path that starts and ends in the same node.

A node $u$ is \emph{reachable} from a node $v$ if there is a path that starts with $v$ and ends with $u$. The \emph{distance} from $v$ to $u$, denoted by $dist_{v, u}(G)$, is the length of the shortest path that starts with $v$ and ends with $u$. The number of shortest paths from $v$ to $u$ in $G$ is denoted by $\sigma_{v,u}(G)$ and the number of shortest paths from $v$ to $u$ in $G$ that contain $w$ is denoted by $\sigma_{v,u}(G, w)$.

\subsection{Our Setting}
In this paper, we treat graphs as the information networks. 
We will assume there is one \emph{target node}, denoted by $t$, which is the destination of all data traveling through the network. 
As mentioned in the introduction, other interpretations include users travelling towards some page in the Web or money transferred to a specific bank account.

We will restrict our attention to graphs in which the target node $t$ is reachable from every node; the set of all such graphs will be denoted by $\mathcal{G}_t$.

To specify which nodes are the sources of information and how much information they send, we will consider node weight functions $b: V \rightarrow \mathds{R}_{\geq 0}$.
The simplest case is when there is a single source: $s \in V$.
To describe such situations, we denote by $\mathds{1}^s$ a node weight function such that $(\mathds{1}^s)(s) = 1$ and $(\mathds{1}^s)(v) = 0$ for $v \in V - \{s\}$.
The multiplication of a weight function $b$ by a constant $x \in \mathds{R}_{\ge 0}$ is defined as $(x \!\cdot\! b)(v) = x \cdot b(v)$ for $v \!\in\! V$.
The addition of two weight functions $b: V \rightarrow \mathds{R}_{\geq 0}$, $b': V' \rightarrow \mathds{R}_{\geq 0}$ is defined as $(b + b')(v) = b(v) + b'(v)$ for every $v \in V \cap V'$, $(b + b')(v) = b(v)$ for every $v \in V - V'$ and $(b + b')(v) = b'(v)$ otherwise.

Let us define several operations on (node-)weighted graphs. The \emph{sum} of two graphs is obtained by summing the corresponding node sets, edge multisets and weight functions.
Formally, for two weighted graphs $(G,b)$, $(G',b')$ with $G = (V,E)$, $G'=(V',E')$ we have:
\[ (G,b) + (G',b') = (V \cup V', E + E', b + b'),\]
where $E+E'$ denotes the sum of multisets $E,E'$.

We will also use a shorthand notation for adding and deleting edges $E'$ from the graph $G = (V, E)$: we define $G + E' = (V, E + E')$ and $G - E' = (V, E - E')$.

For two different nodes $v,u \in V$ we define \emph{merging} and \emph{redirecting}.
\emph{Merging} $v$ into $u$ deletes $v$ from the graph and moves its weight and all outgoing and incoming incident edges to $u$. Formally:
$M_{v \rightarrow u}(G,b) = ((V - \{v\}, E'), b')$, where $E' = E - \Gamma_v(G) + \{ (f_{v \rightarrow u}(w),f_{v \rightarrow u}(w')) : (w,w') \in \Gamma_v(G)\}$ where $f_{v \rightarrow u}(v) = u$ and $f_{v \rightarrow u}(w) = w$ for $w \in V - \{v\}$ and $b'(u) = b(v)+b(u)$ and $b'(w) = b(w)$ for $w \in V - \{v,u\}$. 
Now, \emph{redirecting} $v$ into $u$ deletes outgoing edges of $v$ (except for self-loops) and merges $v$ into $u$:
$R_{v \rightarrow u}(G,b) = M_{v \rightarrow u}((V,E-(\Gamma_v^+(G)-\Gamma_v^-(G))),b)$.

\subsection{Centralities}

A \emph{centrality measure} $F$ is a function that for every node $v$ in a (node-)weighted graph $(G, b)$ assigns a non-negative real value, denoted by $F_v(G, b)$.

In our setting, centrality measures capture how often a node conveys the data to the target node.
Hence, we will consider centrality measures parameterized by the target node $t$ defined on the class of graphs $\mathcal{G}_t$.
We will call them \emph{t-centrality measures}.
We will consider three t-centrality measures which are direct counterparts of the classic measures for the general setting.

The simplest and chronologically the first medial centrality measure was proposed by \citet{Shimbel:1953} under the name \emph{Stress Centrality}.
This measure simply counts the number of shortest paths that go through a specific node.
We adjust Stress Centrality to our setting by fixing the target node of all paths and considering paths from different sources with different weights.

\begin{definition}
For a graph $G = (V,E) \in \mathcal{G}_t$ with weights~$b$, \emph{t-Stress Centrality} of node $v \in V$ is defined as follows:
$$S^t_v(G, b) = \sum_{s \in V} b(s) \cdot \sigma_{s,t}(G, v).$$
\end{definition}

\emph{Betweenness Centrality} \cite{Freeman:1977}, the most widely used medial centrality, can be considered a relative version of Stress Centrality.
In Betweenness Centrality the number of shortest paths from $s$ to $t$ that goes through $v$ is divided by the total number of shortest paths from $s$ to $t$.
We define t-Betweenness Centrality accordingly.

\begin{definition}
For a graph $G = (V,E) \in \mathcal{G}_t$ with weights~$b$, \emph{t-Betweenness Centrality} of node $v \in V$ is defined as follows:
$$B^t_v(G, b) = \sum_{s \in V} b(s) \cdot \frac{\sigma_{s,t}(G, v)}{\sigma_{s,t}(G)}.$$
\end{definition}

The standard Betweenness and Stress Centralities for an unweighted graph are the sum of the t-Betweenness and t-Stress Centralities, respectively, over all target nodes $t \in V$ with unit node weights.

Betweenness and Stress Centralities are based on the underlying assumption that the information travels the network through shortest paths.
Such an assumption makes sense in settings in which the whole structure is known beforehand and the process is optimized.
The opposite approach is proposed in \emph{Random Walk Betweenness Centrality}.

Assume an information packet starts from a source node $s$ and moves randomly through the network.
In each step, it chooses one of the outgoing edges of the node it is in, uniformly at random, and moves along this edge.
To measure the role in transferring the information to one specific target node $t$, node $t$ is treated as an absorbing node in which all packets end their travel (technically, this is achieved by deleting outgoing edges of $t$).
Now, the role in connecting nodes $s$ to $t$ in the random-walk version of Betweenness Centrality is defined as the expected number of times a packet visits a specific node. 

Formally, let us denote by $\mathds{P}(\omega_{G,s}(k) = v)$ the probability that the random walk $\omega_{G,s}$ on graph $G$ that starts in node $s$ after $k$ steps will be at node $v$.
We define t-Random Walk Betweenness Centrality as follows.

\begin{definition}
For a graph $G = (V, E) \in \mathcal{G}_t$ with weights~$b$, \emph{t-Random Walk Betweenness Centrality} of node $v \in V$ is defined as follows:
$$RWB^t_v(G, b) = \sum_{s \in V} \sum_{k=0}^{\infty} b(s) \cdot \mathds{P}(\omega_{G - \Gamma_t^+(G),s}(k) = v).$$
\end{definition}

We note that t-Random Walk Betweenness Centrality is a counterpart of the measure defined by \citet{Bloechl:etal:2011} for directed graphs and not of the measure defined by \citet{Newman:2005} for undirected graphs under the same name.

\begin{figure}[t]
\begin{minipage}[ht]{0.39\linewidth}
    \centering
    \scalebox{0.9}{
    \begin{tikzpicture}
      \node[node, double] (7) 
        {$t$};
      \node[node,above right=1.75em and 0.5em of 7] (5)
        {$v_2$};
      \node[node,left=1.5em of 5] (4)
        {$v_1$};
      \node[node,below right=0.3em and 1.5em of 5] (6)
        {$v_4$};
      \node[node,above=1.75em of 4] (1)
        {$s_1$};
      \node[node,above=1.75em of 5] (2)
        {$s_2$};
      \node[node,above=1.5em of 6] (3)
         {$v_3$};
      \node [above=1.25em of 1] (in1) {};
      \node [above=1.25em of 2] (in2) {};
      \draw[conn] (in1) -- (1);
      \draw[conn] (in2) -- (2);
      \draw[conn] (1) -- (4);
      \draw[conn] (1) -- (5);
      \draw[conn] (1) edge [bend right=15] (3);
      \draw[conn] (2) -- (5);
      \draw[conn] (2) edge [bend left=15] (3);
      \draw[conn] (2) edge [bend right=15] (3);
      \draw[conn] (3) -- (6);
      \draw[conn] (4) edge [bend left=20] (7);
      \draw[conn] (4) -- (7);
      \draw[conn] (4) edge [bend right=20] (7);
      \draw[conn] (5) -- (7);
      \draw[conn] (6) -- (7);
    \end{tikzpicture}
    }
\end{minipage}
\begin{minipage}[ht]{0.59\linewidth}
    \centering{
     \begin{tabular}{ c c c c} 
     \textbf{Node} & \textbf{$B^t$} & \textbf{$S^t$} & \textbf{$RWB^t$} \\
        \hline
        $s_1$ & $1.00$ & $4.00$ & $1.00$ \\ 
        $s_2$ & $1.00$ & $1.00$ & $1.00$ \\
        \hline
        $v_1$ & $0.75$ & $3.00$ & $0.33$ \\ 
        $v_2$ & $1.25$ & $2.00$ & $0.67$ \\ 
        $v_3$ & $0.00$ & $0.00$ & $1.00$ \\ 
        $v_4$ & $0.00$ & $0.00$ & $1.00$ \\ 
        \hline
        $t$ & $2.00$ & $5.00$ & $2.00$ \\ 
     \end{tabular}
    }
\end{minipage}
\caption{An illustration of the t-centralities. 
Nodes $s_1$ and $s_2$ are the only sources, as indicated by the incoming arrows, i.e., $b \!=\! \mathds{1}^{s_1} \!+\! \mathds{1}^{s_2}$.
Node $t$, marked by a double line, is the target node.
The values of t-Betweenness, t-Stress and t-Random Walk Betweenness of all nodes are presented in the table on the right-hand side.
}
\label{fig:t_centralities_example}
\end{figure}

\begin{example} (t-Centralities)
Consider a graph from Figure~\ref{fig:t_centralities_example}.
Let us analyze the role of the intermediary nodes $v_1,v_2,v_3,v_4$ in transmitting data from the source nodes $s_1, s_2$ to the target node $t$ according to the three medial centralities.

There are four shortest paths from $s_1$ to $t$: three paths through $v_1$ and one path through $v_2$.
There is one shortest path from $s_2$ to $t$ and it goes through $v_2$. 
As a result, t-Stress ranks node $v_1$ highest, as it is the only node which is on three of the relevant shortest paths.
In turn, t-Betweenness ranks node $v_2$ highest, as it is on all of the shortest paths from $s_2$ and on $1/4$ of the shortest paths from $s_1$.

Now, consider t-Random Walk Betweenness.
The random walk that starts in the node $s_1$ reaches nodes $v_1, v_2, v_3$ with the probability $1/3$; hence, it also reaches node $v_4$ with probability $1/3$.
The random walk that starts in the node $s_2$ reaches nodes $v_3$ and $v_4$ with the probability $2/3$, node $v_2$ with the probability $1/3$ and cannot reach node $v_1$.
Hence, t-Random Walk Betweenness ranks nodes $v_3$ and $v_4$ highest.
\end{example}

\section{Axiomatization of t-Betweenness Centrality}
We will now present the set of simple properties that we will use to uniquely characterize t-Betweenness Centrality.
First, we will present four common axioms satisfied by t-Betweenness, t-Stress as well as t-Random Walk Betweenness: \emph{Locality}, \emph{Additivity}, \emph{Node Redirect} and \emph{Target Proxy}.
Then, we will introduce two new axioms satisfied only by centralities based on shortest paths, t-Stress and t-Betweenness, namely \emph{Symmetry} and \emph{Direct Link Domination}.
Finally, we introduce a borderline axiom named \emph{Atom $1$-$1$}, satisfied by t-Betweenness and t-Random Walk Betweenness.

\subsection{Common Axioms}
We start with axioms satisfied by all three t-centralities.
For an illustration see Figure~\ref{fig:common_axioms}. 

\begin{axiom}
\label{Locality}
(Locality) For every two graphs $G = (V, E)$, $G' = (V', E') \in \mathcal{G}_t$ with weights $b,b'$ such that $V \cap V' = \{t\}$ and node $w \in V - \{t\}$:
$$F^t_w((G, b) + (G', b')) = F^t_w(G, b)$$
and $F^t_t((G, b) + (G', b')) = F^t_t(G, b) + F^t_t(G', b')$.
\end{axiom}

\begin{axiom}
\label{Additivity}
(Additivity) For every graph $G=(V,E) \in \mathcal{G}_t$ with weights $b, b'$ and node $w \in V$:
$$F^t_w(G, b + b') = F^t_w(G, b) + F^t_w(G, b').$$
\end{axiom}

\begin{axiom}
\label{Node Redirect}
(Node Redirect) For every graph $G=(V, E) \in \mathcal{G}_t$ with weights $b$ and out-twins $v, u \in V - \{t\}$:
$$F^{t}_v(R_{u \rightarrow v}(G, b)) = F^{t}_v(G, b) + F^{t}_{u}(G, b)$$
and $F^{t}_w(R_{u \rightarrow v}(G, b)) \!=\! F^{t}_w(G, b)$ for every $w \!\in\! V - \{v, u\}$.
\end{axiom}

\begin{axiom}
\label{Target Proxy}
(Target Proxy) For every graph $G=(V, E) \in \mathcal{G}_t$ with weights $b$ such that $\Gamma^-_t(G) = \{(v, t)\}$ and $\Gamma^+_v(G) = \{(v, t)\}$, $b(t) = 0$ and node $w \in V - \{t\}$:
$$F^{v}_w(M_{t \rightarrow v}(G, b)) = F^{t}_w(G, b).$$
\end{axiom}

\begin{figure*}[t]
\centering
\begin{tikzpicture}
  \node[node] (v2) 
    [label={[xshift=1.25em, yshift=-1.5em] $2$}] {$v_2$};
  \node[node,below=1.25em of v2, double] (v3)
    [label={[xshift=1.25em, yshift=-1.5em] $3$}] {$t$};
  \node[node,above=1.25em of v2] (s2)
    [label={[xshift=1.25em, yshift=-1.5em] $2$}] {$s_2$};
  \node[node,left=1.25em of s2] (s1)
    [label={[xshift=1.25em, yshift=-1.5em] $1$}] {$s_1$};
  \node[node,below=1.25em of s1] (v1)
    [label={[xshift=1.25em, yshift=-1.5em] $1$}] {$v_1$};
  \node [above=1.25em of s1] (in1) {};
  \node [above left=1.25em and 0.5em of s2] (in2) {};
  \node [above right=1.25em and 0.5em of s2] (in3) {};
  \draw[conn] (in1) -- (s1);
  \draw[conn] (in2) -- (s2);
  \draw[conn] (in3) -- (s2);
  \draw[conn] (s1) -- (v1);
  \draw[conn] (s2) -- (v2);
  \draw[conn] (s2) edge [bend right=25] (v2);
  \draw[conn] (v2) -- (v3);
  \draw[conn] (v3) edge [bend left=25] (v2);
  \draw[conn] (v1) -- (v3);
  \node [below left=3.0em and 1.0em, align=flush center] at (v2) {$G$};
  
  \node [above right=0.5em and 1.0em, align=flush center] (b) at (v2) {};
  \node [above right=0.5em and 7.5em, align=flush center] (e) at (v2) {};
  \draw[conn] (b) -- (e)
    node [midway, below] {\scriptsize Node Redirect};

  \node[node, right=7.5em of v2] (v2) 
    [label={[xshift=1.25em, yshift=-1.5em] $3$}] {$v_2$};
  \node[node,below=1.25em of v2, double] (v3)
    [label={[xshift=1.25em, yshift=-1.5em] $3$}] {$t$};
  \node[node,above=1.25em of v2] (s2)
    [label={[xshift=1.25em, yshift=-1.5em] $2$}] {$s_2$};
  \node[node,left=1.25em of s2] (s1)
    [label={[xshift=1.25em, yshift=-1.5em] $1$}] {$s_1$};
  \node [above=1.25em of s1] (in1) {};
  \node [above left=1.25em and 0.5em of s2] (in2) {};
  \node [above right=1.25em and 0.5em of s2] (in3) {};
  \draw[conn] (in1) -- (s1);
  \draw[conn] (in2) -- (s2);
  \draw[conn] (in3) -- (s2);
  \draw[conn] (s1) -- (v2);
  \draw[conn] (s2) -- (v2);
  \draw[conn] (s2) edge [bend right=25] (v2);
  \draw[conn] (v2) -- (v3);
  \draw[conn] (v3) edge [bend left=25] (v2);
  \node [below left=3.0em and 1.0em, align=flush center] at (v2) {$G_1$};
  
  \node [above right=0.5em and 1.0em, align=flush center] (b) at (v2) {};
  \node [above right=0.5em and 7.0em, align=flush center] (e) at (v2) {};
  \draw[conn] (b) -- (e)
    node [midway, below] {\scriptsize Target Proxy};
  
  \node[node, right=7em of v2, double] (v2) 
    [label={[xshift=1.25em, yshift=-1.5em] $3$}] {$v_2$};
  \node[node,above=1.25em of v2] (s2)
    [label={[xshift=1.25em, yshift=-1.5em] $2$}] {$s_2$};
  \node[node,left=1.25em of s2] (s1)
    [label={[xshift=1.25em, yshift=-1.5em] $1$}] {$s_1$};
  \node [above=1.25em of s1] (in1) {};
  \node [above left=1.25em and 0.5em of s2] (in2) {};
  \node [above right=1.25em and 0.5em of s2] (in3) {};
  \draw[conn] (in1) -- (s1);
  \draw[conn] (in2) -- (s2);
  \draw[conn] (in3) -- (s2);
  \draw[conn] (s1) -- (v2);
  \draw[conn] (s2) -- (v2);
  \draw[conn] (s2) edge [bend right=25] (v2);
  \path[conn] (v2) edge [in=-135,out=-105,loop] node[auto] {} ();
  \path[conn] (v2) edge [in=-75,out=-45,loop] node[auto] {} ();
  \node [below=3.0em, align=flush center] at (v2) {$G_2$};
  
  \node [above right=0.5em and 1.2em, align=flush center] (b) at (v2) {};
  \node [above right=0.5em and 4.5em, align=flush center] (e) at (v2) {};
  \draw[conn] (b) -- (e)
    node [midway, below] {\scriptsize Locality};
  
  \node[node, right=4.5em of v2, double] (v2) 
    [label={[xshift=1.25em, yshift=-1.5em] $1$}] {$v_2$};
  \node[node,above=1.25em of v2] (s1)
    [label={[xshift=1.25em, yshift=-1.5em] $1$}] {$s_1$};
  \node [above=1.25em of s1] (in1) {};
  \draw[conn] (in1) -- (s1);
  \draw[conn] (s1) -- (v2);
  \draw[conn] (v2) edge [loop below] (v2);
  \node [below=3.0em, align=flush center] at (v2) {$G_3$};
  \node[node, right=1.5em of v2, double] (v2') 
    [label={[xshift=1.25em, yshift=-1.5em] $2$}] {$v_2$};
  \node[node,above=1.25em of v2'] (s2)
    [label={[xshift=1.25em, yshift=-1.5em] $2$}] {$s_2$};
  \node [above left=1.25em and 0.5em of s2] (in2) {};
  \node [above right=1.25em and 0.5em of s2] (in3) {};
  \draw[conn] (in2) -- (s2);
  \draw[conn] (in3) -- (s2);
  \draw[conn] (s2) -- (v2');
  \draw[conn] (s2) edge [bend right=25] (v2');
  \draw[conn] (v2') edge [loop below] (v2');
  \node [below=3.0em, align=flush center] at (v2') {$G_4$};
  
  \node [above right=0.5em and 1.5em, align=flush center] (b) at (v2') {};
  \node [above right=0.5em and 5.0em, align=flush center] (e) at (v2') {};
  \draw[conn] (b) -- (e)
    node [midway, below] {\scriptsize Additivity};
  
  \node[node, right=5.0em of v2', double] (v2) 
    [label={[xshift=1.25em, yshift=-1.5em] $1$}] {$v_2$};
  \node[node,above=1.25em of v2] (s1)
    [label={[xshift=1.25em, yshift=-1.5em] $1$}] {$s_1$};
  \node [above=1.25em of s1] (in1) {};
  \draw[conn] (in1) -- (s1);
  \draw[conn] (s1) -- (v2);
  \draw[conn] (v2) edge [loop below] (v2);
  \node [below=3.0em, align=flush center] at (v2) {$G_3$};
  \node[node, right=1.5em of v2, double] (v2') 
    [label={[xshift=1.25em, yshift=-1.5em] $1$}] {$v_2$};
  \node[node,above=1.25em of v2'] (s2)
    [label={[xshift=1.25em, yshift=-1.5em] $1$}] {$s_2$};
  \node [above=1.25em of s2] (in2) {};
  \draw[conn] (in2) -- (s2);
  \draw[conn] (s2) -- (v2');
  \draw[conn] (s2) edge [bend right=25] (v2');
  \draw[conn] (v2') edge [loop below] (v2');
  \node [below=3.0em, align=flush center] at (v2') {$G_5$};
  \node[node, right=1.5em of v2', double] (v2') 
    [label={[xshift=1.25em, yshift=-1.5em] $1$}] {$v_2$};
  \node[node,above=1.25em of v2'] (s2)
    [label={[xshift=1.25em, yshift=-1.5em] $1$}] {$s_2$};
  \node [above=1.25em of s2] (in3) {};
  \draw[conn] (in3) -- (s2);
  \draw[conn] (s2) -- (v2');
  \draw[conn] (s2) edge [bend right=25] (v2');
  \draw[conn] (v2') edge [loop below] (v2');
  \node [below=3.0em, align=flush center] at (v2') {$G_6$};
\end{tikzpicture}
\caption{An illustration of the first four axioms. 
In the initial graph $G$, nodes $s_1$ and $s_2$ are the only sources, $s_2$ with twice as much weight: $b = \mathds{1}^{s_1} + 2 \cdot \mathds{1}^{s_2}$ and node $t$ is the target. 
The values of t-Betweenness Centrality are placed right to the nodes.}
\label{fig:common_axioms}
\end{figure*}
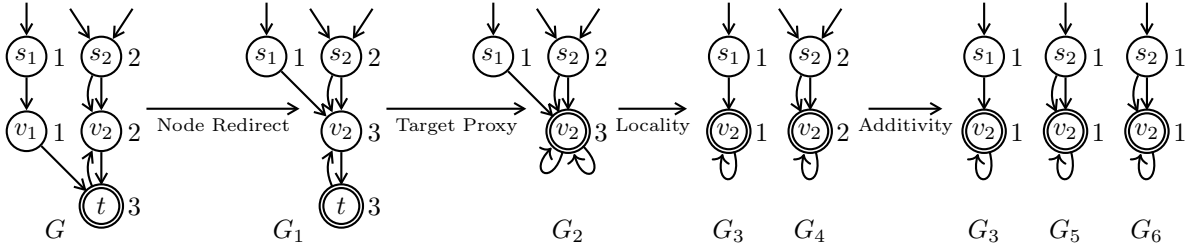

\emph{Locality} describes the operation of joining two networks with the common target.
This may occur for example in the computer network, when two separate subnetworks are connected by one server.
The axiom states that the centrality of all nodes other than the target does not change and the centrality of the target is the sum of its centralities from both separate graphs.
This means that the behavior of data packets is not affected by the existence of another independent part of the network with the same target.

We can also look from the other perspective and say that \emph{Locality} describes splitting the network at the cut vertex.
When the cut vertex is the target node, centralities in each component it joins are independent.
A similar axiom under the same name was proposed for undirected graphs in \cite{Skibski:etal:2019:attachment} where separate connected components are considered.

\emph{Additivity} states that the centrality treated as a function of the node weights is additive. 
This means that the packets travel independently from each other, they do not collide in any way.

\emph{Node Redirect} formalizes the intuition that if the packet from two nodes has the same possible further routes, then redirecting one of these nodes into the other would not change the centralities of other nodes.
Moreover, the centrality of the combined node will be the sum of centralities of both nodes in the original graph.
This axiom was proposed in \cite{Was:Skibski:2020:pagerank}, but in our version we do not allow redirecting of (and to) the target node.

\emph{Target Proxy} can be understood in the following way: if every path to the target node goes through a \emph{proxy}, $v$, then the role in transferring data to the target is the same as the role in transferring data to node $v$.
In the context of the Internet network, this axiom can be interpreted as the \emph{layered system} REST constraint \cite{Fielding:2000}.
It is a rule for designing web API which states that the communication between the clients and the target server should not be affected if the target is hidden behind a firewall, proxy or a load balancer. 

\subsection{Shortest-Paths Axioms}
Let us now present two axioms specific for centralities based on shortest-paths.
See Figure~\ref{fig:B_axioms} for an illustration.

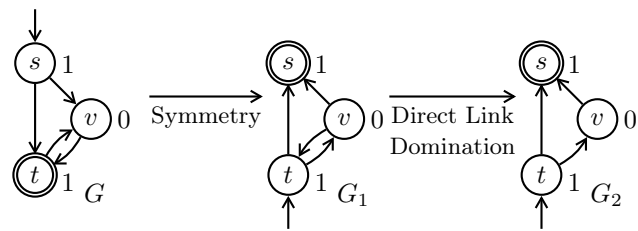
\begin{figure}[b]
\centering
\begin{tikzpicture}
  \node[node] (s) 
    [label={[xshift=1.25em, yshift=-1.5em] $1$}] {$s$};
  \node[node,below right = 1em and 1em of s] (v)
    [label={[xshift=1.25em, yshift=-1.5em] $0$}] {$v$};
  \node[node,below left = 1em and 1em of v, double] (t)
    [label={[xshift=1.25em, yshift=-1.75em] $1$}] {$t$};
  \node [above=1.25em of s] (in1) {};
  \draw[conn] (in1) -- (s);
  \draw[conn] (s) -- (v);
  \draw[conn] (s) -- (t);
  \draw[conn] (v) edge [bend left=15] (t);
  \draw[conn] (t) edge [bend left=15] (v);
  \node [below right=0em and 1.5em, align=flush center] at (t) {$G$};
  
  \node [above right=0.5em and 1.5em, align=flush center] (b) at (v) {};
  \node [above right=0.5em and 6.5em, align=flush center] (e) at (v) {};
  \draw[conn] (b) -- (e)
    node [midway, below] {\small Symmetry};
    
  \node[node, right=8em of s, double] (s) 
    [label={[xshift=1.25em, yshift=-1.5em] $1$}] {$s$};
  \node[node,below right = 1em and 1em of s] (v)
    [label={[xshift=1.25em, yshift=-1.5em] $0$}] {$v$};
  \node[node,below left = 1em and 1em of v] (t)
    [label={[xshift=1.25em, yshift=-1.75em] $1$}] {$t$};
  \node [below=1.25em of t] (in1) {};
  \draw[conn] (in1) -- (t);
  \draw[conn] (v) -- (s);
  \draw[conn] (t) -- (s);
  \draw[conn] (t) edge [bend right=15] (v);
  \draw[conn] (v) edge [bend right=15] (t);
  \node [below right=0em and 1.5em, align=flush center] at (t) {$G_1$};
  
  \node [above right=0.5em and 1.0em, align=flush center] (b) at (v) {};
  \node [above right=0.5em and 6.5em, align=flush center] (e) at (v) {};
  \draw[conn] (b) -- (e)
    node [midway, below, align=center] {\small Direct Link\\ \small Domination};
    
  \node[node, right=8em of s, double] (s) 
    [label={[xshift=1.25em, yshift=-1.5em] $1$}] {$s$};
  \node[node,below right = 1em and 1em of s] (v)
    [label={[xshift=1.25em, yshift=-1.5em] $0$}] {$v$};
  \node[node,below left = 1em and 1em of v] (t)
    [label={[xshift=1.25em, yshift=-1.75em] $1$}] {$t$};
  \node [below=1.25em of t] (in1) {};
  \draw[conn] (in1) -- (t);
  \draw[conn] (v) -- (s);
  \draw[conn] (t) -- (s);
  \draw[conn] (t) edge [bend right=15] (v);
  \node [below right=0em and 1.5em, align=flush center] at (t) {$G_2$};
\end{tikzpicture}
\caption{An illustration of Symmetry and Direct Link Domination.
Node $s$ is the only source: $b = \mathds{1}^s$ and node $t$ is the target. The values of t-Betweenness Centrality are placed right to the nodes.}
\label{fig:B_axioms}
\end{figure}

\begin{axiom}
\label{Symmetry}
(Symmetry) For every graph $G=(V, E) \in \mathcal{G}_t$, source node $s \in V$ and node $w \in V$ such that $G' = (V, \{(u, v) : (v, u) \in E\}) \in \mathcal{G}_s$: 
$$F^{s}_w(G', \mathds{1}^t) = F^{t}_w(G, \mathds{1}^s).$$
\end{axiom}

\begin{axiom}
\label{Direct Link Domination}
(Direct Link Domination) For every graph $G=(V, E) \in \mathcal{G}_t$ with weights $b$ such that $(v, t), (v, u) \in E$, $u \neq t$ and node $w \in V$:
$$F^t_w(G - \{(v, u)\}, b) = F^{t}_w(G, b).$$
\end{axiom}

\emph{Symmetry} states that if there is only one source node, then reversing the graph and swapping source and target nodes does not change centralities in the graph.
Note that this operation applies only if the reversed graph belongs to $\mathcal{G}_s$, i.e., node $s$ which is the target in the reversed graph is reachable from every node.
It is suitable when data packets would be transferred through the same paths from the target to the source in the reversed network. 
This is true if data packets go through the shortest paths.
However, it is not the case for the random walk, for which some paths would be used more often in the reversed network than in the original one.

\emph{Direct Link Domination} states that if from $v$ there is a direct connection to $t$, then we can delete other outgoing edges of $v$. 
This captures the assumption that a node which can send a data packet directly to the target will do so.
For example, when sending parcels through the transportation network it is natural to anticipate that the logistics company will choose the direct connection whenever it is possible.
At the same time, the axiom does not impose any restrictions on the behavior of packets in other nodes. 

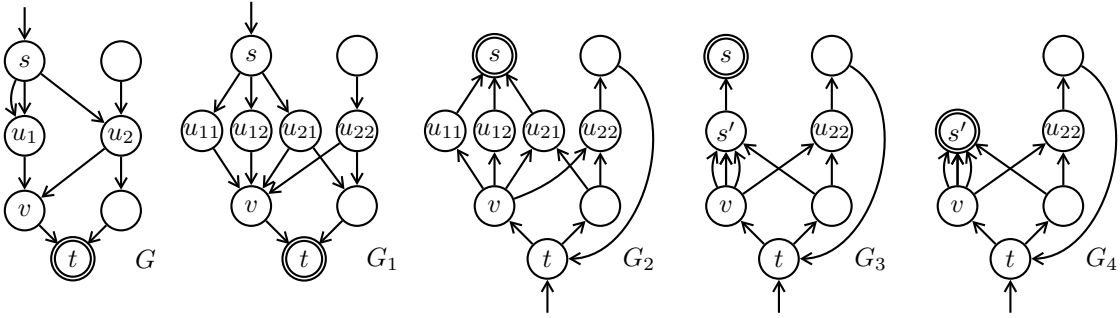
\begin{figure}[!t]
\centering
\begin{tikzpicture}[transform shape]%
  \node[node,yshift=1cm] (r)  at (10, 10) {};
  \node[node,below=1.25em of r] (w2) {$u_2$};
  \node[node,below=1.25em of w2] (u) {};
  \node[node,below left=1.0em of u, double] (t) {$t$};
  \node[node,above left=1.0em of t] (v) {$v$};
  \node[node,above=1.25em of v] (w1) {$u_1$};
  \node[node,above=1.25em of w1] (s) {$s$};
  \node[below=1.25em of t] (in1) {};
  \node[above=1.25em of s] (in1) {};
  \draw[conn] (in1) -- (s);
  \draw[conn] (r) -- (w2);
  \draw[conn] (w2) -- (u);
  \draw[conn] (u) -- (t);
  \draw[conn] (v) -- (t);
  \draw[conn] (w1) -- (v);
  \draw[conn] (w2) -- (v);
  \draw[conn] (s) -- (w1);
  \draw[conn] (s) edge [bend right=25] (w1);
  \draw[conn] (s) -- (w2);
  \node [right=2.0em, align=flush center] at (t) {$G$};
\end{tikzpicture}
\ 
\begin{tikzpicture}[transform shape]
  \node[node,yshift=1cm] (r)  at (10, 10) {};
  \node[node,below=1.25em of r] (w2) {$u_{22}$};
  \node[node,below=1.25em of w2] (u) {};
  \node[node,below left=1.25em of u, double] (t) {$t$};
  \node[node,above left=1.25em of t] (v) {$v$};
  \node[node,above=1.25em of v] (w12) {$u_{12}$};
  \node[node,left=0.25em of w12] (w11) {$u_{11}$};
  \node[node,right=0.25em of w12] (w21) {$u_{21}$};
  \node[node,above=1.25em of w12] (s) {$s$};
  \node[below=1.25em of t] (in1) {};
  \node[above=1.25em of s] (in1) {};
  \draw[conn] (in1) -- (s);
  \draw[conn] (r) -- (w2);
  \draw[conn] (w2) -- (u);
  \draw[conn] (w21) -- (u);
  \draw[conn] (u) -- (t);
  \draw[conn] (v) -- (t);
  \draw[conn] (w12) -- (v);
  \draw[conn] (w11) -- (v);
  \draw[conn] (w2) -- (v);
  \draw[conn] (w21) -- (v);
  \draw[conn] (s) -- (w12);
  \draw[conn] (s) -- (w11);
  \draw[conn] (s) -- (w21);
  \node [right=2.0em, align=flush center] at (t) {$G_1$};
\end{tikzpicture}
\ 
\begin{tikzpicture}[transform shape]
  \node[node,yshift=1cm] (r)  at (10, 10) {};
  \node[node,below=1.25em of r] (w2) {$u_{22}$};
  \node[node,below=1.25em of w2] (u) {};
  \node[node,below left=1.25em of u] (t) {$t$};
  \node[node,above left=1.25em of t] (v) {$v$};
  \node[node,above=1.25em of v] (w12) {$u_{12}$};
  \node[node,left=0.25em of w12] (w11) {$u_{11}$};
  \node[node,right=0.25em of w12] (w21) {$u_{21}$};
  \node[node,above=1.25em of w12, double] (s) {$s$};
  \node[below=1.25em of t] (in1) {};
  \draw[conn] (in1) -- (t);
  \draw[conn] (w2) -- (r);
  \draw[conn] (u) -- (w2);
  \draw[conn] (u) -- (w21);
  \draw[conn] (t) -- (u);
  \draw[conn] (t) -- (v);
  \draw[conn] (v) -- (w12);
  \draw[conn] (v) -- (w11);
  \draw[conn] (v) edge [bend right=15] (w2);
  \draw[conn] (v) -- (w21);
  \draw[conn] (w12) -- (s);
  \draw[conn] (w11) -- (s);
  \draw[conn] (w21) -- (s);
  \draw[conn] (r) edge [bend left=75] (t);
  \node [right=2.5em, align=flush center] at (t) {$G_2$};
\end{tikzpicture}
\ 
\begin{tikzpicture}[transform shape]%
  \node[node,yshift=1cm] (r)  at (10, 10) {};
  \node[node,below=1.25em of r] (w2) {$u_{22}$};
  \node[node,below=1.25em of w2] (u) {};
  \node[node,below left=1.25em of u] (t) {$t$};
  \node[node,above left=1.25em of t] (v) {$v$};
\  \node[node,above=1.25em of v] (w1) {$s'$};
  \node[node,above=1.25em of w1, double] (s) {$s$};
  \node[below=1.25em of t] (in1) {};
  \draw[conn] (in1) -- (t);
  \draw[conn] (w2) -- (r);
  \draw[conn] (u) -- (w2);
  \draw[conn] (t) -- (u);
  \draw[conn] (t) -- (v);
  \draw[conn] (u) -- (w1);
  \draw[conn] (v) -- (w1);
  \draw[conn] (v) edge [bend left=25] (w1);
  \draw[conn] (v) -- (w1);
  \draw[conn] (v) edge [bend right=25] (w1);
  \draw[conn] (v) -- (w2);
  \draw[conn] (w1) -- (s);
  \draw[conn] (r) edge [bend left=75] (t);
  \node [right=2.5em, align=flush center] at (t) {$G_3$};
\end{tikzpicture}
\ 
\begin{tikzpicture}[transform shape]%
  \node[node,yshift=1cm] (r)  at (10, 10) {};
  \node[node,below=1.25em of r] (w2) {$u_{22}$};
  \node[node,below=1.25em of w2] (u) {};
  \node[node,below left=1.25em of u] (t) {$t$};
  \node[node,above left=1.25em of t] (v) {$v$};
  \node[node,above=1.25em of v, double] (w1) {$s'$};
  \node[below=1.25em of t] (in1) {};
  \draw[conn] (in1) -- (t);
  \draw[conn] (w2) -- (r);
  \draw[conn] (u) -- (w2);
  \draw[conn] (t) -- (u);
  \draw[conn] (t) -- (v);
  \draw[conn] (u) -- (w1);
  \draw[conn] (v) -- (w1);
  \draw[conn] (v) edge [bend left=25] (w1);
  \draw[conn] (v) -- (w1);
  \draw[conn] (v) edge [bend right=25] (w1);
  \draw[conn] (v) -- (w2);
  \draw[conn] (r) edge [bend left=75] (t);
  \node [right=2.5em, align=flush center] at (t) {$G_4$};
\end{tikzpicture}
\caption{The key inductive step from the proof of Theorems~\ref{theorem:betweenness} and~\ref{theorem:stress}.
Assume node $v$ is at distance at least $2$ from $s$.
We transform the original graph $G$ in a way that the distance from the source to the target is shortened, but the centrality of $v$ remains unchanged. 
First, using Node Redirect, we create copies of direct successors of $s$, each with the original outgoing edges, but only one incoming edge ($G_1$). 
Then, using Symmetry (and No Target Outlet), we reverse the graph ($G_2$). 
Next, using Node Redirect, we merge copies with an edge to $s$ into one node $s'$ ($G_3$). 
Finally, using Target Proxy, we merge $s$ and $s'$ ($G_4$).}
\label{fig:B_induction}
\end{figure}

\subsection{Atom Axiom}
We conclude with a simple atom axiom.

\begin{axiom}
\label{Atom11}
(Atom $1$-$1$) For every node $s$, natural number $k \in \mathds{N}_+$ and graph $G = (\{s,t\}, k \cdot \{(s,t)\})$:
$$F^t_s(G, \mathds{1}^s) = 1 = F^t_t(G, \mathds{1}^s).$$
\end{axiom}

\emph{Atom $1$-$1$} specifies the centrality in a simple graph with two nodes, the source $s$ and the target $t$, and $k$ edges from $s$ to $t$.
It states that both nodes should have centrality equal to $1$, as this is the amount of information packets both nodes are responsible for.
Hence, the number of edges does not matter.
This is the case, for example, when edges represent links between webpages: if all links from one website point to the other, their number does not matter.
Atom $1$-$1$ is satisfied by t-Betweenness and t-Random Walk Betweenness.

The following theorem contains the axiomatic characterization of t-Betweenness. 
The full proof can be found in the appendix.

\begin{theorem}\label{theorem:betweenness}
A t-centrality satisfies Locality, Additivity, Node Redirect, Target Proxy, Symmetry, Direct Link Domination and Atom $1$-$1$ if and only if it is t-Betweenness Centrality.
\end{theorem}
\begin{proof}[Sketch of the proof]
It is easy to check that t-Betweenness Centrality satisfies the axioms listed in Theorem~\ref{theorem:betweenness} (Lemma~\ref{lemma:b_to_axioms} in the appendix).
Hence, in what follows, we will show that this set of axioms uniquely characterizes a centrality measure.

We begin by showing three simple properties.
\begin{itemize}
    \item \emph{(Anonymity)}: Any node other than the target can be renamed without changing any centralities (Lemma~\ref{lemma:all:anonymity}).
    \item \emph{(Target Self-Loop)}: Deleting a self-loop of the target does not change any centralities (Lemma~\ref{lemma:all:t_loop}).
    \item \emph{(No Target Outlet)}: Deleting any outgoing edge of the target does not change any centralities (Lemma~\ref{lemma:b:no_target_outlet}).
\end{itemize}

No Target Outlet is an especially useful property as it allows us to add edges from the target $t$ to all other nodes in the graph before using Symmetry.
In this way, we make sure that after reversing the graph from each node there will be a path to the new target, possibly through $t$, which is the requirement of the Symmetry axiom.

For now, let us concentrate on a graph with one source $s$ with weight $1$, i.e., $b = \mathds{1}^s$.
We proceed by induction on the distance from the source $s$ to the target $t$: $dist_{s,t}(G)$. 
In the base case we consider $dist_{s,t}(G) \le 2$:
\begin{itemize}
\item If $dist_{s, t}(G) = 0$ (Lemma~\ref{lemma:b:dist_0}), then the only source is also the target.
In such a case, we present the graph as the sum of two graphs:
the original graph with weight of $t$ changed to zero and the second graph with only one node $t$ with unitary weight.
From Locality and Additivity we get that $F^t_t(G, \mathds{1}^t) = F^t_t((\{t\}, \{\}), \mathds{1}^t)$ and $F^t_v(G, \mathds{1}^t) = 0$ for $v \in V - \{t\}$.
Now, the centrality of $t$ in graph $(\{t\}, \{\})$ with weights $\mathds{1}^t$ can be determined based on the Atom axiom by using Target Proxy and Target Self-Loop.

\item If $dist_{s, t}(G) = 1$ (Lemma~\ref{lemma:b:dist_1}), then the only source and the target are connected by at least one edge. 
Here, using Node Redirect and Locality we decompose graph $G$ into the original graph with weight of $s$ changed to zero and a graph with two nodes, source $s$ and target $t$, and $k$ edges from $s$ to $t$, with unitary weight of $s$.
Centralities in the first graph all equal zero from Additivity and centralities in the second graph are known from the Atom axiom.

\item If $dist_{s, t}(G) = 2$ (Lemmas~\ref{lemma:b:dist_2_simple}--\ref{lemma:b:dist_2}), then we know that at least one of the successors of $s$ is a predecessor of $t$. 
First, we use Node Redirect to split each successor of $s$ into several nodes so that each has only one incoming edge.
Then, using Symmetry and Node Redirect for the reversed graph we split predecessors of $t$ so that each copy has only one outgoing edge.
In the resulting graph all nodes which are both successors of $s$ and predecessors of $t$ are isomorphic, which allows us to deduce that they have equal centralities.
To argue what are the centralities of other nodes, we merge isomorphic nodes using Node Redirect and use Target Proxy to obtain the case where the distance from $s$ to $t$ equals one.
\end{itemize}

Let us discuss the inductive step (Lemma~\ref{lemma:b:single_source}). 
Fix graph $G$ with $dist_{s,t}(G) \geq 3$ and some node $v$. 
Since 
\[ dist_{s, v}(G) + dist_{v, t}(G) \geq dist_{s, t}(G) \geq 3,\]
we either have $dist_{s, v}(G) \geq 2$ or $dist_{v, t}(G) \geq 2$.
Let us assume the former; in the other case, we reverse the graph and based on Symmetry proceed in the same way.
Now, we show that we can transform the graph in a way that the distance from the source to the target decreases, but the centrality of $v$ remains unchanged.
We present this key construction in Figure~\ref{fig:B_induction}.

So far, we have considered only one source.
If there are multiple sources, by using Additivity we split the graph into several copies, each with a single unitary source (Lemma~\ref{lemma:b:determined}):
\[F^t_v(G, b) = \sum_{s \in V} b(s) \cdot F^t_v(G, \mathds{1}^s).\]
Based on that, we show that if a centrality measure satisfies Atom $1$-$1$, then it is t-Betweenness Centrality (Lemma~\ref{lemma:axioms_to_betweenness}).
This concludes the proof.
\end{proof}

\section{Axiomatization of t-Stress Centrality}
To characterize t-Stress Centrality, we propose the following modification of Atom $1$-$1$:

\begin{axiom}
\label{Atomkk}
(Atom $k$-$k$) For every node $s$, natural number $k \in \mathds{N}_+$ and graph $G = (\{s,t\}, k \cdot \{(s,t)\})$:
$$F^t_s(G, \mathds{1}^s) = k = F^t_t(G, \mathds{1}^s).$$
\end{axiom}

As Atom $1$-$1$, \emph{Atom $k$-$k$} specifies the centrality in a graph with two nodes, the source $s$ and the target $t$, and $k$ edges from $s$ to $t$.
In such a case, Atom $k$-$k$ states that both nodes have centrality equal to the number of edges between them.
In particular, the nodes' assessment increases when there are more edges between the source and the target.
Such an approach makes sense if we assume the information packets can be duplicated and sent through each edge.
Then, if we measure not the relative importance of nodes, but the absolute number of packets that go through the node we get these values.
Out of three centralities considered by us, Atom $k$-$k$ is satisfied only by t-Stress Centrality.

Now, we show that replacing Atom $1$-$1$ by Atom $k$-$k$ in the axiomatization of t-Betweenness Centrality results in the axiomatization of t-Stress Centrality. 

\begin{theorem}\label{theorem:stress}
A t-centrality satisfies Locality, Additivity, Node Redirect, Target Proxy, Symmetry, Direct Link Domination and Atom $k$-$k$ if and only if it is t-Stress Centrality.
\end{theorem}

The proof of Theorem~\ref{theorem:stress} is analogous to the proof of Theorem~\ref{theorem:betweenness}.
Specifically, we prove that t-Stress Centrality satisfies the axioms (Lemma~\ref{lemma:s_to_axioms}) and then using the reasoning from the proof of Theorem~\ref{theorem:betweenness} we show it is a unique such measure (Lemma~\ref{lemma:axioms_to_stress}).

\section{Axiomatization of t-Random Walk Betweenness}

So far, we have considered two centrality measures based on shortest paths. 
In this section, we present the axiomatization of their random-walk counterpart: t-Random Walk Betweenness.

To this end, we adapt to our setting two axioms from the axiomatization of PageRank in \cite{Was:Skibski:2020:pagerank}: \emph{Edge Swap} and \emph{Edge Multiplication}.
Both axioms are satisfied by t-Random Walk Betweenness Centrality, but not by t-Betweenness nor t-Stress Centralities.

\begin{figure}[b]
\centering
\begin{tikzpicture}
  \node[node] (s1) 
    [label={[xshift=1.25em, yshift=-1.5em] $1$}] {$s_1$};
  \node[node,right = 1.25em of s1] (s2)
    [label={[xshift=1.25em, yshift=-1.5em] $1$}] {$s_2$};
  \node[node,below = 1.25em of s1] (v1)
    [label={[xshift=1.25em, yshift=-1.5em] $1$}] {$v_1$};
  \node[node,below = 1.25em of s2] (v2)
    [label={[xshift=1.25em, yshift=-1.5em] $\frac{1}{2}$}] {$v_2$};
  \node[node,below right = 1.25em and 0.5em of v1, double] (t)
    [label={[xshift=1.25em, yshift=-1.5em] $2$}] {$t$};
  \node [above=1.25em of s1] (in1) {};
  \node [above=1.25em of s2] (in2) {};
  \draw[conn] (in1) -- (s1);
  \draw[conn] (in2) -- (s2);
  \draw[conn] (s1) -- (v1);
  \draw[conn] (s1) edge [bend left=25] (t);
  \draw[conn] (s2) -- (v1);
  \draw[conn] (s2) -- (v2);
  \draw[conn] (v1) -- (t);
  \draw[conn] (v2) -- (t);
  \node [below left=0em and 1em, align=flush center] at (t) {$G$};
  
  \node [above right=0.25em and 1.5em, align=flush center] (b) at (v2) {};
  \node [above right=0.25em and 5.0em, align=flush center] (e) at (v2) {};
  \draw[conn] (b) -- (e)
    node [midway, below, align=center] {\small Edge\\ \small Swap};
    
  \node[node, right = 8em of s1] (s1) 
    [label={[xshift=1.25em, yshift=-1.5em] $1$}] {$s_1$};
  \node[node,right = 1.25em of s1] (s2)
    [label={[xshift=1.25em, yshift=-1.5em] $1$}] {$s_2$};
  \node[node,below = 1.25em of s1] (v1)
    [label={[xshift=1.25em, yshift=-1.5em] $1$}] {$v_1$};
  \node[node,below = 1.25em of s2] (v2)
    [label={[xshift=1.25em, yshift=-1.5em] $\frac{1}{2}$}] {$v_2$};
  \node[node,below right = 1.25em and 0.5em of v1, double] (t)
    [label={[xshift=1.25em, yshift=-1.5em] $2$}] {$t$};
  \node [above=1.25em of s1] (in1) {};
  \node [above=1.25em of s2] (in2) {};
  \draw[conn] (in1) -- (s1);
  \draw[conn] (in2) -- (s2);
  \draw[conn] (s1) -- (v1);
  \draw[conn] (s1) edge [bend left=25] (v1);
  \draw[conn] (s2) edge [bend right=15] (t);
  \draw[conn] (s2) -- (v2);
  \draw[conn] (v1) -- (t);
  \draw[conn] (v2) -- (t);
  \node [below left=0em and 0.75em, align=flush center] at (t) {$G_1$};
  
  \node [above right=0.25em and 1.0em, align=flush center] (b) at (v2) {};
  \node [above right=0.25em and 5.0em, align=flush center] (e) at (v2) {};
  \draw[conn] (b) -- (e)
    node [midway, below, align=center] {\small Edge\\ \small Multiplication};
    
  \node[node, right = 8em of s1] (s1) 
    [label={[xshift=1.25em, yshift=-1.5em] $1$}] {$s_1$};
  \node[node,right = 1.25em of s1] (s2)
    [label={[xshift=1.25em, yshift=-1.5em] $1$}] {$s_2$};
  \node[node,below = 1.25em of s1] (v1)
    [label={[xshift=1.25em, yshift=-1.5em] $1$}] {$v_1$};
  \node[node,below = 1.25em of s2] (v2)
    [label={[xshift=1.25em, yshift=-1.5em] $\frac{1}{2}$}] {$v_2$};
  \node[node,below right = 1.25em and 0.5em of v1, double] (t)
    [label={[xshift=1.25em, yshift=-1.5em] $2$}] {$t$};
  \node [above=1.25em of s1] (in1) {};
  \node [above=1.25em of s2] (in2) {};
  \draw[conn] (in1) -- (s1);
  \draw[conn] (in2) -- (s2);
  \draw[conn] (s1) -- (v1);
  \draw[conn] (s2) edge [bend right=15] (t);
  \draw[conn] (s2) -- (v2);
  \draw[conn] (v1) -- (t);
  \draw[conn] (v2) -- (t);
  \node [below left=0em and 0.75em, align=flush center] at (t) {$G_2$};
\end{tikzpicture}
\caption{An illustration of Edge Swap and Edge Multiplication. 
Nodes $s_1$ and $s_2$ are the only sources: $b = \mathds{1}^{s_1} + \mathds{1}^{s_2}$ and node $t$ is the target.
The values according to t-Random Walk Betweenness are placed right to the nodes.}
\label{fig:PR_axioms}
\end{figure}
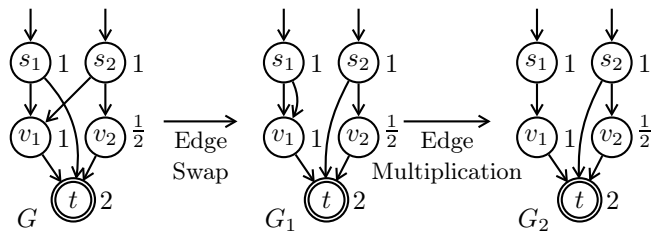

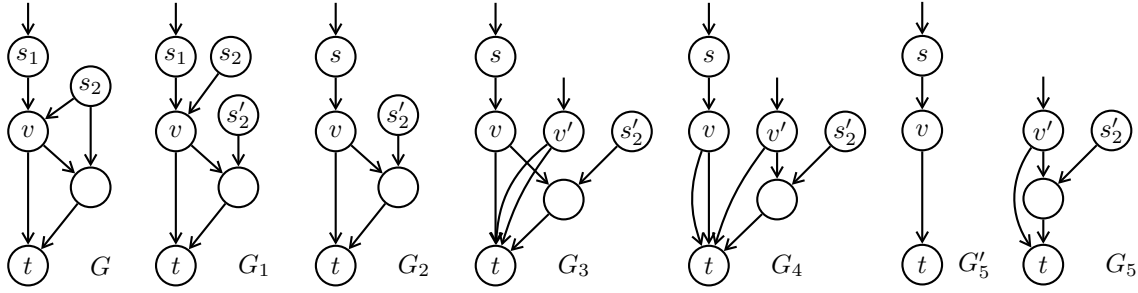
\begin{figure*}[!t]
\centering
\begin{tikzpicture}%
  \node[node] (s1)  at (10, 10) {$s_1$};
  \node[node,below right=0.0em and 1.25em of s1] (s2) {$s_2$};
  \node[node,below=1.25em of s1] (v) {$v$};
  \node[node,below right=1.0em and 1.25em of v] (w) {};
  \node[node,below=3.5em of v] (t) {$t$};
  \node[above=1.25em of s1] (in1) {};
  \draw[conn] (in1) -- (s1);
  \draw[conn] (s1) -- (v);
  \draw[conn] (s2) -- (v);
  \draw[conn] (s2) -- (w);
  \draw[conn] (v) -- (w);
  \draw[conn] (w) -- (t);
  \draw[conn] (v) -- (t);
  \node [right=2.0em, align=flush center] at (t) {$G$};
\end{tikzpicture}
\ \ \ 
\begin{tikzpicture}%
  \node[node] (s1)  at (10, 10) {$s_1$};
  \node[node,right=0.5em of s1] (s2) {$s_2$};
  \node[node,below=1.25em of s1] (v) {$v$};
  \node[node,above right=-0.5em and 1.25em of v] (s3) {$s'_2$};
  \node[node,below right=1.0em and 1.25em of v] (w) {};
  \node[node,below=3.5em of v] (t) {$t$};
  \node[above=1.25em of s1] (in1) {};
  \draw[conn] (in1) -- (s1);
  \draw[conn] (s1) -- (v);
  \draw[conn] (s2) -- (v);
  \draw[conn] (s3) -- (w);
  \draw[conn] (v) -- (w);
  \draw[conn] (w) -- (t);
  \draw[conn] (v) -- (t);
  \node [right=2.0em, align=flush center] at (t) {$G_1$};
\end{tikzpicture}
\ \ \ 
\begin{tikzpicture}%
  \node[node] (s1)  at (10, 10) {$s$};
  \node[node,below=1.25em of s1] (v) {$v$};
  \node[node,above right=-0.5em and 1.25em of v] (s3) {$s'_2$};
  \node[node,below right=1.0em and 1.25em of v] (w) {};
  \node[node,below=3.5em of v] (t) {$t$};
  \node[above=1.25em of s1] (in1) {};
  \draw[conn] (in1) -- (s1);
  \draw[conn] (s1) -- (v);
  \draw[conn] (s3) -- (w);
  \draw[conn] (v) -- (w);
  \draw[conn] (w) -- (t);
  \draw[conn] (v) -- (t);
  \node [right=2.0em, align=flush center] at (t) {$G_2$};
\end{tikzpicture}
\ \ \ 
\begin{tikzpicture}%
  \node[node] (s1)  at (10, 10) {$s$};
  \node[node,below=1.25em of s1] (v) {$v$};
  \node[node,right=1.0em of v] (v') {$v'$};
  \node[node,right=1.0em of v'] (s3) {$s'_2$};
  \node[node,below=1.0em of v'] (w) {};
  \node[node,below=3.5em of v] (t) {$t$};
  \node[above=1.25em of s1] (in1) {};
  \draw[conn] (in1) -- (s1);
  \node[above=1.25em of v'] (in2) {};
  \draw[conn] (in2) -- (v');
  \draw[conn] (s1) -- (v);
  \draw[conn] (s3) -- (w);
  \draw[conn] (v) -- (w);
  \draw[conn] (w) -- (t);
  \draw[conn] (v) -- (t);
  \draw[conn] (v') edge [bend right=10] (t);
  \draw[conn] (v') edge [bend right=25] (t);
  \node [right=2.0em, align=flush center] at (t) {$G_3$};
\end{tikzpicture}
\ \ \ 
\begin{tikzpicture}%
  \node[node] (s1)  at (10, 10) {$s$};
  \node[node,below=1.25em of s1] (v) {$v$};
  \node[node,right=1.0em of v] (v') {$v'$};
  \node[node,right=1.0em of v'] (s3) {$s'_2$};
  \node[node,below=1.0em of v'] (w) {};
  \node[node,below=3.5em of v] (t) {$t$};
  \node[above=1.25em of s1] (in1) {};
  \draw[conn] (in1) -- (s1);
  \node[above=1.25em of v'] (in2) {};
  \draw[conn] (in2) -- (v');
  \draw[conn] (s1) -- (v);
  \draw[conn] (s3) -- (w);
  \draw[conn] (v') -- (w);
  \draw[conn] (w) -- (t);
  \draw[conn] (v') edge [bend right=10] (t);
  \draw[conn] (v) -- (t);
  \draw[conn] (v) edge [bend right=20] (t);
  \node [right=2.0em, align=flush center] at (t) {$G_4$};
\end{tikzpicture}
\ \ \ 
\begin{tikzpicture}%
  \node[node] (s1)  at (10, 10) {$s$};
  \node[node,below=1.25em of s1] (v) {$v$};
  \node[node,below=3.5em of v] (t) {$t$};
  \node[node,right=3.0em of v] (v') {$v'$};
  \node[node,right=1.0em of v'] (s3) {$s'_2$};
  \node[node,below=1.0em of v'] (w) {};
  \node[node,right=3.0em of t] (t') {$t$};
  \node[above=1.25em of s1] (in1) {};
  \draw[conn] (in1) -- (s1);
  \node[above=1.25em of v'] (in2) {};
  \draw[conn] (in2) -- (v');
  \draw[conn] (s1) -- (v);
  \draw[conn] (s3) -- (w);
  \draw[conn] (v') -- (w);
  \draw[conn] (w) -- (t');
  \draw[conn] (v') edge [bend right=35] (t');
  \draw[conn] (v) -- (t);
  \node [right=1.0em, align=flush center] at (t) {$G_5'$};
  \node [right=2.0em, align=flush center] at (t') {$G_5$};
\end{tikzpicture}
\caption{The key inductive step from the proof of Theorem~\ref{theorem:random_walk_betweenness}.
We transform the original acyclic graph $G$ in a way that the number of edges is decreased, but the centrality of $v$ remains unchanged.
First, we create copies of nodes without incoming edges, each with one outgoing edge ($G_1$).
Then, using Node Redirect, we merge copies of direct predecessors of $v$ into one node $s$ ($G_2$).
Now, using the new node technique, we add a node $v'$ with the same number of outgoing edges as $v$, all to $t$, and the weight equal to the centrality of $v$ ($G_3$).
Next, we use Edge Swap to exchange the outgoing edges of $v$ with the outgoing edges of $v'$ ($G_4$).
Finally, using Locality and Edge Multiplication, we split the graph in two graphs, one which has one edge less than $G$ ($G_5$) and another simple graph for which centralities follow from previous lemmas ($G_5'$).}
\label{fig:PR_DAG_induction}
\end{figure*} 

\begin{axiom}
\label{Edge Swap}
(Edge Swap) For every graph $G = (V, E) \in \mathcal{G}_t$ with weights $b$ and edges $(v, v'), (u, u') \in E$ such that $v, u \neq t$, $F_v(G, b) = F_u(G, b)$, $|\Gamma^+_v(G)| = |\Gamma^+_u(G)|$, $G' = G - \{(v, v'),(u, u')\}+\{(v, u'),(u, v')\}\in\mathcal{G}_t$ and node $w\in V$:
$$F^{t}_w(G', b) = F^{t}_w(G, b).$$
\end{axiom}

\begin{axiom}
\label{Edge Multiplication}
(Edge Multiplication) For every graph $G = (V, E) \in \mathcal{G}_t$ with weights $b$, number $k \in \mathds{Z}_{\geq 0}$ and nodes $v, w \in V$:
$$F^{t}_w(G + k \cdot \Gamma^+_v(G), b) = F^{t}_w(G, b).$$
\end{axiom}

\emph{Edge Swap} states that swapping ends of two outgoing edges of nodes with equal centralities and out-degrees does not affect the centrality of any node. 
This means that if information packets are equally often in two nodes, then it does not matter from which of them a node has an incoming edge.
The only difference between our version of the axiom and the original one from \cite{Was:Skibski:2020:pagerank} is the fact that swapped edges cannot start in the target node.

\emph{Edge Multiplication} states that creating additional copies of the outgoing edges of a node does not affect the centrality of any node. 
This means that it is not the total number of edges that matters, but their proportions.

Now, replacing Symmetry and Direct Link Domination by Edge Swap and Edge Multiplication results in the axiomatization of t-Random Walk Betweenness.
Interestingly, our t-Random Walk Betweenness axiomatization is very similar to the axiomatization of t-Betweenness, despite these centralities being based on completely different models of transmission.

Note that Additivity, while satisfied by t-Random Walk Betweenness, is implied by other axioms (Locality, Node Redirect, Target Proxy, Edge Swap, Edge Multiplication and Atom $1$-$1$). 
Hence, it does not appear in the theorem statement.

\begin{theorem}\label{theorem:random_walk_betweenness}
A t-centrality satisfies Locality, Node Redirect, Target Proxy, Edge Swap, Edge Multiplication and Atom $1$-$1$ if and only if it is t-Random Walk Betweenness.
\end{theorem}
\begin{proof}[Sketch of the proof]
It is easy to check that t-Random Walk Betweenness satisfies the axioms listed in Theorem~\ref{theorem:random_walk_betweenness} (Lemma~\ref{lemma:rwb_to_axioms} in the appendix).
Hence, in what follows, we will show that this set of axioms uniquely characterizes a centrality measure.

We begin by considering simple graphs with two nodes: the source $s$ and the target $t$ (which has $0$ weight).
\begin{itemize}
    \item First, we show that if there is only one edge, from $s$ to $t$, then the centralities of both nodes are equal $b(s)$. (Lemma~\ref{lemma:pr:one_arrow}). 
    \item Second, we show that outgoing edges from $t$ does not affect the centralities of $s$ and $t$ (Lemma~\ref{lemma:pr:almost_one_arrow}).
\end{itemize}

Building upon this, we show that if we add to a graph a new node $v$ with several edges $\{v,t\}$, then its centrality will be equal to its weight and centralities of nodes other than $t$ will not change (Lemma~\ref{lemma:pr:k_arrow_adjoin}).
This operation, that we call \emph{new node technique}, is used frequently in the remainder of the proof.
In particular, we use it to show that every node $v$ without incoming edges has centrality equal to its weight (Lemma~\ref{lemma:pr:siphon}).
Furthermore, assuming such $v$ has $k$ outgoing edges, if we replace it with $k$ copies, each with one outgoing edge and weight $b(v)/k$, then the centralities of other nodes will not change (Lemma~\ref{lemma:pr:siphon_split}).

In the next part of the proof, we show that on acyclic graphs the centrality is equal to t-Random Walk Betweenness.
The proof proceeds by induction on the number of edges (Lemmas~\ref{lemma:pr:long_arrow}--\ref{lemma:pr:dag}).
\begin{itemize}
\item If $|E| \le 1$, then the thesis follows from previous lemmas.
\item If $|E| \ge 2$, but all edges ends in $t$, then this result follows from Locality.
\item If $|E| = 2$, and not all edges ends in $t$, then we have: $G = (\{s,v,t\}, \{(s,v), (v,t)\})$. In such a case, centralities of $s$ and $v$ can be deduced by using Target Proxy and previous lemma concerning graph with only one edge. The centrality of $t$ can be obtained using Target Proxy, Edge Swap and Locality.
\item If $|E| \ge 3$, but not all edges end in $t$, then there exists a node other than the target with incoming edges.
Let $v$ be the topologically greatest node (first node in the topological order of nodes). Because it is topologically greatest, all its predecessors have no incoming edges.
Now, we show that we can transform a graph in a way that the number of edges decreases, but the centrality of $v$ remains unchanged. 
This construction is described in Figure~\ref{fig:PR_DAG_induction}.
\end{itemize}

Acyclic graphs are the basis for the induction on the number of cycles which constitutes the top level of the proof.

Now, assume that a graph contains some cycles.
We can also assume that the target node does not have any outgoing edges, as otherwise they can be deleted without changing any centralities (Lemma~\ref{lemma:pr:no_target_outlet}).

Let us discuss the inductive step of the induction on the number of cycles (Lemma~\ref{lemma:pr:determined}):
Fix a node $v$ that belongs to at least one cycle.
Using again the new node technique we add a node $v'$ with the same number of outgoing edges as $v$, all to $t$, and the weight equal to the centrality of $v$.
Now, using Edge Swap, we swap all outgoing edges of $v$ with all outgoing edges of $v'$.
In the new graph, nodes $v$ and $v'$ do not belong to any cycles, hence the number of cycles has decreased.
This concludes the proof.
\end{proof}

On the top level, our proof has a similar structure to the proof of the axiomatization of PageRank \cite{Was:Skibski:2020:pagerank}.
However, there are several differences. 
First of all, three out of six axioms (Locality, Target Proxy and Atom $1$-$1$) have not appeared in the axiomatization of PageRank.
Second of all, these axioms which have appeared have additional constraints, excluding their use on the target node.
Third, in our paper, we restrict ourselves to graphs in which the target node is reachable which makes many constructions from the original proof impossible.
As a result, five out of nine our lemmas for Random Walk Betweenness do not have direct counterparts in the mentioned paper.
Moreover, even some lemmas that have counterparts have vastly different proofs (in particular, key Lemma~\ref{lemma:pr:dag} presented in Figure~\ref{fig:PR_DAG_induction}).

\section{Conclusions}
We proposed the first axiomatization of three medial centralities: Betweenness, Stress and Random Walk Betweenness.
We focused on a setting with one target node and arbitrarily many source nodes.
This allowed us to focus on the key aspect of these measures.
We specified several properties which are satisfied by all three centrality measures. 
Also, we proposed axioms specific for Betweenness and Stress and determined axioms specific for Random Walk Betweenness which highlights the differences between these two approaches.
Our characterization not only deepens the understanding of centrality measures, but also could help in choosing a centrality measure for a specific application at hand.

Our work can be extended in many ways. 
The ultimate goal would be to create an axiomatization of considered measures in a setting with arbitrary many targets. 
This is, however, challenging as it precludes the use of axioms that do not apply to target nodes.
Also, some axioms (e.g., Node Redirect) are not satisfied by Betweenness, Stress and Random Walk Betweenness in their general form with multiple targets.

Another interesting question is how to include in the axiom system other, less popular medial centralities, such as Flow Betweenness Centrality~\cite{Freeman:etal:1991}.
Interestingly, Flow Betweenness Centrality satisfies several of the axioms that we considered (in particular, Locality, Additivity, Symmetry and Atom $k$-$k$) although it is based on yet another model of transition.
Finally, undirected or edge-weighted graphs can be considered.

\section*{Acknowledgements}
Wiktoria Ko{\'s}ny and Oskar Skibski were supported by the Polish National Science Centre Grant No. 2018/31/B/ST6/03201.

\bibliographystyle{plainnat} 
\bibliography{bibliography}

\appendix
\clearpage

\newgeometry{margin=2.5cm}

\begin{multicols}{2}
\section{Proof of Theorem~\ref{theorem:betweenness}}

In this section, we present the full proof of Theorem~\ref{theorem:betweenness}.
We start by showing that t-Betweenness Centrality satisfies the axioms.
In the main part, we prove that the axioms uniquely characterize a centrality measure.

\subsection{t-Betweenness \texorpdfstring{$\Rightarrow$}{=>} Axioms}
We will now consider each of the axioms: Locality, Additivity, Node Redirect, Target Proxy, Symmetry, Direct Link Domination, Atom 1-1 and show that t-Betweenness Centrality satisfies it.

\begin{lemma}\label{lemma:b_to_axioms}
t-Betweenness Centrality satisfies Locality, Additivity, Node Redirect, Target Proxy, Symmetry, Direct Link Domination and Atom $1$-$1$.
\end{lemma}

Let $\Pi_{v, u}(G)$ denote the set of shortest paths from $v$ to $u$ in $G$ (in particular, $\sigma_{v, u}(G) = |\Pi_{v, u}(G)|$).

\subsubsection{Locality:} 
Let us notice that for every $s \in V$ no shortest path from $s$ to $t$ in $G+G'$ includes any edge from $E'$. 
Let us assume such a path exists and let $(v', u')$ be the first edge from $E'$ in this path. 
We have $v' \in V'$, but -- since up to this point the path was in $G$ -- also $v' \in V$, so $v' = t$.
Hence, the path is not a shortest path to $t$. 
This means that for every $s \in V$ it holds $\Pi_{s, t}(G + G') = \Pi_{s, t}(G)$.
In particular, for every $w \in V$, $w' \in V' - \{t\}$ we have:
    \[\sigma_{s,t}(G + G', w) = \sigma_{s,t}(G, w),\]
    \[\sigma_{s,t}(G + G', w') = 0,\]
    \[\sigma_{s,t}(G + G') = \sigma_{s,t}(G).\]

This gives:
\begin{align*}
&B^t_w((G, b) + (G', b'))\\
&= \sum_{s \in V \cup V'} (b+b')(s) \cdot \frac{\sigma_{s,t}(G+G', w)}{\sigma_{s,t}(G+G')}\\
&= \sum_{s \in V - \{t\}} b(s) \cdot \frac{\sigma_{s,t}(G, w)}{\sigma_{s,t}(G)}
   +(b(t) + b'(t)) \cdot 0\\  
&+ \sum_{s' \in V' - \{t\}} b'(s') \cdot \frac{0}{\sigma_{s',t}(G+G')} 
 = B^t_w(G, b).
\end{align*}
    
As for the centrality of $t$, by noting that the empty path is always the only shortest path from $t$ to $t$ we get:
\begin{align*}
&B^t_t((G, b) + (G', b'))\\
&= \sum_{s \in V \cup V'} (b+b')(s) \cdot \frac{\sigma_{s,t}(G+G', t)}{\sigma_{s,t}(G+G')}\\
&= \sum_{s \in V - \{t\}} b(s) \cdot \frac{\sigma_{s,t}(G, t)}{\sigma_{s,t}(G)}
   +b(t) + b'(t)\\  
&+ \sum_{s' \in V' - \{t\}} b'(s') \cdot \frac{\sigma_{s',t}(G', t)}{\sigma_{s',t}(G')}\\
&= B^t_t(G, b) + B^t_t(G', b').
\end{align*}

\subsubsection{Additivity:} 
From the definition we have:
\begin{align*}
&B^t_w(G, b + b') = \sum_{s \in V} (b(s) + b'(s)) \cdot \frac{\sigma_{s,t}(G, w)}{\sigma_{s,t}(G)}\\
&= \sum_{s \in V} b(s) \cdot \frac{\sigma_{s,t}(G, w)}{\sigma_{s,t}(G)} + 
   \sum_{s \in V} b'(s) \cdot \frac{\sigma_{s,t}(G, w)}{\sigma_{s,t}(G)}\\ 
&= B^t_w(G, b) + B^t_w(G, b').
\end{align*}
        
\subsubsection{Node Redirect:} 
Let $(G', b') = R_{u \rightarrow v}(G, b)$.
We know shortest paths from $v$ to $t$ are the same in $G$ as in $G'$, as $u$, the out-twin of $v$, cannot be on a shortest path to $t$ in $G$.
Hence, $\Pi_{v, t}(G') = \Pi_{v, t}(G)$.
Moreover, since $u$ is the out-twin of $v$, they have the same shortest paths to $t$ in $G$, up to appropriately changing first edge, that is $\Pi_{v, t}(G') = \{((v, w_1), ... (w_k, t)) : ((u, w_1), ... (w_k, t)) \in \Pi_{u, t}(G)\}$.
In particular, for every $w \in V$ we have:
    \[\sigma_{v, t}(G', w) = \sigma_{v,t}(G, w) = \sigma_{u,t}(G, w),\]
    \[\sigma_{v, t}(G') = \sigma_{v,t}(G) = \sigma_{u,t}(G).\] 
 
Now, fix $s \in V - \{v, u\}$.
Clearly, shortest paths from $s$ to $t$ that do not go through $v$ or $u$ are the same in $G$ as in $G'$.
Consider paths that go through $v$ or $u$. 
Each such a path goes through one of the incoming edges of $v$ or $u$.
For every path from $s$ to $t$ in $G$ that goes through some edge to $v$ there is a path from $s$ to $t$ in $G'$ that also goes through this edge.
For every path from $s$ to $t$ in $G$ that goes through an edge $(w,u)$ to $u$ there is a path from $s$ to $t$ in $G'$ that goes through the corresponding edge $(w, v)$ and through node $v$.
And similarly, for every path from $s$ to $t$ in $G'$ that goes through $(w,v)$ to $v$ there is a corresponding path from $s$ to $t$ that is either the same, or the same apart from going through $(w, u)$ instead.
Let $R_{u \rightarrow v}(((w_1, w_2), ..., (w_i, u), (u, w_{i+2}), ..., (w_{k-1}, w_{k}))) = ((w_1, w_2), \dots, (w_i, v), (v, w_{i+2}), \dots, (w_{k-1}, w_{k})$.
We get $\Pi_{v, t}(G') = \Pi_{v, t}(G) + \{R_{u \rightarrow v}(\pi) : \pi \in \Pi_{u, t}(G)\}$.
In particular, for every $w \in V - \{v, v'\}$ we have:
    \[\sigma_{s, t}(G', w) = \sigma_{s,t}(G, w),\]
    \[\sigma_{s, t}(G', v) = \sigma_{s,t}(G, v) + \sigma_{s,t}(G, u),\]
    \[\sigma_{s, t}(G') = \sigma_{s,t}(G).\]

This gives:
\begin{align*}
&B^t_v(G', b') \\
&= \sum_{s \in V - \{v, u\}} b'(s) \frac{\sigma_{s,t}(G', v)}{\sigma_{s,t}(G')} 
    + b'(v) \frac{\sigma_{v,t}(G', v)}{\sigma_{v,t}(G')}\\
&= \sum_{s \in V - \{v, u\}} b(s) \frac{\sigma_{s,t}(G, v) + \sigma_{s,t}(G, u)}{\sigma_{s,t}(G)}\\ 
&+ (b(v) + b(u)) \frac{\sigma_{v,t}(G, v)}{\sigma_{v,t}(G)}\\
&= B^t_v(G, b) + B^t_u(G, b)
\end{align*}

and for every $w \in V - \{v, u\}$:
\begin{align*}
&B^t_w(G', b') \\
&= \sum_{s \in V - \{v, u\}} b'(s) \frac{\sigma_{s,t}(G', w)}{\sigma_{s,t}(G')} 
    + b'(v) \frac{\sigma_{v,t}(G', w)}{\sigma_{v,t}(G')}\\
&= \sum_{s \in V - \{v, u\}} b(s) \frac{\sigma_{s,t}(G, w)}{\sigma_{s,t}(G)} 
    + (b(v) + b(u)) \frac{\sigma_{v,t}(G, w)}{\sigma_{v,t}(G)}\\
&= B^t_w(G, b).
\end{align*}

\subsubsection{Target Proxy:}
Let $G' = M_{t \rightarrow v}(G, b)$.
Let us notice for every $s \in V - \{t\}$ every path from $s$ to $t$ has $(v, t)$ as the last edge, because it is the only edge ending in $t$.
Hence, for every shortest path from $s$ to $t$ in $G$ there is a path from $s$ to $v$ in $G'$ that lacks only the last edge and for every shortest path from $s$ to $v$ in $G'$ there is a path from $s$ to $t$ in $G$ that is the same apart from the additional edge $(v, t)$ at the end.
This means that for every $s \in V$ it holds $\Pi_{s, v}(G') = \{(e_1, \dots, e_{k-1}): (e_1, \dots, e_{k-1}, e_k) \in \Pi_{s, t}(G)\}$.
In particular, for every $w \in V - \{t\}$ we have:
    \[\sigma_{s, v}(G', w) = \sigma_{s,t}(G, w),\]
    \[\sigma_{s, v}(G') = \sigma_{s,t}(G).\]

This gives for every $w \in V - \{t\}$:
\begin{align*}
B^v_w(G', b) &= \sum_{s \in V - \{t\}} b(s) \cdot \frac{\sigma_{s,v}(G', w)}{\sigma_{s,v}(G')}\\
             &= \sum_{s \in V - \{t\}} b(s) \cdot \frac{\sigma_{s,t}(G, w)}{\sigma_{s,t}(G)} 
             + 0 \cdot \frac{\sigma_{t,t}(G', w)}{\sigma_{t,t}(G')}\\
             &= B^t_w(G, b).
\end{align*}

\subsubsection{Symmetry:} 
Fix a path $\pi = ((v_1, v_2), \dots, (v_{k-1}, v_k))$ and let $\pi^R = ((v_k, v_{k-1}), \dots, (v_2, v_1))$.
Note that for every path $\pi$ from $s$ to $t$ in $G$ there is a path $\pi^R$ from $t$ to $s$ in $G'$.
This means that for every $s \in V$ it holds $\Pi_{t, s}(G') = \{\pi^R: \pi \in \Pi_{s, t}(G)\}$.
In particular, for every $w \in V$ we have:
    \[\sigma_{t,s}(G', w) = \sigma_{s,t}(G, w),\]
    \[\sigma_{t,s}(G') = \sigma_{s,t}(G).\]

This gives for every $w \in V$:
\begin{align*}
B^s_w(G', \mathds{1}^t) = \frac{\sigma_{t,s}(G', w)}{\sigma_{t,s}(G')}\! 
                       =\!\frac{\sigma_{s,t}(G, w)}{\sigma_{s,t}(G)} = B^t_w(G, \mathds{1}^s).
\end{align*}

\subsubsection{Direct Link Domination:}
Let $G' = G - \{(v, u)\}$.
We know deleted edge $(v, u)$ is not a part of any shortest path from any $s$ to $t$, because if it was, then the path having the same edges from $s$ to $v$ and then edge $(v, t)$ would be shorter.
This means that for every $s \in V$ it holds $\Pi_{s, t}(G') = \Pi_{s, t}(G)$.
In particular, for every $w \in V$ we have:
    \[\sigma_{s,t}(G', w) = \sigma_{s,t}(G, w),\]
    \[\sigma_{s,t}(G') = \sigma_{s,t}(G).\]

This gives for every $w \in V$:
\begin{align*}
B^t_w(G', b) &= \sum_{s \in V} b(s) \cdot \frac{\sigma_{s,t}(G', w)}{\sigma_{s,t}(G')}
           \\&= \sum_{s \in V} b(s) \cdot \frac{\sigma_{s,t}(G, w)}{\sigma_{s,t}(G)} = B^t_w(G, b).
\end{align*}

\subsubsection{Atom \texorpdfstring{$1$-$1$}{1-1}:}
In the graph $(\{s,t\}, k \cdot \{(s,t)\})$ from Atom axioms each edge constitutes a shortest path from $s$ to $t$, and both nodes are on all of these paths.
From definitions we have:
    \[B^t_s(G, \mathds{1}^s) = \frac{\sigma_{s,t}(G, s)}{\sigma_{s,t}(G)}\!=\!\frac{k}{k}\! =\!\frac{\sigma_{s,t}(G, t)}{\sigma_{s,t}(G)} = B^t_t(G, \mathds{1}^s).\]

\subsection{Axioms \texorpdfstring{$\Rightarrow$}{=>} t-Betweenness}

We will now prove that there is at most one centrality that satisfies Locality, Additivity, Node Redirect, Target Proxy, Symmetry, Direct Link Domination and Atom $1$-$1$.
Combined with the fact that t-Betweenness satisfies the axioms, we get that the axioms determine the centrality to be t-Betweenness.

First, let us introduce an auxiliary axiom, Atom $f(k)$-$f(k)$, as a generalization of both Atom $1$-$1$ and Atom $k$-$k$.
It allows us to construct a joint proof of Theorems~\ref{theorem:betweenness} and~\ref{theorem:stress}.
\begin{axiom}
\label{AtomXX}
(Atom $f(k)$-$f(k)$) There exists a function $f: \mathds{N}_+ \rightarrow \mathds{R}_{\geq 0}$, such that for every node $s$, natural number $k \in \mathds{N}_+$ and graph $G = (\{s,t\}, k \cdot \{(s,t)\})$:
$$F^t_s(G, \mathds{1}^s) = f(k) = F^t_t(G, \mathds{1}^s).$$
\end{axiom}

\begin{lemma}\label{lemma:all:anonymity}
(Anonymity) Let $G = (V,E) \in \mathcal{G}_t$ be an arbitrary graph with node $v \neq t$ and assume $v' \not \in V$. 
Let $f: V \rightarrow V - \{v\} \cup \{v'\}$ be a renaming function such that $f(v) = v'$ and $f(w) = w$ for $w \in V - \{v\}$ and let $G' = (\{f(w) : w \in V\}, \{(f(w), f(w')): (w, w') \in E\})$ and $b' = b \cdot f^{-1}$.
Now, if $F$ satisfies Node Redirect, then for every node $u \in V$: 
    \[F^{t}_{f(u)}(G', b') = F^t_u(G, b).\]
\end{lemma}

\begin{proof}
Consider a graph $G''$ obtained from $G$ by adding node $v'$ as the out-twin of node $v$ with no incoming edges and zero weight:
$G'' = (V \cup \{v'\}, E + \{(v', w): (v, w) \in E\})$ and $b''(w) = b''(w)$ for every $w \in V$, $b''(v')=0$.
Since $R_{v'\rightarrow v}(G'') = G$, from Node Redirect we know that for every $u \in V - \{v\}$:
\begin{align} 
    F^t_u(G, b) & = F^t_{u}(G'', b''),\label{eq:all:anonymity:G''_u}\\
    F^t_v(G, b) & = F^t_{v}(G'', b'') + F^t_{v'}(G'', b'').\label{eq:all:anonymity:G''_v}
\end{align}

Now, observe that $R_{v \rightarrow v'}(G'', b'') = (G', b')$.
Hence, from Node Redirect we know that for every node $u \in V - \{v\}$:
\begin{align} 
    F^t_{u}(G', b') & = F^t_u(G'', b''),\label{eq:all:anonymity:G'_u}\\
    F^t_{v'}(G', b') & = F^t_{v}(G'', b'') + F^t_{v'}(G'', b'').\label{eq:all:anonymity:G'_v}
\end{align}
Hence, from \eqref{eq:all:anonymity:G''_u} and~\eqref{eq:all:anonymity:G'_u} we get that for every $u \in V - \{v\}$ it holds $F^t_{u}(G', b') = F^t_u(G, b)$ and from \eqref{eq:all:anonymity:G''_v} and~\eqref{eq:all:anonymity:G'_v} we get $F^t_{v'}(G', b') = F^t_v(G, b)$.
\end{proof}

\begin{lemma}\label{lemma:all:t_loop}
(Target Self-Loop) If $F$ satisfies Locality, Node Redirect and Target Proxy then for every graph $G = (V, E) \in \mathcal{G}_t$ such that $(t, t) \in E$ and node $v \in V$: 
    \[F^t_v(G - \{(t, t)\}, b) = F^t_v(G, b).\]
\end{lemma}

\begin{proof}
Consider graph $(\{t\}, \{(t, t)\})$ in which node $t$ has zero weight (we will denote the zero weight function by $0$).
From Locality we know that for every node $v \in V - \{t\}$:
\begin{align*} 
F^t_v(G, b) & = F^t_v(G - \{(t, t)\}, b) \\
F^t_t(G, b) & = F^t_t(G - \{(t, t)\}, b) + F^t_t((\{t\}, \{(t, t)\}), 0).
\end{align*}
Thus, it remains to show that $F^t_t((\{t\}, \{(t, t)\}), 0) = 0$.

To this end, first consider graph $(\{t, r\}, \{(t, r)\})$ with zero weight function.
From Target Proxy we get:
\begin{equation} \label{eq:all:t_loop:TP_t}
F^t_t((\{t\}, \{(t, t)\}), 0) = F^r_t((\{t, r\}, \{(t, r)\}), 0).
\end{equation}
Now, let us split $t$ into two out-twins $t$ and $t'$, obtaining $G' = (\{t, t', r\}, \{(t, r), (t', r)\})$.
From Node Redirect we get:
\begin{equation} \label{eq:all:t_loop:NR_t}
    F^r_t((\{t, r\}, \{(t, r)\}), 0) = F^r_t(G', 0) + F^r_{t'}(G', 0).
\end{equation}
From Locality we know that:
\begin{align}
    F^r_t(G', 0) & = F^r_t((\{t, r\}, \{(t, r)\}), 0), \label{eq:all:t_loop:LOC_t}\\
    F^r_{t'}(G', 0) & = F^r_{t'}((\{t', r\}, \{(t', r)\}), 0).\label{eq:all:t_loop:LOC_t_2}
\end{align}
However, Lemma~\ref{lemma:all:anonymity} (Anonymity) implies that the right-hand sides of \eqref{eq:all:t_loop:LOC_t} and \eqref{eq:all:t_loop:LOC_t_2} are equal. 
Hence, \eqref{eq:all:t_loop:NR_t} simplifies to:
\[ F^r_t((\{t, r\}, \{(t, r)\}), 0) = 2 \cdot F^r_t((\{t, r\}, \{(t, r)\}), 0) \]
which implies $F^r_t((\{t, r\}, \{(t, r)\}), 0) = 0$ and from \eqref{eq:all:t_loop:TP_t} also $F^t_t((\{t, t\}, \{(t, t)\}), 0) = 0$. 
This concludes the proof.
\end{proof}

\begin{lemma}\label{lemma:b:no_target_outlet}
(No Target Outlet) If $F$ satisfies Locality, Node Redirect, Target Proxy and Direct Link Domination then for every graph $G = (V, E) \in \mathcal{G}_t$ such that $(t, v) \in E$ and node $u \in V$: 
    \[F^t_u(G - \{(t, v)\}, b) = F^t_u(G, b).\]
\end{lemma}

\begin{proof}
Fix $u \in V$.
If $v = t$, then the thesis follows from Lemma~\ref{lemma:all:t_loop} (Target Self-Loop).

Assume otherwise. 
From Lemma~\ref{lemma:all:t_loop} (Target Self-Loop) we know that adding a self-loop to $t$ does not change the centrality of $u$:
\begin{equation}\label{eq:b:no_target_outlet:1}
F^t_u(G, b) = F^t_u(G + \{(t, t)\}, b).
\end{equation}
Now, since $t$ has an edge to $t$ in the former graph, from Direct Link Domination we know that removing edge $(t,v)$ does not affect centralities in the graph:
\begin{equation}\label{eq:b:no_target_outlet:2}
F^t_u(G + \{(t, t)\}, b) = F^t_u(G + \{(t, t)\} - \{(t, v)\}, b).
\end{equation}
Finally, using Lemma~\ref{lemma:all:t_loop} (Target Self-Loop) again, we can delete a self-loop of $t$:
\begin{equation}\label{eq:b:no_target_outlet:3}
F^t_u(G + \{(t, t)\} - \{(t, v)\}, b) = F^t_u(G - \{(t, v)\}, b).
\end{equation}
Combining  \eqref{eq:b:no_target_outlet:1}, \eqref{eq:b:no_target_outlet:2} and \eqref{eq:b:no_target_outlet:3} concludes the proof.
\end{proof}

\begin{lemma}\label{lemma:b:dist_0}
(Distance 0) If $F$ satisfies Locality, Additivity, Node Redirect, Target Proxy and Atom $f(k)$-$f(k)$, then for every graph $G = (V,E) \in \mathcal{G}_t$ and node $v \in V - \{t\}$:
    \[F^t_v(G, \mathds{1}^t) = 0\]
and $F^t_t(G, \mathds{1}^t) = f(1)$.
\end{lemma}

\begin{proof}
Consider graph $(\{t\}, \{\})$ in which node $t$ has weight~$1$.
From Locality we know that for every $v \in V - \{t\}$:
\begin{align*}
F^t_v(G, \mathds{1}^t) & = F^t_v(G, 0), \\
F^t_t(G, \mathds{1}^t) & = F^t_t(G, 0) + F^t_t((\{t\}, \{\}), \mathds{1}^t).
\end{align*}
Now, in a graph with zero weights all centralities are equal zero which follows from Additivity, as for every $u \in V$ it holds:
\[ F^t_u(G,0) = F^t_u(G,0) + F^t_u(G,0). \]
Thus, it remains to prove that $F^t_t((\{t\}, \{\}), \mathds{1}^t)$ is determined.

To this end, let us consider graph $(\{t, r\}, \{(t, r)\})$ with weight $\mathds{1}^t$.
From Atom $f(k)$-$f(k)$ we know that: 
\begin{equation}\label{eq:b:dist_0:1}
F^r_t((\{t, r\}, \{(t, r)\}), \mathds{1}^t) = f(1).
\end{equation}
Now, from Target Proxy:
\begin{equation}\label{eq:b:dist_0:2}
F^t_t((\{t, r\}, \{(t, r)\}), \mathds{1}^t) = F^t_t((\{t\}, \{(t, t)\}), \mathds{1}^t).
\end{equation}
However, from Lemma~\ref{lemma:all:t_loop} (Target Self-Loop) we know that removing self-loop of $t$ does not affect its centrality. Hence:
\begin{equation}\label{eq:b:dist_0:3}
F^t_t((\{t\}, \{(t, t)\}), \mathds{1}^t) = F^t_t((\{t\}, \{\}), \mathds{1}^t).
\end{equation}
By combining \eqref{eq:b:dist_0:1}, \eqref{eq:b:dist_0:2} and \eqref{eq:b:dist_0:3} we get the thesis.
\end{proof}

\begin{lemma}\label{lemma:b:dist_1}
(Distance 1) If $F$ satisfies Locality, Additivity, Node Redirect, Direct Link Domination and Atom $f(k)$-$f(k)$, then for every graph $G = (V, E) \in \mathcal{G}_t$, source node $s \neq t$, such that $(s, t) \in E$ and node $v \in V - \{s, t\}$: 
    \[F^t_v(G, \mathds{1}^s) = 0,\]
and $F^t_s(G, \mathds{1}^s) = F^t_t(G, \mathds{1}^s) = f(m_{(s,t)}(E))$.
\end{lemma}                                                              
\begin{proof}
Let $k = m_{(s,t)}(G)$.
Consider a graph obtained from $G$ by deleting all outgoing edges of $s$ other than $k$ edges to $t$: $G_1 = (V_1, E_1) = (V, E - \{(s, u) : (s, u) \in E \wedge u \neq t\})$.
From Direct Link Domination for every node $v \in V$ we know that:
\begin{equation} \label{eq:b:dist_1:DLD_v}
    F^t_v(G, \mathds{1}^s) = F^t_v(G_1, \mathds{1}^s).
\end{equation}

Let us split $s$ into two out-twins: $s$ with all incoming edges and zero weight and $s'$ with no incoming edges and unit weight; the resulting graph is $G_2 = (V_2, E_2) = (V \cup \{s'\}, E_1 + k \cdot \{(s', t)\})$.
We have $R_{s' \rightarrow s}(G_2) = G_1$.
Hence, from Node Redirect for every node $v \in V - \{s\}$ we know that:
\begin{align} 
    F^t_v(G_1, \mathds{1}^s) & = F^t_v(G_2, \mathds{1}^{s'}),\label{eq:b:dist_1:NR_v}\\
    F^t_s(G_1, \mathds{1}^s) & = F^t_s(G_2, \mathds{1}^{s'}) + F^t_{s'}(G_2, \mathds{1}^{s'}).\label{eq:b:dist_1:NR_s}
\end{align}

Note that $(G_2,\mathds{1}^{s'}) = (G_1, 0) + ((\{s',t\}, k \cdot \{(s',t)\}), \mathds{1}^{s'})$.
In graph $(G_1, 0)$ all centralities are equal zero, as Additivity implies $F^t_v(G_1, 0) = F^t_v(G_1, 0) + F^t_v(G_1, 0)$ for every $v \in V$.
Hence, from Locality, we get that for every node $v \in V - \{t\}$:
\begin{align} 
F^t_v(G_2, \mathds{1}^{s'}) & = 0,\label{eq:b:dist_1:LOC_v}\\
F^t_{s'}(G_2, \mathds{1}^{s'}) & = F^t_{s'}((\{s', t\}, k \cdot \{(s', t)\}), \mathds{1}^{s'}) \!=\! f(k)\label{eq:b:dist_1:LOC_s'},\\
F^t_{t}(G_2, \mathds{1}^{s'}) & = F^t_{t}((\{s', t\}, k \cdot \{(s', t)\}), \mathds{1}^{s'}) \!=\! f(k).\label{eq:b:dist_1:LOC_t}
\end{align}
Here, we used Atom $f(k)$-$f(k)$ to argue about centralities in graph $(\{s', t\}, k \cdot \{(s', t)\})$.
Now, by combining \eqref{eq:b:dist_1:DLD_v}, \eqref{eq:b:dist_1:NR_v} and~\eqref{eq:b:dist_1:LOC_v} we get $F^t_v(G, \mathds{1}^s) = 0$ for every $v \in V - \{s, t\}$.
Furthermore, from \eqref{eq:b:dist_1:DLD_v}, \eqref{eq:b:dist_1:NR_s},~\eqref{eq:b:dist_1:LOC_v} for $s$ and~\eqref{eq:b:dist_1:LOC_s'} we get that $F^t_s(G, \mathds{1}^s) = f(k)$ and from \eqref{eq:b:dist_1:DLD_v}, \eqref{eq:b:dist_1:NR_v} and~\eqref{eq:b:dist_1:LOC_t} we get that $F^t_t(G, \mathds{1}^s) = f(k)$.
\end{proof}

\begin{lemma}\label{lemma:b:dist_2_simple}
(Distance 2 Simple) If $F$ satisfies Locality, Additivity, Node Redirect, Target Proxy, Symmetry, Direct Link Domination and Atom $f(k)$-$f(k)$, then for every graph $G = (V, E) \in \mathcal{G}_t$ with source node $s \neq t$, such that $(s, t) \notin E$, which satisfies following assumptions:
\begin{itemize}
    \item for every node $u \in S^1_s(G)$ it holds $\Gamma^-_u(G) = \{(s, u)\}$,
    \item for every node $u \in P^1_t(G)$ it holds $\Gamma^+_u(G) = \{(u, t)\}$,
    \item $W = S^1_s(G) \cap P^1_t(G)$ is not empty,
\end{itemize}
for every node $v \in V - W - \{s, t\}$: 
    \[F^t_v(G, \mathds{1}^s) = 0,\]
for every node $w \in W$: $F^t_w(G, \mathds{1}^s) = f(|W|)/|W|$ and $F^t_s(G,\mathds{1}^s) = F^t_t(G,\mathds{1}^s) = f(|W|)$.
\end{lemma}  

\begin{proof}
For a graph $G = (V, E)$ with weights $b$ and set of nodes $U \subseteq V$ let $F^t_U(G, b) = \sum_{v \in U} F^t_v(G, b)$.
We will assume that $t$ has no self-loops, as from Lemma~\ref{lemma:all:t_loop} (Target Self-Loop) we know that removing them does not affect centralities in a graph.

Let $S = S^1_s(G) - \{s\}$ and $T = P^1_t(G)$ (note $W = T \cap S$).
Consider a graph obtained from $G$ by merging all nodes from $T$ into one node, $t'$: $G_1 = (V_1, E_1) = (V -T \cup \{t'\}, E - \bigcup_{u \in T}\Gamma_u(G) + \{(t', t)\} + \{(w, t') : (w, u) \in E \wedge u \in T\})$.
From Node Redirect we know that for every $v \in V - T$: 
\begin{align} 
    F^t_v(G, \mathds{1}^s) & = F^t_v(G_1, \mathds{1}^s)\label{eq:b:dist_2_simple:NR_v},\\
    F^t_T(G, \mathds{1}^s) & = F^t_{t'}(G_1, \mathds{1}^s)\label{eq:b:dist_2_simple:NR_T}.
\end{align}
Notice that in $G_1$ the number of edges from $s$ to $t'$ is equal $|T|$.
We can also see that the target $t$ has only one predecessor, $t'$.
Consider a graph obtained from $G_1$ by merging $t$ into $t'$: $G_2 = (V_2, E_2) = M_{t \rightarrow t'}(G_1)$.
From Target Proxy we know that for every $v \in V_2$:
\begin{equation} \label{eq:b:dist_2_simple:TR_v}
    F^t_v(G_1, \mathds{1}^s) = F^{t'}_v(G_2, \mathds{1}^s).
\end{equation}
In graph $G_2$ the distance from the source $s$ to the target $t'$ is equal one.
Moreover, the number of edges from $s$ to $t'$ in $G_2$ is the same as in $G_1$, so it is equal $|W|$.
Hence, from Lemma~\ref{lemma:b:dist_1} (Distance 1) we get that for every $v \in V_2 - \{s, t'\}$:
\begin{align} 
    F^{t'}_v(G_2, \mathds{1}^s) & = 0,\label{eq:b:dist_2_simple:dist_1_v}\\
    F^{t'}_s(G_2, \mathds{1}^s) & = F^{t'}_{t'}(G_2, \mathds{1}^s) = f(|W|).\label{eq:b:dist_2_simple:dist_1_s}
\end{align}
Combining \eqref{eq:b:dist_2_simple:NR_v},~\eqref{eq:b:dist_2_simple:TR_v} and~\eqref{eq:b:dist_2_simple:dist_1_v} we know that for every $v \in V - T - \{s, t\}$:
\begin{equation} \label{eq:b:dist_2_simple:zero_v}
    F^t_v(G, \mathds{1}^s) = 0
\end{equation}
Combining \eqref{eq:b:dist_2_simple:NR_v},~\eqref{eq:b:dist_2_simple:TR_v} and~\eqref{eq:b:dist_2_simple:dist_1_s} we know that:
\begin{equation} \label{eq:b:dist_2_simple:s}
    F^t_s(G, \mathds{1}^s) = f(|W|).
\end{equation}
Finally, combining \eqref{eq:b:dist_2_simple:NR_T},~\eqref{eq:b:dist_2_simple:TR_v} and~\eqref{eq:b:dist_2_simple:dist_1_s} we get that:
\begin{equation} \label{eq:b:dist_2_simple:T}
    F^t_T(G, \mathds{1}^s) = f(|W|).
\end{equation}
Thus, it remains to prove that $F^t_t(G,\mathds{1}^s)=f(|W|)$ and determine how nodes from $T$ share their join centrality $F^t_T(G, \mathds{1}^s) = f(|W|)$ among themselves.

To this end, let us now go back to graph $G$ and for each node $u \in V - S$ add an edge $(t, u)$, obtaining graph $G_3 = (V, E_3) = (V, E + \{(t, u) : u \in V - S\})$.
From Lemma~\ref{lemma:b:no_target_outlet} (No Target Outlet) we know it does not affect centralities in the graph: i.e., for every $v \in V$ we have $F^t_v(G, \mathds{1}^s) = F^t_v(G_3, \mathds{1}^s)$.
In $G_3$ not only from $s$ there is an edge to each node from $S$ and a path to $t$, but also from $t$ there is an edge to every node $u \in V - S$; hence, from $s$ there is a path to every node.
This means that the reversed graph $G_4 = (V, E_4) = (V, \{(u, w) : (w, u) \in E_3\})$ is in the class $\mathcal{G}_s$. 
From Symmetry we know that centralities in graphs $G_3$ and $G_4$ are the same: for every $v \in V$: $F^t_v(G_3, \mathds{1}^s) = F^s_v(G_4, \mathds{1}^t)$. 
As a result, we get that for every $v \in V$:
\begin{equation} \label{eq:b:dist_2_simple:SYM}
    F^t_v(G, \mathds{1}^s) = F^s_v(G_4, \mathds{1}^t).
\end{equation}

Now, we proceed as before.
We merge all nodes from $S$ into one node, $s'$, obtaining $G_5 = (V_5, E_5) = (V -S \cup \{s'\}, E_4 - \bigcup_{u \in S}\Gamma_u(G) + \{(s', s)\} + \{(w, s') : (w, u) \in E_4 \wedge u \in S\})$.
From Node Redirect we know that for every $v \in V - S$:
\begin{equation}\label{eq:b:dist_2_simple:SYM_NR_v}
    F^s_v(G_4, \mathds{1}^t) = F^s_v(G_5, \mathds{1}^t).
\end{equation}
Now target $s$ has only one predecessor, $s'$.
We merge $s$ into $s'$ obtaining $G_6 = (V_6, E_6) = M_{s \rightarrow s'}(G_5)$.
From Target Proxy we know that for every $v \in V_6$:
\begin{equation} \label{eq:b:dist_2_simple:SYM_TR_v}
    F^s_v(G_5, \mathds{1}^t) = F^{s'}_v(G_6, \mathds{1}^t).
\end{equation}
We obtained a graph in which distance from the source $t$ to the target $s'$ is equal one.
Hence, from Lemma~\ref{lemma:b:dist_1} (Distance 1) we know that for every $v \in V_6 - \{s', t\}$:
\begin{align}  
    F^{s'}_v(G_6, \mathds{1}^t) & = 0,\label{eq:b:dist_2_simple:SYM_dist_1_v}\\
    F^{s'}_t(G_6, \mathds{1}^t) & = f(|W|).\label{eq:b:dist_2_simple:SYM_dist_1_t}
\end{align}
Combining \eqref{eq:b:dist_2_simple:SYM},~\eqref{eq:b:dist_2_simple:SYM_NR_v},~\eqref{eq:b:dist_2_simple:SYM_TR_v} and~\eqref{eq:b:dist_2_simple:SYM_dist_1_v} (or, respectively,~\eqref{eq:b:dist_2_simple:SYM_dist_1_t}) we know that for every $v \in V - S - \{s, t\}$:
\begin{align} 
    F^t_v(G, \mathds{1}^s) & = 0,\label{eq:b:dist_2_simple:SYM_zero_v}\\
    F^t_t(G, \mathds{1}^s) & = f(|W|).\label{eq:b:dist_2_simple:t}
\end{align}

Thus, it remains to determine centralities of nodes from $T$.
Recall that from \eqref{eq:b:dist_2_simple:T} we know that $F^t_T(G, \mathds{1}^s) = f(|W|)$ and by the definition $W = T \cap S$.
Also, from \eqref{eq:b:dist_2_simple:SYM_zero_v} we have $F^t_{T-S}(G, \mathds{1}^s) = F^t_{T-W}(G, \mathds{1}^s) = 0$.
Hence,
\[ F^t_W(G, \mathds{1}^s) = F^t_T(G, \mathds{1}^s) - F^t_{T-W}(G, \mathds{1}^s) = f(|W|). \]
But each of the nodes in $W$ has exactly one incoming edge, from $s$, and one outgoing edge, to $t$.
Hence, from Lemma~\ref{lemma:all:anonymity} (Anonymity) they all have the same centralities.
This means that for every node $w \in W$ it holds $F^t_w(G, \mathds{1}^s) = f(|W|)/|W|$, which concludes the proof.
\end{proof}

\begin{lemma}\label{lemma:b:dist_2}
(Distance 2) If $\bar{F},\hat{F}$ satisfy Locality, Additivity, Node Redirect, Target Proxy, Symmetry, Direct Link Domination and Atom $f(k)$-$f(k)$, then for every graph $G = (V,E) \in \mathcal{G}_t$ with the source node $s$ such that $dist_{s, t}(G) = 2$ it holds $\bar{F}^t_v(G, \mathds{1}^s) = \hat{F}^t_v(G,\mathds{1}^s)$ for every $v \in V$.
\end{lemma}                                                              
\begin{proof}
Let $F$ be any centrality that satisfies Locality, Additivity, Node Redirect, Target Proxy, Symmetry, Direct Link Domination and Atom $f(k)$-$f(k)$.
We will show that centralities (according to $F$) in $G$ are uniquely defined based on the centralities in a simpler graph that satisfies assumption of Lemma~\ref{lemma:b:dist_2_simple}.
Since, we know that $\bar{F}$ and $\hat{F}$ satisfies these axioms and from Lemma~\ref{lemma:b:dist_2_simple} that they agree on this simpler graph, this will imply they agree also on graph $G$.

We will assume $t$ has no self-loops in $G$, as from Lemma~\ref{lemma:all:t_loop} (Target Self-Loop) we know that deleting self-loops of $t$ does not change centralities in the graph.

Let $S = S^1_s-\{s\} = \{u_1,\dots,u_k\}$ be the set of successors of $s$ (other than $s$ itself) and let $k_i = m_{(s, u_i)}(E)$ for every $u_i$.
We begin by splitting each node $u_i$ into $k_i+1$ out-twins.
Specifically, for each $u_i$, we delete all its incoming edges from $s$ and add out-twins $u_i^1, ..., u_i^{k_i}$ such that each $u_i^j$ has only one incoming edge $(s,u_i^j)$ and zero weight.
Formally, consider a graph $G_1 = (V_1,E_1)$ defined as follows:
\begin{align*} 
V_1 & = V \cup \{u_i^j : u_i \in S, 1 \le j \le k_i\},\\
E_1 & = E \!-\! \{(s,\!u_i) : u_i \!\in\! S\} \!+\! \{(s,\!u_i^j) : u_i \!\in\! S, 1 \!\le\! j \!\le\! k_i\} \\
& + \{(u_i^j,w) : u_i \in S, (u_i,w) \in E, 1 \le j \le k_i\}.
\end{align*}
From Node Redirect we know that for every $v \in V - S$:
\begin{equation} \label{eq:b:dist_2:NR_v}
    F^t_v(G, \mathds{1}^s) = F^t_v(G_1, \mathds{1}^s)
\end{equation}
and for every $u_i \in S$:
\begin{equation} \label{eq:b:dist_2:NR_u_i}
F^t_{u_i}(G, \mathds{1}^s) = F^t_{u_i}(G_1, \mathds{1}^s) + \sum_{j=1}^{k_i} F^t_{u_i^j}(G_1, \mathds{1}^s).
\end{equation}

In graph $G_1$, every successor of $s$ (apart from possibly $s$), has exactly one incoming edge: the one from $s$.
Note also that because $dist_{s, t}(G) = 2$ there were some $u_i \in S$ such that $(u_i, t) \in E$.
This means that there exists some $u_i^j \in V_1$ such that $(u_i^j, t) \in E_1$, hence $dist_{s,t}(G_1) = 2$.

Let us now for each node $u \in V - S^1_s(G_1)$ add an edge $(t, u)$, obtaining graph $G_2 = (V_1, E_2) = (V_1, E_1 + \{(t, u) : u \in V - S^1_s(G_1)\})$.
From Lemma~\ref{lemma:b:no_target_outlet} (No Target Outlet) we know that for every $v \in V_1$: $F^t_v(G_1, \mathds{1}^s) = F^t_v(G_2, \mathds{1}^s)$.
In $G_2$ not only from $s$ there is an edge to each node from $S^1_s(G_1)$ and a path to $t$, but also from $t$ there is an edge to every node $u \in V - S^1_s(G_1)$; this implies that from $s$ there is a path to every node.
This means that the reversed graph $G_3 = (V_1, E_3) = (V_1, \{(u, w) : (w, u) \in E_2\})$ is in the class $\mathcal{G}_s$. 
From Symmetry we know that for every $v \in V_1$: $F^t_v(G_2, \mathds{1}^s) = F^s_v(G_3, \mathds{1}^t)$. We get that for every $v \in V_1$:
\begin{equation} \label{eq:b:dist_2:SYM}
    F^t_v(G_1, \mathds{1}^s) = F^s_v(G_3, \mathds{1}^t)
\end{equation}
Note that in graph $G_3$ every predecessor of $s$ (apart from possibly $s$), has exactly one outgoing edge: the one to $s$.

Now, let us assume $s$ has no self-loops, as we could delete them and from Lemma~\ref{lemma:all:t_loop} (Target Self-Loop) we know that centralities would not change, because $s$ is now the target. 
Hence, under this assumption, every predecessor of $s$ in $G_3$ has only one outgoing edge: to $s$.

Let us now proceed similarly as before, by splitting each successor $w_i \in S^1_t(G_3)$ of $t$ into out-twins.
Specifically, for each successor $w_i$ of $t$ we delete all its incoming edges from $t$ and add out-twins $w_i^1, ..., w_i^{l_i}$ for $l_i = m_{(t, w_i)}(E_3)$ such that each $w_i^j$ has only one incoming edge $(t,w_i^j)$ and zero weight. 
Let us denote the obtained graph by $G_4 = (V_4, E_4)$.
From Node Redirect we know that for every $v \in V_1 - S^1_t(G_3)$: 
\begin{equation} \label{eq:b:dist_2:SYM_NR_v}
    F^s_v(G_3, \mathds{1}^t) = F^s_v(G_4, \mathds{1}^t)
\end{equation}
and for every $w_i \in S^1_t(G_3)$:
\begin{equation}\label{eq:b:dist_2:SYM_NR_u_i}
F^s_{w_i}(G_3, \mathds{1}^t) = F^s_{w_i}(G_4, \mathds{1}^t) + \sum_{j=1}^{l_i} F^s_{w_i^j}(G_4, \mathds{1}^t).
\end{equation}
In graph $G_4$, every successor of $t$ has exactly one incoming edge: the one from $t$.

As a result, in graph $G_4 = (V_4, E_4)$ every successor of the source $t$ has exactly one incoming edge, from $t$ and every predecessor of the target $s$ has exactly one outgoing edge, to $s$.
Moreover, $s \neq t$, $(t, s) \notin E_4$ and $W = S^1_t(G_4) \cap P^1_s(G_4)$ is not empty.
Hence, graph $G_4$ satisfies the assumptions of Lemma~\ref{lemma:b:dist_2_simple} which implies for every node $v \in V_4$:
\begin{equation} \label{eq:b:dist_2:final}
\bar{F}^t_v(G_4,\mathds{1}^s) = \hat{F}^t_v(G_4,\mathds{1}^s).
\end{equation}
Now, from \eqref{eq:b:dist_2:final} combined with \eqref{eq:b:dist_2:SYM_NR_v} and \eqref{eq:b:dist_2:SYM_NR_u_i} applied for both $\bar{F}$ and $\hat{F}$ we get that for every $v \in V_1$:
\[ \bar{F}^t_v(G_3,\mathds{1}^s) = \hat{F}^t_v(G_3,\mathds{1}^s) \]
Finally, this combined with \eqref{eq:b:dist_2:NR_v}, \eqref{eq:b:dist_2:NR_u_i} and \eqref{eq:b:dist_2:SYM} applied for both $\bar{F}$ and $\hat{F}$ implies that for every $v \in V$:
\[ \bar{F}^t_v(G,\mathds{1}^s) = \hat{F}^t_v(G,\mathds{1}^s). \]
This concludes the proof.
\end{proof}

\begin{lemma}\label{lemma:b:single_source}
(Single Source) If $\bar{F}, \hat{F}$ satisfy Locality, Additivity, Node Redirect, Target Proxy, Symmetry, Direct Link Domination and Atom $f(k)$-$f(k)$, then for every graph $G = (V,E) \in \mathcal{G}_t$ with the source node $s$ it holds $\bar{F}^t_v(G, \mathds{1}^s) = \hat{F}^t_v(G,\mathds{1}^s)$ for every $v \in V$.
\end{lemma}

\begin{proof}
We use induction on the distance from $s$ to $t$: $dist_{s, t}(G)$. If $dist_{s, t}(G) < 3$, then the thesis follows from Lemma~\ref{lemma:b:dist_0} (Distance 0), Lemma~\ref{lemma:b:dist_1} (Distance 1) and Lemma~\ref{lemma:b:dist_2} (Distance 2).

Let us discuss the inductive step.
Let $F$ be any centrality that satisfies Locality, Additivity, Node Redirect, Target Proxy, Symmetry, Direct Link Domination and Atom $f(k)$-$f(k)$.
Fix a graph $G = (V, E)$ and a node $v \in V$.
We know that $dist_{s, v}(G) + dist_{v, t}(G) \geq dist_{s, t}(G) \geq 3$, so either distance from the source node to $v$ or from $v$ to the target is greater or equal $2$: $dist_{s, v}(G) \geq 2$ or $dist_{v, t}(G) \geq 2$.

First, we argue that the later case can be reduced to the former one.
Assume $dist_{v, t}(G) \geq 2$.
For each node $u \in V$ we add an edge $(t, u)$, obtaining graph $G' = (V, E') = (V, E + \{(t, u) : u \in V\})$.
From Lemma~\ref{lemma:b:no_target_outlet} (No Target Outlet) this does not affect centralities in the graph.
In $G'$ we know not only that from $s$ there is a path to $t$, but also from $t$ there is an edge to every node $u$, so from $s$ there is a path to $u$.
This means that the reversed graph $G'' = (V, E'') = (V, \{(u, w) : (w, u) \in E'\})$ is in the class $\mathcal{G}_s$. 
From Symmetry we know that the centralities of any node in graphs $G'$ and $G''$ are equal.
Now, in graph $G''$ the distance between the source and the target is the same, but the distance from the source to the fixed node $v$ is greater or equal two.
Hence, in what follows we assume $dist_{s, v}(G) \geq 2$.

Let $S = S^1_s-\{s\} = \{u_1,\dots,u_k\}$ be the set of successors of $s$ (other than $s$ itself) and let $k_i = m_{(s, u_i)}(E)$ for every $u_i$.
We begin by splitting each node $u_i$ into $k_i+1$ out-twins.
Specifically, for each $u_i$, we delete all its incoming edges from $s$ and add out-twins $u_i^1, ..., u_i^{k_i}$ such that each $u_i^j$ has only one incoming edge $(s,u_i^j)$ and zero weight.
Formally, consider a graph $G_1 = (V_1,E_1)$ defined as follows:
\begin{align*} 
V_1 & = V \cup \{u_i^j : u_i \in S, 1 \le j \le k_i\},\\
E_1 & = E \!-\! \{(s,\!u_i) : u_i \!\in\! S\} \!+\! \{(s,\!u_i^j) : u_i \!\in\! S, 1 \!\le\! j \!\le\! k_i\} \\
& + \{(u_i^j,w) : u_i \in S, (u_i,w) \in E, 1 \le j \le k_i\}.
\end{align*}
From Node Redirect we know that:
\begin{equation} \label{eq:b:single_source:NR}
    F^t_v(G, \mathds{1}^s) = F^t_v(G_1, \mathds{1}^s),
\end{equation}
because $v \notin S$.
What is important, in $G_1$ every successor of $s$ (apart from possibly $s$), has exactly one incoming edge, the one from $s$.

Let us now for each node $u \in V - S^1_s(G_1)$ add an edge $(t, u)$, obtaining graph $G_2 = (V_1, E_2) = (V_1, E_1 + \{(t, u) : u \in V - S^1_s(G_1)\})$.
From Lemma~\ref{lemma:b:no_target_outlet} (No Target Outlet) we know that $F^t_v(G_1, \mathds{1}^s) = F^t_v(G_2, \mathds{1}^s)$.
In $G_2$ we know not only that from $s$ there is an edge to each node from $S^1_s(G_1)$ and a path to $t$, but also from $t$ there is an edge to every node $u \in V - S^1_s(G_1)$; hence, from $s$ there is a path to every node.
This means that the reversed graph $G_3 = (V_1, E_3) = (V_1, \{(u, w) : (w, u) \in E_2\})$ is in the class $\mathcal{G}_s$. 
From Symmetry we know that $F^t_v(G_2, \mathds{1}^s) = F^s_v(G_3, \mathds{1}^t)$.
Hence, we get:
\begin{equation} \label{eq:b:single_source:SYM}
    F^t_v(G_1, \mathds{1}^s) = F^s_v(G_3, \mathds{1}^t)
\end{equation}

Now, we delete self-loops of $s$, obtaining $G_4 = (V_4, E_4) = (V_1, E_3 - \{m_{(s, s)}(E_3) \cdot (s, s)\}$. 
Node $s$ is now the target, so from Lemma~\ref{lemma:all:t_loop} (Target Self-Loop) we Get:
\begin{equation} \label{eq:b:single_source:SYM_t_loop}
    F^s_v(G_3, \mathds{1}^t) = F^s_v(G_4, \mathds{1}^t).
\end{equation}
Note that in $G_4$ every predecessor of $s$ has only one outgoing edge, to $s$.
Now, we merge them into one node, $s'$, obtaining $G_5 = (V_5, E_5) = (V_4 - P^1_s(G_4) \cup \{s'\}, E_4 - \Gamma^-_s(G_4) + \{(s', s)\} + \{(w, s') : (w, u), (u, s) \in E\})$.
From Node Redirect we know that:
\begin{equation} \label{eq:b:single_source:SYM_NR}
    F^s_v(G_4, \mathds{1}^t) = F^s_v(G_5, \mathds{1}^t).
\end{equation}

Finally, $s$ has only one incoming edge in $G_5$, from $s'$ (and $s'$ has only one outgoing edge, to $s)$.
Let us merge $s$ into $s'$, obtaining $G_6 = M_{s \rightarrow s'}(G_5, \mathds{1}^t)$.
From Target Proxy we know that:
\begin{equation} \label{eq:b:single_source:SYM_TR}
    F^s_v(G_5, \mathds{1}^t) = F^{s'}_v(G_6, \mathds{1}^t).
\end{equation}

Now, observe that $dist_{s, t}(G) = dist_{t, s'}(G_6) + 1$.
Hence, from the inductive assumption we know that $\bar{F}^{s'}_v(G_6, \mathds{1}^t)=\hat{F}^{s'}_v(G_6, \mathds{1}^t)$.
This combined with \eqref{eq:b:single_source:NR}-\eqref{eq:b:single_source:SYM_TR} applied for both $\bar{F}$ and $\hat{F}$ shows that also $\bar{F}^t_v(G, \mathds{1}^s) = \hat{F}^t_v(G, \mathds{1}^s)$ for every $v \in V$ which concludes the proof.
\end{proof}

\begin{lemma}\label{lemma:b:determined}
There is at most one centrality measure that satisfies Locality, Additivity, Node Redirect, Target Proxy, Symmetry, Direct Link Domination and Atom $f(k)$-$f(k)$.
\end{lemma}

\begin{proof}
Let $\bar{F},\hat{F}$ be two centrality measures that satisfies Locality, Additivity, Node Redirect, Target Proxy, Symmetry, Direct Link Domination and Atom $f(k)$-$f(k)$.
Fix a graph $G = (V, E)$ with weights $b$ and a node $v \in V$. 
We will prove that $\bar{F}^t_v(G,b) = \hat{F}^t_v(G,b)$.

If $b = 0$, then the centrality of $v$ is equal zero for both centralities $\bar{F}, \hat{F}$, as it holds $F^t_v(G, 0) = F^t_v(G, 0) + F^t_v(G, 0)$ for any $F$ that satisfies Additivity.

Otherwise, let us decompose weight function into non-empty sum of functions positive for only one node: $b = \sum_{s \in V} b(s) \cdot \mathds{1}^s$.
From Additivity for $F \in \{\bar{F}, \hat{F}\}$ we know that:
\begin{equation} \label{eq:b:determined:ADD}
    F^t_v(G, b)  = \sum_{s \in V} F^t_v(G, b(s) \cdot \mathds{1}^s).
\end{equation}
Moreover, from Additivity function $F^t_v(G, b(s) \cdot \mathds{1}^s)$ is additive with respect to $b(s)$ and from the definition it is non-negative, so we know it is linear, that is:
\begin{equation} \label{eq:b:determined:linear}
    F^t_v(G, b(s) \cdot \mathds{1}^s) = b(s) \cdot F^t_v(G, \mathds{1}^s).
\end{equation}
Now, Lemma~\ref{lemma:b:single_source} (Single Source) implies that:
\[ \bar{F}^t_v(G, \mathds{1}^s) = \hat{F}^t_v(G,\mathds{1}^s). \]
This combined with \eqref{eq:b:determined:ADD} and~\eqref{eq:b:determined:linear} applied for both $\bar{F}$ and $\hat{F}$ shows that $\bar{F}^t_v(G, b) = \hat{F}^t_v(G,b)$ which concludes the proof.
\end{proof}

\begin{lemma}\label{lemma:axioms_to_betweenness}
If $F$ satisfies Locality, Additivity, Node Redirect, Target Proxy, Symmetry, Direct Link Domination and Atom $1$-$1$, then it is t-Betweenness Centrality. 
\end{lemma}

\begin{proof}
Because $F$ satisfies Atom $1$-$1$, it satisfies Atom $f(k)$-$f(k)$, for $f(k) = 1$.
From Lemma~\ref{lemma:b:determined} we know that $F$ is uniquely determined.
But from Lemma~\ref{lemma:b_to_axioms} we know that t-Betweenness satisfies these axioms, so $F$ must be t-Betweenness Centrality.
\end{proof}

\section{Proof of Theorem~\ref{theorem:stress}}

In this section, we present the proof of Theorem~\ref{theorem:stress}.
We start by showing that t-Stress Centrality satisfies the axioms.
In the main part, we prove that the axioms uniquely characterize a centrality measure.
The proof in analogous to the proof of Theorem~\ref{theorem:betweenness}.

\subsection{t-Stress Centrality \texorpdfstring{$\Rightarrow$}{=>} Axioms}
We will now consider each of the axioms: Locality, Additivity, Node Redirect, Target Proxy, Symmetry, Direct Link Domination, Atom $k$-$k$ and show that t-Stress Centrality satisfies it.
The proofs are analogous to the proofs from Lemma~\ref{lemma:b_to_axioms}.

\begin{lemma}\label{lemma:s_to_axioms}
t-Stress Centrality satisfies Locality, Additivity, Node Redirect, Target Proxy, Symmetry, Direct Link Domination and Atom $k$-$k$.
\end{lemma}

Let $\Pi_{v, u}(G)$ denote the set of shortest paths from $v$ to $u$ in $G$ (in particular, $\sigma_{v, u}(G) = |\Pi_{v, u}(G)|$).

\subsubsection{Locality:} 
Let us notice that for every $s \in V$ no shortest path from $s$ to $t$ in $G+G'$ includes any edge from $E'$. 
Let us assume such a path exists and let $(v', u')$ be the first edge from $E$' in this path.
We have $v' \in V'$, but -- since up to this point the path was in $G$ -- also $v' \in V$, so $v' = t$.
Hence, the path is not a shortest path to $t$.
This means that for every $s \in V$ it holds $\Pi_{s, t}(G + G') = \Pi_{s, t}(G)$.
In particular, for every $w \in V$, $w' \in V' - \{t\}$ we have:
    \[\sigma_{s,t}(G + G', w) = \sigma_{s,t}(G, w),\]
    \[\sigma_{s,t}(G + G', w') = 0,\]
    \[\sigma_{s,t}(G + G') = \sigma_{s,t}(G).\]

This gives:
\begin{align*}
&S^t_w((G, b) + (G', b'))\\
&= \sum_{s \in V \cup V'} (b+b')(s) \cdot \sigma_{s,t}(G+G', w)\\
&= \sum_{s \in V - \{t\}} b(s) \cdot \sigma_{s,t}(G, w)
   +(b(t) + b'(t)) \cdot 0\\  
&+ \sum_{s' \in V' - \{t\}} b'(s') \cdot 0
 = S^t_w(G, b).
\end{align*}
    
As for the centrality of $t$, by noting that the empty path is always the only shortest path from $t$ to $t$ we get:
\begin{align*}
&S^t_t((G, b) + (G', b'))\\
&= \sum_{s \in V \cup V'} (b+b')(s) \cdot \sigma_{s,t}(G+G', t)\\
&= \sum_{s \in V - \{t\}} b(s) \cdot \sigma_{s,t}(G, t)
   +b(t) + b'(t)\\  
&+ \sum_{s' \in V' - \{t\}} b'(s') \cdot \sigma_{s',t}(G', t)\\
&= S^t_t(G, b) + S^t_t(G', b').
\end{align*}

\subsubsection{Additivity:} 
From the definition we have:
\begin{align*}
&S^t_w(G, b + b') = \sum_{s \in V} (b(s) + b'(s)) \cdot \sigma_{s,t}(G, w)\\
&= \sum_{s \in V} b(s) \cdot \sigma_{s,t}(G, w) + 
   \sum_{s \in V} b'(s) \cdot \sigma_{s,t}(G, w)\\ 
&= S^t_w(G, b) + S^t_w(G, b').
\end{align*}

\subsubsection{Node Redirect:} 
Let $(G', b') = R_{u \rightarrow v}(G, b)$.
We know shortest paths from $v$ to $t$ are the same in $G$ as in $G'$, as $u$, the out-twin of $v$, cannot be on a shortest path to $t$ in $G$.
Hence, $\Pi_{v, t}(G') = \Pi_{v, t}(G)$.
Moreover, since $u$ is the out-twin of $v$, they have the same shortest paths to $t$ in $G$, up to appropriately changing first edge, that is $\Pi_{v, t}(G') = \{((v, w_1), ... (w_k, t)) : ((u, w_1), ... (w_k, t)) \in \Pi_{u, t}(G)\}$.
In particular, for every $w \in V$ we have:
    \[\sigma_{v, t}(G', w) = \sigma_{v,t}(G, w) = \sigma_{u,t}(G, w),\]
    \[\sigma_{v, t}(G') = \sigma_{v,t}(G) = \sigma_{u,t}(G).\] 

Now, fix $s \in V - \{v, u\}$.
Clearly, shortest paths from $s$ to $t$ that do not go through $v$ or $u$ are the same in $G$ as in $G'$.
Consider paths that go through $v$ or $u$. 
Each such a path goes through one of the incoming edges of $v$ or $u$.
For every path from $s$ to $t$ in $G$ that goes through some edge to $v$ there is a path from $s$ to $t$ in $G'$ that also goes through this edge.
For every path from $s$ to $t$ in $G$ that goes through an edge $(w,u)$ to $u$ there is a path from $s$ to $t$ in $G'$ that goes through the corresponding edge $(w, v)$ and through node $v$.
And similarly, for every path from $s$ to $t$ in $G'$ that goes through $(w,v)$ to $v$ there is a corresponding path from $s$ to $t$ that is either the same, or the same apart from going through $(w, u)$ instead.
Let $R_{u \rightarrow v}(((w_1, w_2), ..., (w_i, u), (u, w_{i+2}), ..., (w_{k-1}, w_{k}))) = ((w_1, w_2), \dots, (w_i, v), (v, w_{i+2}), \dots, (w_{k-1}, w_{k})$.
We get $\Pi_{v, t}(G') = \Pi_{v, t}(G) + \{R_{u \rightarrow v}(\pi) : \pi \in \Pi_{u, t}(G)\}$.
In particular, for every $w \in V - \{v, v'\}$ we have:
    \[\sigma_{s, t}(G', w) = \sigma_{s,t}(G, w),\]
    \[\sigma_{s, t}(G', v) = \sigma_{s,t}(G, v) + \sigma_{s,t}(G, u),\]
    \[\sigma_{s, t}(G') = \sigma_{s,t}(G).\]

This gives:
\begin{align*}
&S^t_v(G', b') \\
&= \sum_{s \in V - \{v, u\}} b'(s) \sigma_{s,t}(G', v)
    + b'(v) \sigma_{v,t}(G', v)\\
&= \sum_{s \in V - \{v, u\}} b(s) (\sigma_{s,t}(G, v) + \sigma_{s,t}(G, u))\\ 
&+ (b(v) + b(u)) \sigma_{v,t}(G, v)\\
&= S^t_v(G, b) + S^t_u(G, b).
\end{align*}

and for every $w \in V - \{v, u\}$:
\begin{align*}
&S^t_w(G', b') \\
&= \sum_{s \in V - \{v, u\}} b'(s) \sigma_{s,t}(G', w) 
    + b'(v) \sigma_{v,t}(G', w)\\
&= \sum_{s \in V - \{v, u\}} b(s) \sigma_{s,t}(G, w) 
    + (b(v) + b(u)) \sigma_{v,t}(G, w)\\
&= S^t_w(G, b).
\end{align*}

\subsubsection{Target Proxy:}
Let $G' = M_{t \rightarrow v}(G, b)$.
Let us notice for every $s \in V - \{t\}$ every path from $s$ to $t$ has $(v, t)$ as the last edge, because it is the only edge ending in $t$.
Hence, for every shortest path from $s$ to $t$ in $G$ there is a path from $s$ to $v$ in $G'$ that lacks only the last edge and for every shortest path from $s$ to $v$ in $G'$ there is a path from $s$ to $t$ in $G$ that is the same apart from the additional edge $(v, t)$ at the end.
This means that for every $s \in V$ it holds $\Pi_{s, v}(G') = \{(e_1, \dots, e_{k-1}): (e_1, \dots, e_{k-1}, e_k) \in \Pi_{s, t}(G)\}$.
In particular, for every $w \in V - \{t\}$ we have:
    \[\sigma_{s, v}(G', w) = \sigma_{s,t}(G, w),\]
    \[\sigma_{s, v}(G') = \sigma_{s,t}(G).\]

This gives for every $w \in V - \{t\}$:
\begin{align*}
S^v_w(G', b) &= \sum_{s \in V - \{t\}} b(s) \cdot \sigma_{s,v}(G', w)\\
             &= \sum_{s \in V - \{t\}} b(s) \cdot \sigma_{s,t}(G, w) 
             + 0 \cdot \sigma_{t,t}(G', w)\\
             &= S^t_w(G, b).
\end{align*}

\subsubsection{Symmetry:} 
Fix a path $\pi = ((v_1, v_2), \dots, (v_{k-1}, v_k))$ and let $\pi^R = ((v_k, v_{k-1}), \dots, (v_2, v_1))$.
Note that for every path $\pi$ from $s$ to $t$ in $G$ there is a path $\pi^R$ from $t$ to $s$ in $G'$.
This means that for every $s \in V$ it holds $\Pi_{t, s}(G') = \{\pi^R: \pi \in \Pi_{s, t}(G)\}$.
In particular, for every $w \in V$ we have:
    \[\sigma_{t,s}(G', w) = \sigma_{s,t}(G, w),\]
    \[\sigma_{t,s}(G') = \sigma_{s,t}(G).\]

This gives for every $w \in V$:
\begin{align*}
S^s_w(G', \mathds{1}^t) = \sigma_{t,s}(G', w)
                        = \sigma_{s,t}(G, w) = S^t_w(G, \mathds{1}^s).
\end{align*}

\subsubsection{Direct Link Domination:}
Let $G' = G - \{(v, u)\}$.
We know deleted edge $(v, u)$ is not a part of any shortest path from any $s$ to $t$, because if it was, then the path having the same edges from $s$ to $v$ and then edge $(v, t)$ would be shorter.
This means that for every $s \in V$ it holds $\Pi_{s, t}(G') = \Pi_{s, t}(G)$.
In particular, for every $w \in V$ we have:
    \[\sigma_{s,t}(G', w) = \sigma_{s,t}(G, w),\]
    \[\sigma_{s,t}(G') = \sigma_{s,t}(G).\]

This gives for every $w \in V$:
\begin{align*}
S^t_w(G', b) &= \sum_{s \in V} b(s) \cdot \sigma_{s,t}(G', w) 
           \\&= \sum_{s \in V} b(s) \cdot \sigma_{s,t}(G, w) = S^t_w(G, b).
\end{align*}

\subsubsection{Atom \texorpdfstring{$k$-$k$}{k-k}:} 
In the graph $(\{s,t\}, k \cdot \{(s,t)\})$ from Atom axioms each edge constitutes a shortest path from $s$ to $t$, and both nodes are on all of these paths.
From definitions we have:
    \[S^t_s(G, \mathds{1}^s) = \sigma_{s,t}(G, s) = k = \sigma_{s,t}(G, t) = S^t_t(G, \mathds{1}^s).\]

\subsection{Axioms \texorpdfstring{$\Rightarrow$}{=>} t-Stress}
We will now prove that there is at most one centrality that satisfies Locality, Additivity, Node Redirect, Target Proxy, Symmetry, Direct Link Domination and Atom $k$-$k$.
The proof is joint with the proof of Lemma~\ref{lemma:axioms_to_betweenness}.

\begin{lemma}\label{lemma:axioms_to_stress}
If $F$ satisfies Locality, Additivity, Node Redirect, Target Proxy, Symmetry, Direct Link Domination and Atom $k$-$k$, then it is t-Stress Centrality. 
\end{lemma}

\begin{proof}
Because $F$ satisfies Atom $k$-$k$, it satisfies Atom $f(k)$-$f(k)$, for $f(k) = k$.
From Lemma~\ref{lemma:b:determined} we know that $F$ is uniquely determined.
But from Lemma~\ref{lemma:s_to_axioms} we know that t-Stress satisfies these axioms, so this is our determined centrality.
\end{proof}

\section{Proof of Theorem~\ref{theorem:random_walk_betweenness}}
In this section, we present the full proof of Theorem~\ref{theorem:random_walk_betweenness}.
We start with showing that the Random Walk Betweenness satisfies the axioms.
In the main part, we prove that the axioms uniquely characterize a centrality measure.

\subsection{t-Random Walk Betweenness \texorpdfstring{$\Rightarrow$}{=>} Axioms}
We will now consider each of the axioms: Locality, Additivity, Node Redirect, Target Proxy, Edge Swap and Edge Multiplication, Atom $1$-$1$ and show that t-Random Walk Betweenness satisfies it.

\begin{lemma}\label{lemma:rwb_to_axioms}
t-Random Walk Betweenness satisfies Locality, Node Redirect, Target Proxy, Edge Swap, Edge Multiplication and Atom $1$-$1$.
\end{lemma}

Let us define a shorthand $G_t = G - \Gamma^+_t(G)$.
\subsubsection{Locality:}
Fix node $s \in V$.
Notice that for every $v \in V - \{t\}$ outgoing edges of $v$ in graphs $G_t$ and $(G+G')_t$ are the same.
Hence, for every $v, u \in V$ the random walk starting in $s$ will have the same probability of transition from $v$ to $u$ in graphs $G_t$ and $(G+G')_t$.
Also, for every $v \in V - \{t\}$ and $u' \in V' - \{t\}$ the probability of transition from $v$ to $u'$ will be zero, because there are no edges from $v$ to $u'$.
This means that for every $u \in V$, $u' \in V' - \{t\}$ and $k \in \mathds{N}_{\geq 0}$:
    \[\mathds{P}(\omega_{(G+G')_t,s}(k) \!=\! u) = \mathds{P}(\omega_{G_t,s}(k) \!=\! u),\]
    \[\mathds{P}(\omega_{(G+G')_t,s}(k) \!=\! u') = 0.\]
Similarly, we get that for every $s \in V'$, $u \in V$ and $k \in \mathds{N}_{\geq 0}$:
    \[\mathds{P}(\omega_{(G+G')_t,s'}(k) \!=\! u) = 0.\]

This gives for every $w \in V - \{t\}$:
\begin{align*}
&RWB^t_w(G+G', b+b')\\ 
&= \sum_{s \in V - \{t\}} \sum_{k=0}^{\infty} b(s) \cdot \mathds{P}(\omega_{(G+G')_t,s}(k) \!=\! w)\\
&+ \sum_{k=0}^{\infty} (b(t) + b'(t)) \cdot 0\\
&+ \sum_{s' \in V' - \{t\}} \sum_{k=0}^{\infty} b'(s') \cdot \mathds{P}(\omega_{(G+G')_t,s'}(k) \!=\! w)\\
&= \sum_{s \in V - \{t\}} \sum_{k=0}^{\infty} b(s) \cdot \mathds{P}(\omega_{G_t,s}(k) \!=\! w) \\
&+ \sum_{k=0}^{\infty} b(t) \cdot 0\\
&= RWB^t_w(G, b)\\
\end{align*}
and:
\begin{align*}
&RWB^t_t(G+G', b+b')\\ 
&= \sum_{s \in V - \{t\}} \sum_{k=0}^{\infty} b(s) \cdot \mathds{P}(\omega_{(G+G')_t,s}(k) \!=\! t)\\
&+ (b(t) + b'(t)) \\
&+ \sum_{s' \in V' - \{t\}} \sum_{k=0}^{\infty} b'(s') \cdot \mathds{P}(\omega_{(G+G')_t,s'}(k) \!=\! t)\\
&= \sum_{s \in V - \{t\}} \sum_{k=0}^{\infty} b(s) \cdot \mathds{P}(\omega_{G_t,s}(k) \!=\! t) + b(t)\\
&+ \sum_{s' \in V' - \{t\}} \sum_{k=0}^{\infty} b'(s') \cdot \mathds{P}(\omega_{G'_t,s'}(k) \!=\! t) + b'(t)\\
&= RWB^t_t(G, b) + RWB^t_t(G', b').\\
\end{align*}

\subsubsection{Additivity:}
From the definition we have:

\begin{align*}
&RWB^t_w(G, b+b')\\ 
&= \sum_{s \in V} \sum_{k=0}^{\infty} (b(s) + b'(s)) \cdot \mathds{P}(\omega_{G_t,s}(k) \!=\! w)\\
&= \sum_{s \in V} \sum_{k=0}^{\infty} b(s) \cdot \mathds{P}(\omega_{G_t,s}(k) \!=\! w)\\
&+ \sum_{s \in V} \sum_{k=0}^{\infty} b'(s) \cdot \mathds{P}(\omega_{G_t,s}(k) \!=\! w)\\
&= RWB^t_w(G, b) + RWB^t_w(G, b').\\
\end{align*}

\subsubsection{Node Redirect:}
For every graph $G=(V,E)$ and nodes $v, u \in V$ let us denote by $p_{v,u}((V, E)) = \frac{m_{(v,u)}(E)}{|\Gamma^+_v(G)|}$ the probability of the direct transition from node $v$ to node $u$ of the random walk on $G$.

Fix a graph $G=(V,E)$ and out-twins $v, u \in V$.
For every $w \in V$ we have:
    \[p_{u,w}(G_t) = p_{v,w}(G_t).\]
    
Let $(G', b') = R_{u \rightarrow v}(G, b)$ and $G' = (V', E')$.
We will now examine transition probabilities in $G'_t$ for every nodes $w, r \in V' - \{v\}$: from $v$ to $v$, from $w$ to $v$, from $v$ to $w$ and finally from $w$ to $r$.
First, outgoing edges of $v$ are the same in $G_t$ and $G'_t$, so we know that:     \[p_{v,v}(G'_t) = p_{v,v}(G_t).\]
Next, because each outgoing edge of $w$ ending in $v$ in $G'_t$ has corresponding edge ending in $v$ or $u$ in $G_t$ we know that:
    \[p_{w,v}(G'_t) = p_{w,v}(G_t) + p_{w,u}(G_t).\]
Then, the number of edges $(v, w)$ is the same in $G_t$ and $G'_t$, as well as out-degree of $v$ is the same in $G_t$ and $G'_t$, because outgoing edges of $v$ are the same.
Hence we know that:
    \[p_{v, w}(G'_t) = p_{v, w}(G_t).\]
Finally, number of edges $(w, r)$ is the same in $G_t$ and $G'_t$, because only edges incident to $v$ and $u$ are different, as well as out-degree of $w$ is the same in $G_t$ and $G'_t$.
Hence, we know that:
    \[p_{w, r}(G'_t) = p_{w, r}(G_t).\]

Now we want to consider the whole random walk.
Firstly, for every $w \in V - \{v, u\}$ and every $k \in \mathds{N}_{\geq 0}$ we know that:
    \[\mathds{P}(\omega_{G_t,u}(k) \!=\! w) = \mathds{P}(\omega_{G_t,v}(k) \!=\! w),\]
because $v$ and $u$ are out-twins.
Moreover:
\begin{multline*}
    \mathds{P}(\omega_{G_t,u}(k) \!=\! v) + \mathds{P}(\omega_{G_t,u}(k) \!=\! u)  = \\
    \mathds{P}(\omega_{G_t,v}(k) \!=\! v) + \mathds{P}(\omega_{G_t,v}(k) \!=\! u),
\end{multline*}
because in step zero both sides are equal one, and in further steps the random walks starting from $v$ and from $u$ behave the same.

We would like to prove for every $s \in V'$, $w \in V' - \{v\} $ and $k \in \mathds{N}_{\geq 0}$ that $\mathds{P}(\omega_{G'_t,s}(k) \!=\! w) = \mathds{P}(\omega_{G_t,s}(k) \!=\! w)$ and $\mathds{P}(\omega_{G'_t,s}(k) \!=\! v) = \mathds{P}(\omega_{G_t,s}(k) \!=\! v) + \mathds{P}(\omega_{G_t,s}(k) \!=\! u)$.
The proof will be inductive.
The basis of the induction is $k = 0$:
    \[\mathds{P}(\omega_{G'_t,s}(0) \!=\! w) = [s \!=\! w] = \mathds{P}(\omega_{G_t,s}(0) \!=\! w),\] 
\begin{multline*}
    \mathds{P}(\omega_{G'_t,s}(0) \!=\! v) = [s\!=\!v] =\\
    \mathds{P}(\omega_{G_t,s}(0) \!=\! v) + \mathds{P}(\omega_{G_t,s}(0) \!=\! u),
\end{multline*}
because $s \neq u$.

Let us show the inductive step.
\begin{align*}
&\mathds{P}(\omega_{G'_t,s}(k+1) \!=\! w)\\ 
&= \sum_{r \in V - \{v, u\}} p_{r,w}(G'_t) \cdot \mathds{P}(\omega_{G'_t,s}(k) \!=\! r) \\
&+ p_{v,w}(G'_t) \cdot \mathds{P}(\omega_{G'_t,s}(k) \!=\! v)\\ 
&= \sum_{r \in V - \{v, u\}} p_{r,w}(G_t) \cdot \mathds{P}(\omega_{G_t,s}(k) \!=\! r) \\
&+ p_{v,w}(G_t) \cdot (\mathds{P}(\omega_{G_t,s}(k) \!=\! v) + \mathds{P}(\omega_{G_t,s}(k) \!=\! u)) \\
&= \mathds{P}(\omega_{G_t,s}(k+1) \!=\! w)
\end{align*}
and
\begin{align*}
&\mathds{P}(\omega_{G'_t,s}(k+1) \!=\! v)\\ 
&= \sum_{r \in V - \{v, u\}} p_{r,v}(G'_t) \cdot \mathds{P}(\omega_{G'_t,s}(k) \!=\! r) \\
&+ p_{v,v}(G'_t) \cdot \mathds{P}(\omega_{G'_t,s}(k) \!=\! v)\\ 
&= \sum_{r \in V - \{v, u\}} (p_{r,v}(G_t) + p_{r,u}(G_t)) \cdot \mathds{P}(\omega_{G_t,s}(k) \!=\! r) \\
&+ p_{v,v}(G_t) \cdot (\mathds{P}(\omega_{G_t,s}(k) \!=\! v) + \mathds{P}(\omega_{G_t,s}(k) \!=\! u))\\ 
&= \mathds{P}(\omega_{G_t,s}(k+1) \!=\! v) + \mathds{P}(\omega_{G_t,s}(k+1) \!=\! u).
\end{align*}

Knowing how the random walk behaves, we deduce that:

\begin{align*}
&RWB^t_w(G', b')\\ 
&= \sum_{s \in V - \{v\}} \sum_{k=0}^{\infty} b(s) \cdot \mathds{P}(\omega_{G'_t,s}(k) \!=\! w)\\
&+ \sum_{k=0}^{\infty} (b(v) + b(u)) \cdot \mathds{P}(\omega_{G'_t,v}(k) \!=\! w)\\
&= \sum_{s \in V - \{v\}} \sum_{k=0}^{\infty} b(s) \cdot \mathds{P}(\omega_{G_t,s}(k) \!=\! w)\\
&+ \sum_{k=0}^{\infty} b(v) \!\cdot\! \mathds{P}(\omega_{G_t,v}(k) \!=\! w)
    + \sum_{k=0}^{\infty} b(u) \!\cdot\! \mathds{P}(\omega_{G_t,u}(k) \!=\! w)\\
&= RWB^t_w(G, b)\\
\end{align*}
and
\begin{align*}
&RWB^t_v(G', b')\\ 
&= \sum_{s \in V - \{v\}} \sum_{k=0}^{\infty} b(s) \cdot \mathds{P}(\omega_{G'_t,s}(k) \!=\! v)\\
&+ \sum_{k=0}^{\infty} (b(v) + b(u)) \cdot \mathds{P}(\omega_{G'_t,v}(k) \!=\! v)\\
&= \sum_{s \in V - \{v\}} \sum_{k=0}^{\infty} b(s) \!\cdot\! (\mathds{P}(\omega_{G_t,s}(k) \!=\! v) + \mathds{P}(\omega_{G_t,s}(k) \!=\! u))\\
&+ \sum_{k=0}^{\infty} b(v) \!\cdot\! (\mathds{P}(\omega_{G_t,v}(k) \!=\! v) + \mathds{P}(\omega_{G_t,v}(k) \!=\! u))\\ 
&+ \sum_{k=0}^{\infty} b(u) \!\cdot\! (\mathds{P}(\omega_{G_t,u}(k) \!=\! v) + \mathds{P}(\omega_{G_t,u}(k) \!=\! u))\\
&= RWB^t_v(G, b) + RWB^t_u(G, b).\\
\end{align*}

\subsubsection{Target Proxy:}
Let $(G', b) = M_{t \rightarrow v}(G, b)$.
Notice that the transition probability from $v$ to every node $w \in V - \{t\}$ is equal to zero in $G_t$, because the only outgoing edges from $v$ leads to $t$.
It is also equal to zero in $G'_v$, because outgoing edges of target $v$ are deleted.
In particular, the transition probability from $v$ to $w$ is the same in $G_t$ and $G'_v$.
Outgoing edges of every node $w \in V - \{v, t\}$ are the same in $G_t$ and $G'_v$, so the transition probability from $w$ to every node $r \in V - \{t\}$ is also the same in $G_t$ and $G'_v$.
This means that for every $s, w \in V - \{t\}$ and every $k \in \mathds{N}_{\geq 0}$:
    \[\mathds{P}(\omega_{G'_v,s}(k) \!=\! w) = \mathds{P}(\omega_{G_t,s}(k) \!=\! w).\]

This gives:
\begin{align*}
    &RWB^{v}_w(G', b) = \sum_{s \in V - \{t\}} \sum_{k=0}^{\infty} b(s) \cdot \mathds{P}(\omega_{G'_v,s}(k) \!=\! w)\\
    &= \sum_{s \in V - \{t\}} \sum_{k=0}^{\infty} b(s) \cdot \mathds{P}(\omega_{G_t,s}(k) \!=\! w)\\ 
    &+ \sum_{k=0}^{\infty} 0 \cdot \mathds{P}(\omega_{G_t,t}(k) \!=\! w)\\
    &= RWB^t_w(G, b).\\
\end{align*}

\subsubsection{Edge Swap:}
We will need an alternative definition of t-Random Walk Betweenness. Namely, t-Random Walk Betweenness is the unique function satisfying for every graph $G = (V, E) \in \mathcal{G}_t$ with weights $b$ a system of equations, one for every node $w \in V$:
    \[RWB^t_w(G, b) = 
    b(w) + \sum_{(r, w) \in \Gamma^-_w(G_t)} \frac{RWB^t_r(G_t, b)}{\Gamma^+_u(G_t)}.\]

We obtain this equivalent definition from the fact that our t-Random Walk Betweenness is the special case of \citet{Bloechl:etal:2011} betweenness-like centrality.

We will show that $RWB^t(G, b)$ values satisfy also the system of equations for $RWB^t(G', b)$.
Since it has only one solution, this will imply that
$RWB^t(G, b) = RWB^t(G', b)$.

Note that the out-degree of every node is the same in $G$ and $G'$.
Moreover, incoming edges of every $w \in V - \{v', u'\}$ are the same in $G$ and $G'$.
Hence, for every $w \in V - \{v', u'\}$ we have its $RWB^t(G', b)$ equation:

\begin{align*}
RWB^t_w(G', b) 
& =  b(w) + \sum_{(r, w) \in \Gamma^-_w(G'_t)} \frac{RWB^t_r(G'_t, b)}{\Gamma^+_u(G'_t)} \\
& = b(w) + \sum_{(r, w) \in \Gamma^-_w(G_t)} \frac{RWB^t_r(G'_t, b)}{\Gamma^+_u(G_t)}.
\end{align*}

This equation is satisfied by $RWB^t(G, b)$ values, because the $RWB^t(G, b)$ equation for $w$ is satisfied.

Now let us analyze $RWB^t(G', b)$ equation for $v'$ (equation for $u'$ is analogous). We use the assumption that $v$ and $u$ have equal out-degrees and t-Random Walk Betweenness values in $(G, b)$:

\begin{align*}
&RWB^t_{v'}(G', b)\\
&= b(v') 
    + \frac{RWB^t_{u}(G'_t, b)}{\Gamma^+_{u}(G'_t)} \\
&+ \sum_{(r, v') \in \Gamma^-_{v'}(G'_t) - \{(u, v')\}} \frac{RWB^t_r(G'_t, b)}{\Gamma^+_r(G'_t)}\\
&= b(v') 
    + \frac{RWB^t_{v}(G'_t, b)}{\Gamma^+_{v}(G_t)} \\
&+ \sum_{(r, v') \in \Gamma^-_{v'}(G_t) - \{(v, v')\}} \frac{RWB^t_r(G'_t, b)}{\Gamma^+_r(G_t)}.
\end{align*}

This equation is satisfied by $RWB^t(G, b)$ values, because the $RWB^t(G, b)$ equation for $v'$ is satisfied.

\subsubsection{Edge Multiplication:}
Let $G' = (V, E') = G + k \cdot \Gamma^+_v(G)$.
Notice that for every node $u \in V$ and every $k \in \mathds{N}_{\geq 0}$ the probability of the random walk being in $u$ in the step $k+1$ provided that it was in $v$ in the step $k$ is the same in $G'$ as in $G$.
Specifically, for every $s \in V$:

\begin{align*}
&\mathds{P}(\omega_{G'_t,s}(k+1) \!=\! u | \omega_{G'_t,s}(k) \!=\! v)
= \frac{k \cdot m_{(v, u)}(E)}{k \cdot |\Gamma^+_v(G)|} \\ 
&= \frac{m_{(v, u)}(E')}{|\Gamma^+_v(G')|}
= \mathds{P}(\omega_{G_t,s}(k+1) \!=\! u | \omega_{G_t,s}(k) \!=\! v).
\end{align*}

This means that the transition probability between every two nodes are the same in $G'$ and in $G$.
Hence:

\begin{align*}
RWB^{t}_w(G', b) = RWB^{t}_w(G, b).
\end{align*}

\subsubsection{Atom:}
We know that the random walk starting in $s$ is in $s$ in step zero and in $t$ in step one, where it ends its travel.
Hence:

\begin{align*}
&RWB^t_s(G, \mathds{1}^s) 
= \sum_{k=0}^{\infty} \mathds{P}(\omega_{G_t,s}(k) \!=\! s)\\
&= \mathds{P}(\omega_{G_t,s}(0) \!=\! s)
= 1
\end{align*}
and:
\begin{align*}
&RWB^t_t(G, \mathds{1}^s) 
= \sum_{k=0}^{\infty} \mathds{P}(\omega_{G_t,s}(k) \!=\! t)\\
&= \mathds{P}(\omega_{G_t,s}(1) \!=\! t)
= 1.
\end{align*}

\subsection{Axioms \texorpdfstring{$\Rightarrow$}{=>} t-Random Walk Betweenness}

We will now prove that there is at most one centrality that satisfies Locality, Node Redirect, Target Proxy, Edge Swap, Edge Multiplication and Atom $1$-$1$.
Combined with the fact that t-Random Walk Betweenness satisfies these axioms, we get that the axioms determine the centrality to be t-Random Walk Betweenness.

\begin{lemma}\label{lemma:pr:one_arrow}
(1-Arrow Graph) If $F$ satisfies Locality, Node Redirect, Target Proxy and Atom $1$-$1$ then for every graph $G = (\{s, t\}, \{(s, t)\})$ and $x \in \mathds{R}_{\geq 0}$:
    \[F^t_s(G, x \cdot \mathds{1}^s) = F^t_t(G, x \cdot \mathds{1}^s) = x.\]
\end{lemma}

\begin{proof}
Consider graph $G_1 = (\{s, s', t\}, \{(s, t), (s', t)\})$.
From Node Redirect we have that for every $y, z \in \mathds{R}_{\geq 0}$:
\begin{multline*}
    F^t_s(G, (y + z) \cdot \mathds{1}^s) =\\
    F^t_s(G_1, y \cdot \mathds{1}^s + z \cdot \mathds{1}^{s'}) + F^t_{s'}(G_1, y \cdot \mathds{1}^s + z \cdot \mathds{1}^{s'})
\end{multline*}
and:
    \[F^t_t(G, (y + z) \cdot \mathds{1}^s) = F^t_t(G_1, y \cdot \mathds{1}^s + z \cdot \mathds{1}^{s'}).\] 
Now, from Locality we know for every $y, z \in \mathds{R}_{\geq 0}$ it holds:
\begin{align*}
    F^t_s(G_1, y \cdot \mathds{1}^s + z \cdot \mathds{1}^{s'}) & = F^t_s(G, y \cdot \mathds{1}^s), \\
    F^t_{s'}(G_1, y \cdot \mathds{1}^s + z \cdot \mathds{1}^{s'}) & = F^t_{s'}((\{s', t\}, \{(s', t)\}), z \cdot \mathds{1}^{s'})
\end{align*}
and
\begin{multline*}
    F^t_t(G_1, y \cdot \mathds{1}^s + z \cdot \mathds{1}^{s'}) =\\
    F^t_t(G, y \cdot \mathds{1}^s) + F^t_{t}((\{s', t\}, \{(s', t)\}), z \cdot \mathds{1}^{s'}).
\end{multline*}
However, Lemma~\ref{lemma:all:anonymity} (Anonymity) implies that:
\begin{align*}
F^t_{s'}((\{s', t\}, \{(s', t)\}, z \cdot \mathds{1}^{s'}) & = F^t_{s}(G, z \cdot \mathds{1}^s),\\
F^t_{t}((\{s', t\}, \{(s', t)\}, z \cdot \mathds{1}^{s'}) & = F^t_{t}(G, z \cdot \mathds{1}^s).
\end{align*}
To summarize, we now know that for $v \in \{s, t\}$:
\begin{equation*} \label{eq:pr:one_arrow:ADD}
    F^t_v(G, (y + z) \cdot \mathds{1}^s) \!=\! F^t_v(G, y \cdot \mathds{1}^s) \!+\! F^t_v(G, z \cdot \mathds{1}^s),
\end{equation*}
that is, the function $F^t_v(G, x \cdot \mathds{1}^s)$ is additive with respect to $x$. 
From the definition it is non-negative so we know it is linear, that is:
\begin{equation} \label{eq:pr:one_arrow:LIN}
    F^t_v(G, x \cdot \mathds{1}^s) = x \cdot F^t_v(G, \mathds{1}^s).
\end{equation}
From Atom $1$-$1$ and~\eqref{eq:pr:one_arrow:LIN} we know that $F^t_s(G, x \cdot \mathds{1}^s) = x$ and $F^t_t(G, x \cdot \mathds{1}^s) = x$, which concludes the proof.
\end{proof}

\begin{lemma}\label{lemma:pr:almost_one_arrow}
(Almost 1-Arrow Graph) If $F$ satisfies Locality, Node Redirect, Target Proxy, Edge Multiplication and Atom $1$-$1$ then for every graph $G = (\{s, t\}, E)$, such that $E - \Gamma^+_t(G) = \{(s, t)\}$ and $x \in \mathds{R}_{\geq 0}$:
    \[F^t_s(G, x \cdot \mathds{1}^s) = F^t_t(G, x \cdot \mathds{1}^s) = x.\]
\end{lemma}

\begin{proof}
We will assume $t$ has no self-loops in $G$, as from Lemma~\ref{lemma:all:t_loop} (Target Self-Loop) we know that deleting self-loops of $t$ does not change centralities in the graph.
Let $k = m_{(t,s)}(E)$, that is $G = (\{s, t\}, \{(s, t), k \cdot (t, s)\})$.

First let us assume $x = 0$.
Consider graph $G' = (\{s, t\}, \{(s, t), 2k \cdot (t, s)\})$.
From Edge Multiplication we know that for $v \in \{s, t\}$: $F^t_v(G, 0) = F^t_v(G', 0)$.
Now let us split $s$ into two out-twins $s$ and $s'$, each with half of the incoming edges, obtaining:
\[ G'' = (\{s, s', t\}, \{(s, t), k \cdot (t, s), (s', t), k \cdot (t, s')\}). \]
From Node Redirect we know that:
\begin{align*} 
    F^t_s(G', 0) &= F^t_s(G'', 0) + F^t_{s'}(G'', 0)\\
    F^t_t(G', 0) &= F^t_t(G'', 0)
\end{align*}
From Locality and Lemma~\ref{lemma:all:anonymity} (Anonymity) we have that:
\begin{multline*} 
F^t_v((\{s, t\}, \{(s, t), 2k \cdot (t, s)\}), 0) =\\
F^t_v((\{s, t\}, \{(s, t), k \cdot (t, s)\}), 0)\\
+ F^t_v((\{s, t\}, \{(s, t), k \cdot (t, s)\}), 0).
\end{multline*}
This means:
\begin{equation} \label{eq:pr:almost_one_arrow:zero}
    F^t_v(G, 0) = 0,
\end{equation}
which concludes the proof for $x=0$.

In the general case we again split $s$ into out-twins $s$ with the original weight and no incoming edges and $s'$ with zero weight and all the incoming edges, separate them and rename $s'$ to $s$.
From Node Redirect, Locality and Lemma~\ref{lemma:all:anonymity} (Anonymity) we have that:
\begin{equation} \label{eq:pr:almost_one_arrow:sum}
    F^t_v(G, x \!\cdot\! \mathds{1}^s) \!=\! F^t_v(G, 0) \!+\! F^t_v((\{s, t\}, \{(s, t)\}), x \!\cdot\! \mathds{1}^s).
\end{equation}

To conclude the proof, in~\eqref{eq:pr:almost_one_arrow:sum} we substitute the first term with 0 (based on \eqref{eq:pr:almost_one_arrow:zero}) and the second with $x$ (based on Lemma~\ref{lemma:pr:one_arrow}).
\end{proof}

\begin{lemma}\label{lemma:pr:k_arrow_adjoin}
(k-Arrow Adjoin) If $F$ satisfies Locality, Node Redirect, Target Proxy, Edge Multiplication and Atom $1$-$1$, then for every graph $G = (V, E) \in \mathcal{G}_t$, new node $w \notin V$ with weight $x$, numbers $k, l \in \mathds{Z}_{\geq 0}$ and node $v \in V - \{t\}$, let $G' = (V \cup \{w\}, E + k \cdot \{(w, t)\} + l \cdot \{(t, w)\})$ and $b' = b + x \cdot \mathds{1}^w$:
    \[F^t_v(G', b') = F^t_v(G, b),\]
$F^t_w(G', b') = x$ and $F^t_t(G', b') = F^t_t(G, b) + x$.
\end{lemma}

\begin{proof}
Consider $G_1 = (\{w, t\}, \{(w, t), l \cdot (t, w)\})$.
From Lemma~\ref{lemma:pr:almost_one_arrow} (Almost 1-Arrow Graph) we know that:
    \[F^t_w(G_1, x \cdot \mathds{1}^w) = F^t_t(G_1, x \cdot \mathds{1}^w) = x.\]
Now we multiply the edge from $w$ to $t$ by $k$, obtaining $G_2 = (\{w, t\}, k \cdot \{(w, t)\} + l \cdot \{(t, w)\})$.
From Edge Multiplication we know that:
\begin{align*}
    F^t_w(G_2, x \cdot \mathds{1}^w) & = F^t_w(G_1, x \cdot \mathds{1}^w),\\ 
    F^t_t(G_2, x \cdot \mathds{1}^w) & = F^t_t(G_1, x \cdot \mathds{1}^w).
\end{align*}
Finally, joining $G$ and $G_2$ we obtain $G'$.
From Locality we know that for $v \in V - \{t\}$:
\begin{align*}
    F^t_v(G', b') & = F^t_v(G, b)\\
    F^t_t(G', b') & = F^t_t(G, b) + F^t_t(G_2, x \cdot \mathds{1}^w).
\end{align*}
Combining above equalities concludes the proof.
\end{proof}

\begin{lemma}\label{lemma:pr:siphon}
(Siphon) If $F$ satisfies Locality, Node Redirect, Target Proxy, Edge Swap, Edge Multiplication and Atom $1$-$1$, then for every graph $G \in \mathcal{G}_t$, and node $v$, such that $\Gamma^-_v(G) = \emptyset$:
    \[F^t_v(G, b) = b(v).\]
\end{lemma}

\begin{proof}
First, let us assume $v = t$, that is $G = (\{t\}, \{\})$, $b = x \cdot \mathds{1}^t$.
Consider graph $G_1 = (\{t, r\}, \{(t, r)\})$.
From Lemma~\ref{lemma:pr:one_arrow} (1-Arrow Graph) we know that:
    \[F^r_t(G_1, x \cdot \mathds{1}^t) = x.\]
We merge $r$ into $t$ obtaining $G_2 = (\{t\}, \{(t, t)\})$.
From Target Proxy we know that:
    \[F^r_t(G_1, x \cdot \mathds{1}^t) = F^t_t(G_2, x \cdot \mathds{1}^t).\]
But from Lemma~\ref{lemma:all:t_loop} (Target Self-Loop) we also know that:         \[F^t_t(G_2, x \cdot \mathds{1}^t) = F^t_t((\{t\}, \{\}), x \cdot \mathds{1}^t).\]
Combining above equations results in:
\begin{equation} \label{eq:pr:siphon:isolated_t}
    F^t_t((\{t\}, \{\}), x \cdot \mathds{1}^t) = x.
\end{equation}
This concludes the proof of the case $v = t$.

Let us now assume $v \neq t$.
Let $G = (V, E)$. Consider a graph obtained from $G$ by adding a new node $w$ with weight $F^t_v(G, b)$ and out-degree $|\Gamma^+_v(G)|$: $G_1 = (V \cup \{w\}, E + \{|\Gamma^+_v(G)| \cdot (w, t)\})$, $b_1 = b + F^t_v(G, b) \cdot \mathds{1}^w$.
From Lemma~\ref{lemma:pr:k_arrow_adjoin} (k-Arrow Adjoin) we know that:
\begin{align}
    F^t_v(G_1, b_1) & = F^t_v(G, b),\label{eq:pr:siphon:k_Arrow_v}\\
    F^t_w(G_1, b_1) & = b_1(w).\label{eq:pr:siphon:k_Arrow_w}
\end{align}
From~\eqref{eq:pr:siphon:k_Arrow_v}, ~\eqref{eq:pr:siphon:k_Arrow_w} and the definition of $b_1$ we know that $v$ and $w$ have equal centralities in $(G_1, b_1)$.
Swapping all edges outgoing from $v$ and from $w$ we obtain graph $G_2 = (V \cup \{w\}, E - \Gamma^+_v(G) + \{(w, u) : (v, u) \in E\}+ \{|\Gamma^+_v(G)| \cdot (v, t)\})$.
From Edge Swap we know that:
\begin{equation} \label{eq:pr:siphon:ES_v}
    F^t_v(G_1, b_1) = F^t_v(G_2, b_1).
\end{equation}
From Lemma~\ref{lemma:pr:k_arrow_adjoin} (k-Arrow Adjoin) we know that:
\begin{equation}
    F^t_v(G_2, b_1) = b_1(v).\label{eq:pr:siphon:k_Arrow_2_v}
\end{equation}
Combining \eqref{eq:pr:siphon:k_Arrow_v},~\eqref{eq:pr:siphon:ES_v} and~\eqref{eq:pr:siphon:k_Arrow_v} we know that $F^t_v(G_2, b_1) = b(v)$.
Together with~\eqref{eq:pr:siphon:isolated_t} it concludes the proof.
\end{proof}

\begin{lemma}\label{lemma:pr:siphon_split}
(Siphon Split) If $F$ satisfies Locality, Node Redirect, Target Proxy, Edge Swap, Edge Multiplication and Atom $1$-$1$, then for every graph $G = (V, E) \in \mathcal{G}_t$, and node $s \in V - \{t\}$, such that $\Gamma^-_s(G) = \{\}$ and $\Gamma^+_s(G) = \{(s, v_1), ..., (s, v_k)\}$, let:
\begin{itemize}
    \item $V' = V - \{s\} \cup \{s_1, ..., s_k\}$,
    \item $E' = E - \{(s, v_1), ..., (s, v_k)\} + \{(s_1, v_1), ..., (s_k, v_k)\}$,
    \item $b' = b - b(s) \cdot \mathds{1}^s + \frac{b(s)}{k} \cdot \mathds{1}^{s_1} + ... + \frac{b(s)}{k} \cdot \mathds{1}^{s_k}$
\end{itemize} 
then for every node $u \in V - \{s\}$:
    \[F^t_u((V', E'), b') = F^t_u(G, b).\]
\end{lemma}

\begin{proof}
First, we split $s$ into out-twins $s_1, s_2, ..., s_{k}$, each with weight $\frac{b(s)}{k}$.
Formally, consider graph $G_1 = (V', E - \{(s, v_1), ..., (s, v_k)\} + \{(s_1, v_1), ..., (s_1, v_k)\} + ... + \{(s_k, v_1), ..., (s_k, v_k)\})$.
From Node Redirect we know that for every $u \in V - \{s\}$:
    \[F^t_u(G, b) = F^t_u(G_1, b').\]
Moreover, from Lemma~\ref{lemma:pr:siphon} (Siphon) we know that for every $i \in 1,...,k$:
    \[F^t_{s_i}(G_1, b') = \frac{b(s)}{k}.\]
In particular for every $i, j \in 1,...,k$, $i \neq j$ we know that:
    \[F^t_{s_i}(G_1, b') = F^t_{s_j}(G_1, b')\]
Hence, for every $i, j \in 1,...,k$, $i \neq j$ we swap edges $(s_i, v_j)$ and $(s_j, v_i)$, obtaining graph $G_2 = (V', E - \{(s, v_1), ..., (s, v_k)\} + \{(s_1, v_1), ..., (s_1, v_1)\} + ... + \{(s_k, v_k), ..., (s_k, v_k)\})$.
From Edge Swap we know that for every $u \in V - \{s\}$:
    \[F^t_u(G_1, b') = F^t_u(G_2, b').\]
Finally we divide edges outgoing of every $s_i$ by $k$.
From Edge Multiplication we conclude the proof.
\end{proof}

\begin{lemma}\label{lemma:pr:long_arrow}
(Long-Arrow Graph) If $F$ satisfies Locality, Node Redirect, Target Proxy, Edge Swap, Edge Multiplication and Atom $1$-$1$, then for every $x \in \mathds{R}_{\geq 0}$, graph $G = (\{s, v, t\}, \{(s, v), (v, t)\})$, weights $b = x \cdot \mathds{1}^s$ and node $w \in \{s, v, t\}$:
    \[F^t_w(G, x \cdot \mathds{1}^s) = x.\]
\end{lemma}

\begin{proof}
First we will calculate centralities of nodes $s$ and $v$.
Consider a graph obtained from $G$ by merging $t$ into $v$: $G_1 = (\{s, v\}, \{(s, v), (v, v)\})$.
From Target Proxy we know that:
\begin{align}
    F^t_s(G, x \cdot \mathds{1}^s) & = F^v_s(G_1, x \cdot \mathds{1}^s), \label{eq:pr:long_arrow:s}\\
    F^t_v(G, x \cdot \mathds{1}^s) & = F^v_v(G_1, x \cdot \mathds{1}^s). \label{eq:pr:long_arrow:v}
\end{align}
From Lemma~\ref{lemma:pr:almost_one_arrow} (Almost 1-Arrow Graph) combined with~\eqref{eq:pr:long_arrow:s} we know that:
\begin{equation}\label{eq:pr:long_arrow:s_x}
    F^v_s(G, x \cdot \mathds{1}^s) = x.
\end{equation}
Combined with~\eqref{eq:pr:long_arrow:v} we know that:
\begin{equation}\label{eq:pr:long_arrow:v_x}
    F^v_v(G, x \cdot \mathds{1}^s) = x.
\end{equation}

Consider graph obtained from $G$ by adding $t$ self-loop and a new target node $r$: $G_2 = (\{s, v, t, r\}, \{(s, v), (v, t), (t, r)\})$. 
From Lemma~\ref{lemma:all:t_loop} (Target Self-Loop) and Target Proxy we know that:
\begin{align*}
    F^t_t(G, x \cdot \mathds{1}^s) & = F^r_t(G_2, x \cdot \mathds{1}^s)\\
    F^t_v(G, x \cdot \mathds{1}^s) & = F^r_v(G_2, x \cdot \mathds{1}^s)
\end{align*}
Now, adding a new node $v'$ with an edge to $r$ and with weight $x$ we obtain graph $G_3 = (\{s, v, v', t, r\}, \{(s, v), (v, t), (v', r), (t, r)\})$, $b_3 = x \cdot \mathds{1}^s + x \cdot \mathds{1}^{v'}$.
From Lemma~\ref{lemma:pr:k_arrow_adjoin} (k-Arrow Adjoin) we know that:
\begin{equation*}
    F^r_t(G_2, x \cdot \mathds{1}^s) = F^r_t(G_3, b_3)
\end{equation*}
and that nodes $v$ and $v'$ have equal centralities $x$ in $(G_3, b_3)$.
Consider a graph obtained by swapping edges $(v, t), (v', r)$: $G_4 = (\{s, v, v', t, r\}, \{(s, v), (v, r), (v', t), (t, r)\})$.
From Edge Swap we know that:
\begin{equation*}
    F^r_t(G_3, b_3) = F^r_t(G_4, b_3).
\end{equation*}
Now, consider graph $G_5 = (\{v', t, r\}, \{(v', t), (t, r)\})$.
From Locality we know that:
\begin{equation*}
    F^r_t(G_4, b_3) = F^r_t(G_5, x \cdot \mathds{1}^{v'}).
\end{equation*}

Node $t$ in graph $G_5$ is the middle node of Long-Arrow Graph. From~\eqref{eq:pr:long_arrow:v_x} we know that:
\begin{equation*}
    F^r_t(G_5, x \cdot \mathds{1}^{v'}) = x.
\end{equation*}
Combining above equations we obtained that:
\begin{equation}
    F^t_t(G, x \cdot \mathds{1}^s) = x.
\end{equation}
\end{proof}

\begin{lemma}\label{lemma:pr:dag}
(DAG) If $F$ satisfies Locality, Node Redirect, Target Proxy, Edge Swap, Edge Multiplication and Atom $1$-$1$, then for every acyclic $G \in \mathcal{G}_t$, weights $b$ and node $v$ centrality $F^t_v(G, b)$ is determined.
\end{lemma}

\begin{proof}
Let $G = (V, E)$. We will do an induction on $|E|$.
\begin{itemize}
    \item If $|E| = 0$, then the thesis follows from Lemma~\ref{lemma:pr:siphon} (Siphon).
    \item If $|E| = 1$, then $G = (\{s, t\}, \{(s, t)\})$, $b = x \cdot \mathds{1}^s + y \cdot \mathds{1}^t$.
    From the Locality we know that $F^t_s(G, b) = F^t_s(G, x \cdot \mathds{1}^s)$ and $F^t_t(G, b) = F^t_t(G, x \cdot \mathds{1}^s) + F^t_t((\{t\}, \{\}), y \cdot \mathds{1}^t)$.
    Centralities in $(G, x \cdot \mathds{1}^s)$ are determined from Lemma~\ref{lemma:pr:one_arrow} (1-Arrow Graph) and centrality of $t$ in $((\{t\}, \{\}), y \cdot \mathds{1}^t)$ is determined from Lemma~\ref{lemma:pr:siphon} (Siphon).
    \item If $|E| = 2$ and one edge does not end in $t$, then $G = (\{s, v, t\}, \{(s, v), (v, t)\})$, $b = x \cdot \mathds{1}^s + y \cdot \mathds{1}^v + z \cdot \mathds{1}^t$.
    We split $v$ into $v$ with incoming edge and zero weight and $v'$ with no incoming edges and weight $y$.
    We also separate graphs $G_1 = (\{s, v, t\}, \{(s, v), (v, t)\})$, $b_1 = x \cdot \mathds{1}^s$ and  $G_2 = (\{v', t\}, \{(v', t)\})$, $b_2 = y \cdot \mathds{1}^v + z \cdot \mathds{1}^t$.
    From Node Redirect and Locality we know that $F^t_s(G, b) = F^t_s(G_1, b_1)$, $F^t_v(G, b) = F^t_v(G_1, b_1) + F^t_v(G_2, b_2)$ and $F^t_t(G, b) = F^t_t(G_1, b_1) + F^t_t(G_2, b_2)$.
    Centrality values in $(G_1, b_1)$ are known from Lemma~\ref{lemma:pr:long_arrow} (Long-Arrow) and in $(G_2, b_2)$ are known from the previous case. 
    \item If $|E| \geq 2$ and all edges end in $t$, then $G = (\{s_1, ..., s_{n-1}, t\}, \{k_1 \cdot (s_1, t), ..., k_{n-1} \cdot (s_{n-1}, t)\})$, for $n = |V|$.
    We make each of the edges unique and from Edge Multiplication we know for $G' = (\{s_1, ..., s_{n-1}, t\}, \{(s_1, t), ..., (s_{n-1}, t)\})$ and every $i \in 1, ..., n-1$ that $F^t_{s_i}(G, b) = F^t_{s_i}(G', b)$ and $F^t_{t}(G, b) = F^t_{t}(G', b)$.
    We separate each of the nodes $s_1, ..., s_{n-1}$ into an independent graph and from Locality for $G_i = (\{s_i, t\}, \{(s_i, t)\})$, $b_i = b(s) \cdot \mathds{1}^{s_i} + \frac{b(t)}{n-1} \cdot \mathds{1}^{t}$ we have $F^t_{s_i}(G', b) = F^t_{s_i}(G_i, b_i)$ and $F^t_{t}(G', b) = \sum_{i \in 1, ..., n-1} F^t_{t}(G_i, b_i)$.
    Centralities in $(G_i, b_i)$ are determined from the second case.
\end{itemize}

Let us discuss the inductive step.
Let $s_1, s_2, ... s_k$ be all the nodes without incoming edges.
For every $i \in 1, ..., k$ from Lemma~\ref{lemma:pr:siphon} (Siphon) we know that:
\begin{equation} \label{eq:pr:dag:siphon}
    F^t_{s_i}(G, b) = b(s_i).
\end{equation}
Each $s_i$ with outgoing edges $(s_i, v_{i,1}), ..., (s_i, v_{i,l_i})$ we split it into a series of $l_i$ nodes $V_1^{i} =  \{s_{i, 1}, ..., s_{i, l_1}\}$ with only one outgoing edge each, $E_1^{i} = \{(s_{i, 1}, v_{i,1}), ..., (s_{i, l_1}, v_{i, l_1})\}$, and equal weights, $b_1^i = \frac{b(s_i)}{l_i} \cdot \mathds{1}^{s_{i, 1}} + ... + \frac{b(s_i)}{l_i} \cdot \mathds{1}^{s_{i, l_i}}$.
From Lemma~\ref{lemma:pr:siphon_split} (Siphon Split) for graph $G_1 = (V_1, E_1) = (V - \{s_1, ..., s_k\} \cup V_1^{1} \cup ... \cup V_1^{k}, E - \Gamma^+_{s_1} -... - \Gamma^+_{s_k} + E_1^{1} + ... + E_1^{k}\})$, weights $b_1 = b - b(s_1) \cdot \mathds{1}^{s_1} -... - b(s_k) \cdot \mathds{1}^{s_k} + b_1^1 + ... + b_1^k$ and for every node $u \in V - \{s_1, ..., s_k\}$ we have:
\begin{align} 
    F^t_{u}(G, b) & = F^t_{u}(G_1, b_1),\label{eq:pr:dag:siphon_split}\\
    |E| & = |E_1|.\label{eq:pr:dag:siphon_split_E}
\end{align}
In the obtained graph $G_1$ every node without incoming edges has exactly one outgoing edge.

Let $v \neq t$ be the topologically greatest node that has some incoming edges $\{(s_1, v), ..., (s_k, v)\}$.
Because it is topologically greatest, all its predecessors have no incoming edges.
Hence, each of them has only one outgoing edge, to $v$.
We now merge them into one node, $s$, and from Node Redirect for $G_2 = (V_2, E_2) = (V_1 - \{s_1, ..., s_k\} \cup \{s\}, E_1 - \{(s_1, v), ..., (s_k, v)\} + \{(s, v)\})$, $b_2 = b_1 - b(s_1) \cdot \mathds{1}^{s_1} -... - b(s_k) \cdot \mathds{1}^{s_k} + (b(s_1) + ... + b(s_k)) \cdot \mathds{1}^{s}$ and every node $u \in V_1 - \{s_1, ..., s_k\}$ we know that:
\begin{align} 
    F^t_{u}(G_1, b_1) & = F^t_{u}(G_2, b_2),\label{eq:pr:dag:NR}\\
    |E_1| & \geq |E_2|.\label{eq:pr:dag:NR_E}
\end{align}

Now let us add new node $v'$, obtaining graph $G_3 = (V_3, E_3) = (V_2 \cup \{v'\}, E_2 + \{\Gamma^+_v(G_2) \cdot (v', t)\})$, $b_3 = b_2 + F^t_{v}(G_2, b_2) \cdot \mathds{1}^{v'}$ (for the yet unknown $F^t_{v}(G_2, b_2)$).
From Lemma~\ref{lemma:pr:k_arrow_adjoin} (k-Arrow Adjoin) for every node $u \in V_2 - \{t\}$ we have $F^t_{u}(G_2, b_2) = F^t_{u}(G_3, b_3)$,  and $F^t_{t}(G_2, b_2)  = F^t_{t}(G_3, b_3) - F^t_{v}(G_2, b_2)$ is determined if $F^t_{v}(G_2, b_2)$ is determined.
And $F^t_{v'}(G_3, b_3) = F^t_{v}(G_2, b_2)$, that is nodes $v$ and $v'$ have equal centralities in $(G_3, b_3)$.
As for the number of edges, $|E_2| + \Gamma^+_v(G_2) = |E_3|$.

We exchange all the edges outgoing from $v$ with all the edges outgoing from $v'$, obtaining graph $G_4 = (V_3, E_4) = (V_3, E_3 - \Gamma^+_{v}(G_3) - \Gamma^+_{v'}(G_3) + \{(v', w) : (v, w) \in E_3\} +  \{\Gamma^+_v(G_3) \cdot (v, t)\})$.
From Edge Swap we know that for every node $u \in V_3$ we have 
$F^t_{u}(G_3, b_3) = F^t_{u}(G_4, b_3)$. And $|E_3| = |E_4|$.

Multiple edges from $v$ to $t$ are unified, obtaining graph $G_5 = (V_3, E_5) = (V_3, E_4 - \{(\Gamma^+_v(G_3) - 1) \cdot (v, t)\})$.
From Edge Multiplication for every node $u \in V_3$ we know that $F^t_{u}(G_4, b_3) = F^t_{u}(G_5, b_3)$. The number of edges decreased by $\Gamma^+_v(G_3) - 1$, that is $|E_4| - (\Gamma^+_v(G_3) - 1) = |E_5|$.

In total, we obtained for every node $u \in V_2 - \{t\}$ that:
\begin{align} 
    & F^t_{u}(G_2, b_2) = F^t_{u}(G_5, b_3),\label{eq:pr:dag:ES_u}\\
    & F^t_{t}(G_2, b_2) = F^t_{t}(G_5, b_3) - F^t_{v}(G_5, b_3),\label{eq:pr:dag:ES_t}\\
    & |E_2| + 1 = |E_5|.\label{eq:pr:dag:ES_E}
\end{align}

Finally, we separate graphs $G_6 = (\{s, v, t\}, \{(s, v), (v, t)\})$, $b_6 = b_3(s) \cdot \mathds{1}^{s} + b_3(v) \cdot \mathds{1}^{v}$ and $G_7 = (V_5 - \{s, v\}, E_5 - \{(s, v), (v, t\})$, $b_7 = b_3 - b_3(s) \cdot \mathds{1}^{s} - b_3(v) \cdot \mathds{1}^{v}$.
From Locality for every node $u \in V_5 - \{s, v, t\}$ we know that:
\begin{align} 
    & F^t_{u}(G_5, b_3) = F^t_{u}(G_7, b_7),\label{eq:pr:dag:LOC_u}\\
    & F^t_{v}(G_5, b_3) = F^t_{v}(G_6, b_6),\label{eq:pr:dag:LOC_v}\\
    & F^t_{t}(G_5, b_3) = F^t_{v}(G_6, b_6) + F^t_{u}(G_7, b_7).\label{eq:pr:dag:LOC_t}\\
    & |E_5| - 2 = |E_7|.\label{eq:pr:dag:LOC_E}
\end{align} 
The centralities in $(G_6, b_6)$ are known from the case $E = 2$.
This combined with \eqref{eq:pr:dag:ES_u} for $v$ and~\eqref{eq:pr:dag:LOC_v} causes the weight $b_7(v') = b_3(v') = F^t_{v}(G_2, b_2)$ to stop being unknown.
Now centralities in $(G_7, b_7)$ are also determined, from the inductive assumption, because from \eqref{eq:pr:dag:siphon_split_E}, \eqref{eq:pr:dag:NR_E}, \eqref{eq:pr:dag:ES_E} and~\eqref{eq:pr:dag:LOC_E} we have $|E_7| = |E_2| - 1 \leq |E| - 1$.

This combined with \eqref{eq:pr:dag:siphon_split}, \eqref{eq:pr:dag:NR}, \eqref{eq:pr:dag:ES_u} and~\eqref{eq:pr:dag:LOC_u} gives for $u \in V - \{v, t, s_1, ..., s_k\}$ that $F^t_{u}(G, b)$ is determined.
Combined with \eqref{eq:pr:dag:siphon_split}, \eqref{eq:pr:dag:NR}, \eqref{eq:pr:dag:ES_u} and~\eqref{eq:pr:dag:LOC_v} gives that $F^t_{v}(G, b)$ is determined.
And combined with \eqref{eq:pr:dag:siphon_split}, \eqref{eq:pr:dag:NR}, \eqref{eq:pr:dag:ES_t}, \eqref{eq:pr:dag:LOC_v} and~\eqref{eq:pr:dag:LOC_t} gives that $F^t_{t}(G, b)$ is determined.
Together with~\eqref{eq:pr:dag:siphon}, we obtained that centrality of every node in acyclic $(G, b)$ is determined.
\end{proof}

\begin{lemma}\label{lemma:pr:no_target_outlet}
\label{No Target Outlet}
(No Target Outlet) If $F$ satisfies Locality, Node Redirect, Target Proxy, Edge Swap, Edge Multiplication and Atom $1$-$1$, then for every graph $G = (V, E) \in \mathcal{G}_t$ such that $(t, v) \in E$ and node $u$:
    \[F^{t}_u((V, E - \{(t, v )\}), b) = F^{t}_u(G, b).\] 
\end{lemma}

\begin{proof}
If $v = t$, then the thesis follows from Lemma~\ref{lemma:all:t_loop} (Target Self-Loop).
Let us now assume $v \neq t$.

Let us split $v$ into out-twins $v'$ with zero weight and an incoming edge from $t$ and $v$ with the original weight and the rest of incoming edges, obtaining $G_1 = (V_1, E_1) = (V \cup \{v'\}, E - \{(t, v)\} + \{(t,  v')\} + \{(v', w) : (v, w) \in E\})$.
From Node Redirect for every $u \in V - \{v\}$ we know that:
\begin{align} 
    F^{t}_u(G, b) & = F^{t}_u(G_1, b),\label{eq:pr:no_target_outlet:NR_u}\\
    F^{t}_v(G, b) & = F^{t}_v(G_1, b) + F^{t}_{v'}(G_1, b).\label{eq:pr:no_target_outlet:NR_v}
\end{align}
Let $x = F^{t}_{v'}(G_1, b)$.

We add node $v''$ that has the same out-degree $k = \Gamma^+_{v'}(G_1)$ as $v'$ and weight $x$, obtaining graph $G_2 = (V_2, E_2) = (V_1 \cup \{v''\}, E + k \cdot \{(v'', t)\})$, $b_2 = b + x \cdot \mathds{1}^{v''}$.
From Lemma~\ref{lemma:pr:k_arrow_adjoin} (k-Arrow Adjoin) for every $u \in V_1 - \{t\}$ we know that $F^{t}_u(G_1, b) = F^{t}_u(G_2, b_2)$, $F^{t}_t(G_1, b) = F^{t}_t(G_2, b_2) - x$.
We also know that $F^{t}_{v''}(G_2, b_2) = x = F^{t}_{v'}(G_2, b_2)$, that is nodes $v'$ and $v''$ have the same centralities in $(G_2, b_2)$. 
Now we swap edges outgoing of $v'$ and of $v''$, obtaining graph $G_3 = (V_2, E_3) = (V_2, E_2 - \Gamma^+_{v'}(G_2) + \{(v'', w) : (v', w) \in E_2\} - \Gamma^+_{v''}(G_2) + k \cdot \{(v', t)\})$.
From Edge Swap for every $u \in V_2$ we know that $F^{t}_u(G_2, b_2) = F^{t}_u(G_3, b_2)$.
In total, for every $u \in V_1 - \{t\}$ we know that:
\begin{align} 
    & F^{t}_u(G_1, b) = F^{t}_u(G_3, b_2),\label{eq:pr:no_target_outlet:ES_u}\\
    & F^{t}_t(G_1, b) = F^{t}_t(G_3, b_2) - x,\label{eq:pr:no_target_outlet:ES_t}\\
    & F^{t}_{v''}(G_3, b_2) = x.\label{eq:pr:no_target_outlet:ES_v''}
\end{align}

Let us now delete node $v'$, obtaining graph $G_4 = (V_4, E_4) = (V_2 - \{v'\}, E_3 - \{(t,  v')\} - k \cdot \{(v', t)\})$.
From Lemma~\ref{lemma:pr:k_arrow_adjoin} (k-Arrow Adjoin) for $u \in V_4 - \{t\}$ we know that:
\begin{align} 
    F^{t}_u(G_3, b_2) & = F^{t}_u(G_4, b_2),\label{eq:pr:no_target_outlet:k_Arrow_u}\\
    F^{t}_t(G_3, b_2) & = F^{t}_t(G_4, b_2) + b_2(v') = F^{t}_t(G_4, b_2),\label{eq:pr:no_target_outlet:k_Arrow_t}\\
    F^{t}_{v'}(G_3, b_2) & = b_2(v') = 0.\label{eq:pr:no_target_outlet:k_Arrow_v'}
\end{align}
But now from \eqref{eq:pr:no_target_outlet:ES_u} for $v'$ and~\eqref{eq:pr:no_target_outlet:k_Arrow_v'} we know that $x = F^{t}_{v'}(G_3, b_2) = 0$.

Let us use this fact in previous equations.
From \eqref{eq:pr:no_target_outlet:NR_v},~\eqref{eq:pr:no_target_outlet:ES_u} for $v$ and~\eqref{eq:pr:no_target_outlet:k_Arrow_u} for $v$ we have $F^{t}_v(G, b) = F^{t}_v(G_4, b_2)$.
As for $t$, from \eqref{eq:pr:no_target_outlet:NR_u},~\eqref{eq:pr:no_target_outlet:ES_t} and~\eqref{eq:pr:no_target_outlet:k_Arrow_t} we obtained that $F^{t}_t(G, b) = F^{t}_t(G_4, b_2)$.
Moreover, from \eqref{eq:pr:no_target_outlet:NR_u},~\eqref{eq:pr:no_target_outlet:ES_u} and~\eqref{eq:pr:no_target_outlet:k_Arrow_u} for $u \in V - \{v, t\}$ we know that $F^{t}_u(G, b) = F^{t}_u(G_4, b_2)$.
Hence, for every node $u \in V$ we obtained that:
\begin{equation} \label{eq:pr:no_target_outlet:x_0_u}
    F^{t}_u(G, b) = F^{t}_u(G_4, b_2).
\end{equation}

Finally we merge $v''$ into $v$, obtaining $G_5 = (V, E - \{(t, v)\})$.
From \eqref{eq:pr:no_target_outlet:ES_v''} and~\eqref{eq:pr:no_target_outlet:k_Arrow_u} for $v''$ we know that $F^{t}_{v''}(G_4, b_2) = 0$.
Combining this with Node Redirect and~\eqref{eq:pr:no_target_outlet:x_0_u} we conclude the proof.
\end{proof}

\begin{lemma}\label{lemma:pr:determined}
If F satisfies Locality, Node Redirect, Target Proxy, Edge Swap, Edge Multiplication and Atom $1$-$1$ then it is t-Random Walk Betweenness Centrality.
\end{lemma}

\begin{proof}
Fix graph $G = (V, E) \in \mathcal{G}_t$ and weights $b$.
We will prove the thesis by induction on the number of cycles in $G$.
If there are no cycles, then the thesis follows from the fact that $F^t$ is uniquely determined on acyclic graphs (Lemma~\ref{lemma:pr:dag} (DAG)) and the fact that Random Walk Betweenness satisfies the axioms (Lemma~\ref{lemma:rwb_to_axioms}). 

Assume otherwise.
Let us also assume there are no edges from $t$, as they could be deleted and from Lemma~\ref{lemma:pr:no_target_outlet} (No Target Outlet) centralities would not change.
Fix a node $v$, that belongs to at least one cycle (note that $v \neq t$, because $t$ has no outgoing edges) and let $k = |\Gamma^+_{v}(G)|$.
Let $x = RWB^{t}_v(G, b)$.

Let us for now consider another graph $(G', b')$, where $G' = (V', E')$ and:
    \[V' = V \cup \{v'\},\]
    \[E' = E - \Gamma^+_{v}(G) + \{(v', u) : (v, u) \in E\} + k \cdot \{(v, t)\},\]
    \[b' = b + x \cdot \mathds{1}^{v'}.\]

Let us note, that $G' \in \mathcal{G}_t$, that is $t$ is reachable from every node.
If in $G$ from a node $u \in V$ there was a path to $t$ without node $v$, then all of its edges are still present in $G'$.
If every path from $u$ to $t$ passed through $v$, then the path up to this point is present in $G'$, and then $v$ has direct edges to $t$.
Finally, $v'$ has edges to some other nodes, from which we already know $t$ is reachable.

All the cycles that are present in $G'$ are present also in $G$.
Note that in $G'$ there are no cycles involving $v'$, because it has no incoming edges.
There are also no cycles involving $v$, because it has only outgoing edges to $t$, and $t$ has no outgoing edges itself.
The cycles on nodes $V' - \{v, v'\}$ have all their edges present also in $G$.

In turn, cycles through $v$ in $G$ are not present in $G_2$.
This means that the number of cycles in $G'$ is strictly smaller than in $G$.

From the inductive assumption we know that in $G'$ centrality $F^t$ is equal t-Random Walk Betweenness, in particular we have for every $u \in V'$ that:
\begin{equation} \label{eq:pr:determined:induction_u}
    F^{t}_u(G', b') = RWB^{t}_u(G', b').
\end{equation}

Now we notice that $G'$ is obtained from $G$ by firstly adding a new node $v'$ with out-degree $k$ and weight $x$, obtaining intermediary $G_1 = (V_1, E_1) = (V \cup \{v'\}, E + k \cdot \{(v', t)\})$ with weights $b'$ and then by swapping edges outgoing from $v$ and $v'$, obtaining $(G', b')$.
Because $RWB^{t}$ satisfies our axioms, from Lemma~\ref{lemma:pr:k_arrow_adjoin} (k-Arrow Adjoin) we know that $RWB^{t}_v(G, b) = RWB^{t}_v(G_1, b_1)$ and $RWB^{t}_{v'}(G_1, b_1) = x$, that is $v$ and $v'$ have equal $RWB^{ t}$ in $(G_1, b_1)$.
From Edge Swap for every $u \in V_1$ we know that $RWB^{t}_u(G_1, b_1) = RWB^{t}_u(G', b')$.
In total:
\begin{align} 
    & RWB^{t}_{v}(G', b') = RWB^{t}_v(G, b) = x,\label{eq:pr:determined:ES_v}\\
    & RWB^{t}_{v'}(G', b') = x,\label{eq:pr:determined:ES_v'}
\end{align}

From \eqref{eq:pr:determined:induction_u} for $v$ and~\eqref{eq:pr:determined:ES_v} we obtained that $F^{t}_{v}(G', b') = x$.
But from \eqref{eq:pr:determined:induction_u} for $v'$ and~\eqref{eq:pr:determined:ES_v'} we obtained that $F^{t}_{v'}(G', b') = x$, that is $v$ and $v'$ have equal $F^t$ centralities in $(G', b')$.

Now we revert our changes, that is swap edges outgoing from $v$ and $v'$ and delete $v'$.
From Edge Swap for every $u \in V_1$ we know that $F^t_u(G', b') = F^t_u(G_1, b_1)$.
From Lemma~\ref{lemma:pr:k_arrow_adjoin} (k-Arrow Adjoin) for every $u \in V - \{t\}$ we know that $F^t_u(G_1, b_1) = F^t_u(G, b)$ and $F^t_t(G_1, b_1) - x = F^t_t(G, b)$.
In total, for every $u \in V - \{t\}$:
\begin{align} 
    & F^t_u(G, b) = F^t_u(G', b'),\label{eq:pr:determined:revert_k_Arrow_u}\\
    & F^t_t(G, b) = F^t_t(G', b') \!-\! x = F^t_t(G', b') \!-\! F^{t}_{v}(G', b'). \label{eq:pr:determined:revert_k_Arrow_t}
\end{align}

Finally, \eqref{eq:pr:determined:revert_k_Arrow_u} and~\eqref{eq:pr:determined:revert_k_Arrow_t} together with the inductive assumption determine values of $F^t$ on $(G, b)$.
We also know that Random Walk Betweenness satisfies the axioms (Lemma~\ref{lemma:rwb_to_axioms}), which concludes the proof.
\end{proof}
\end{multicols}

\end{document}